%% file: asap.tex
\pdfoutput=1
\documentclass[a4paper,USenglish,cleveref,thm-restate]{lipics-v2021}

\title{As Soon as Possible but Rationally}

\author{Véronique Bruyère}{Université de Mons (UMONS), Belgium \and \url{https://informatique-umons.be/bruyere-veronique/}}{veronique.bruyere@umons.ac.be}{https://orcid.org/0000-0002-9680-9140}{}
\author{Christophe Grandmont}{Université de Mons (UMONS), Belgium \and Université libre de Bruxelles (ULB), Belgium \and \url{https://chrisgdt.github.io/}}{christophe.grandmont@umons.ac.be}{https://orcid.org/0009-0009-4573-0123}{}
\author{Jean-François Raskin}{Université Libre de Bruxelles (ULB), Belgium \and \url{https://verif.ulb.ac.be/jfr/}}{jean-francois.raskin@ulb.be}{https://orcid.org/0000-0002-3673-1097}{Supported by Fondation ULB (\url{https://www.fondationulb.be/en/})}

\authorrunning{V.\ Bruyère, C.\ Grandmont, and J.-F.\ Raskin}



\funding{This work has been supported by the Fonds de la Recherche Scientifique – FNRS under Grant n° T.0023.22 (PDR Rational)}

\nolinenumbers
\hideLIPIcs

\usepackage{stmaryrd}
\usepackage{tikz}
\usepackage{graphics}

\newtheorem{problem}[theorem]{Problem}

\crefname{enumi}{Step}{Step}

\input{macros}

\input{tikz}

\begin{document}

\maketitle

\begin{abstract}
This paper addresses complexity problems in rational verification and synthesis for multi-player games played on weighted graphs, where the objective of each player is to minimize the cost of reaching a specific set of target vertices. In these games, one player, referred to as the system, declares his strategy upfront. The other players, composing the environment, then rationally make their moves according to their objectives. The rational behavior of these responding players is captured through two models: they opt for strategies that either represent a Nash equilibrium or lead to a play with a Pareto-optimal cost tuple.
\end{abstract}

\section{Introduction}

Nowadays, formal methods play a crucial role in ensuring the reliability of critical computer systems. Still, the application of formal methods on a large scale remains elusive in certain areas, notably in multi-agent systems. Those systems pose a significant challenge for formal verification and automatic synthesis because of their heterogeneous nature, encompassing everything from conventional reactive code segments to fully autonomous robots and even human operators. Crafting formal models that accurately represent the varied components within these systems is often a too complex task.

Although constructing detailed operational models for humans or sophisticated autonomous robots might be problematic, it is often more feasible to identify the \emph{overarching goals} that those agents pursue. Incorporating these goals is crucial in the design and validation process of systems that interact with such entities. Typically, a system is not expected to function flawlessly under all conditions but rather in scenarios where the agents it interacts with act in alignment with their objectives, i.e., they \emph{behave rationally}. \emph{Rational synthesis} focuses on creating a system that meets its specifications against any behavior of environmental agents that is guided by their goals (and not against any of their behaviors).
\emph{Rational verification} studies the problem of ensuring that a system enforces certain correctness properties, not in all conceivable 
scenarios, but specifically in scenarios where environmental agents behave in accordance with their objectives.

Rationality can be modeled in various ways. In this paper, we focus on two general approaches. The first approach comes from game theory where rationality is modeled by the concept of equilibrium, such as \emph{Nash equilibria} (NE)~\cite{Nash50} or \emph{subgame perfect equilibria} (SPE), a refinement of NEs~\cite{osbornebook}. The second approach treats the environment as a single agent but with multiple, sometimes conflicting, goals, aiming for behaviors that achieve a \emph{Pareto-optimal} balance among these objectives. The concept of Pareto-optimality (PO) and its application in multi-objective analysis have been explored primarily in the field of optimization~\cite{DBLP:conf/focs/PapadimitriouY00}, but also in formal methods~\cite{DBLP:conf/fossacs/AlurDMW09,DBLP:conf/cav/BrenguierR15}.
These two notions of rationality are different in nature: in the first setting, each component of the environment playing an equilibrium is considered to be an independent selfish individual, excluding cooperation scenarios, while in the second setting, several components of the environment can collaborate and agree on trade-offs. The challenge lies in adapting these concepts to \emph{reactive systems} characterized by ongoing, non-terminating interactions with their environment. This necessitates the transition from two-player zero-sum games on graphs, the classical approach used to model a fully adversarial environment (see e.g.~\cite{PnueliR89}), to the more complex setting of \emph{multi-player non zero-sum games on graphs} used to model environments composed of various rational agents. 

Rational synthesis and rational verification have attracted large attention recently, see e.g.~\cite{DBLP:conf/mfcs/BriceRB23, Pareto-Rational-Verification, ConduracheFGR16, FismanKL10-crsp-kupferman, DBLP:journals/ai/GutierrezNPW20, DBLP:journals/amai/GutierrezNPW23, KupfermanPV14-ncrsp, tacas2022}. But the results obtained so far, with a few exceptions that we detail below, are limited to the \emph{qualitative} setting formalized as Boolean outcomes associated with $\omega$-regular objectives. Those objectives are either specified using linear temporal specifications or automata over infinite words (like parity automata). The complexity of those problems is now well understood (with only a few complexity gaps remaining, see e.g.~\cite{ConduracheFGR16,tacas2022}). The methods to solve those problems and get completeness results for worst-case complexity are either based on automata theory (using mainly automata over infinite trees) or by reduction to zero-sum games. \emph{Quantitative} objectives are less studied and revealed to be much more challenging. For instance, it is only very recently that the rational verification problem was studied, for SPEs in non zero-sum games with mean-payoff, energy, and discounted-sum objectives in~\cite{DBLP:conf/mfcs/BriceRB23}, for an LTL specification in multi-agent systems that behave according to an NE with mean-payoff objectives in~\cite{DBLP:journals/amai/GutierrezNPW23} or with quantitative probabilistic LTL objectives in~\cite{DBLP:conf/atal/Hyland0KW24}. In~\cite{AlmagorKP18}, the rational synthesis problem was studied for the quantitative extension LTL[$\mathcal F$] of LTL where the Boolean operators are replaced with arbitrary functions mapping binary tuples into the interval $[0, 1]$.

In this paper, we consider \emph{quantitative reachability} objectives. Our choice for studying these objectives was guided by their fundamental nature and also by their relative simplicity. Nevertheless, as we will see, they are challenging for both rational synthesis and rational verification. Indeed, to obtain worst-case optimal algorithms and establish completeness results, 
we had to resort to the use of \emph{innovative} theoretical tools, more advanced than those necessary for the qualitative framework. In our endeavor, we have established the exact complexity of most studied decision problems in rational synthesis and rational verification.

\begin{table}
\resizebox{\columnwidth}{!}{%
\begin{threeparttable}[b]
\centering
\caption{Summary of complexity results.}

\begin{tabular}{l|l|l|l|l}
                             & \textbf{Non-coop. verif.} & \textbf{Universal non-coop. verif.} & \textbf{Coop. synthesis}  & \textbf{Non-coop. synthesis} \\ \hline
\textbf{PO, weights}     & \piComplete{}       & \pspaceComplete{}  & \pspaceComplete{} & Unknown, \nexptimeHard{}\tnote{1}~~\cite{gaspard-pareto-quanti-reach}                              \\
\textbf{PO, qualitative} & \piComplete{}       & \pspaceComplete{}  & \pspaceComplete{} & \nexptimeComplete{}~\cite{Stackelberg-Pareto-Synthesis}   \\

\textbf{NE, weights}     & \conpComplete{}     & \conpComplete{}    & \npComplete{}     & Unknown, \exptimeHard{}\tnote{2} \\
\textbf{NE, qualitative}   &  \conpComplete{}~\cite{grandmont23} & \conpComplete{}~\cite{grandmont23}   & \npComplete{}~\cite{ConduracheFGR16}          & \pspaceComplete{}~\cite{ConduracheFGR16}                               
\end{tabular}
\begin{tablenotes}
\item [1] In~\cite{gaspard-pareto-quanti-reach}, the authors show the \nexptimeComplete{}ness when the weight function is the same for each player, we do not have the result for the general case of multiple weight functions.
\item [2] For the important special case of one-player environments, we provide an algorithm that runs in \exptime{} and we can prove \pspaceHard{}ness. The \exptimeHard{}ness of the general case already holds for two-player environments.
\setcounter{footnote}{1} 
\end{tablenotes}
\end{threeparttable}}
\label{table:results}
\end{table}

\subparagraph*{Technical Contributions.}
In this work, we explore both verification and synthesis problems through the lens of rationality, defined by Pareto-optimality and Nash equilibria, for quantitative reachability objectives. For the synthesis problem, we also consider the \emph{cooperative} variant where the environment cooperates with the system: we want to decide whether the system has a strategy and the environment a rational response to this strategy such that the objective of the system is enforced. Our results are presented in \cref{table:results}, noting that all results lacking explicit references are, to our knowledge, novel contributions.
For completeness, the table includes (new and known) results for the qualitative scenario. 

The results for \emph{PO rationality} are as follows. (1) For the verification problems, we assume that the behavior of the system is formalized by a \emph{nondeterministic Mealy machine}, used to represent a (usually infinite) set of its possible implementations. For each of those implementations, we verify that the quantitative reachability objective of the system is met against any rational behavior of the environment. We establish that this problem is \pspaceComplete{}. To obtain the upper bound, we rely on a \emph{genuine combination of techniques} based on Parikh automata and a recursive \pspace{} algorithm (for positive Boolean combinations of bounded safety objectives, a problem of independent interest).
Parikh automata are used to guess a compact representation of certificates which are paths of possibly exponential length in the size of the problem input.
When the Mealy machine is deterministic, we show that the complexity goes down to \piComplete{}ness, as the previous \pspace{} algorithm is replaced by a \conp{} oracle. (2) For the synthesis problems, we only consider the cooperative version which we prove to be \pspaceComplete{}, as the non-cooperative version was established to be \nexptimeComplete{} in~\cite{gaspard-pareto-quanti-reach} when there is a common weight function for each player.

The results for \emph{NE rationality} are as follows. (1) We establish that, surprisingly, the verification problems are \conpComplete{} both for the general case of a nondeterministic Mealy machine and for the special case where it is deterministic. The upper bounds for those problems are again based on Parikh automata certificates but here there is no need to use a \conp{} oracle. (2) For the synthesis problems, the landscape is more challenging. For the cooperative case, we were able to establish that the problem is \npComplete{}. For the non-cooperative case, we have partially solved the problem and established the following results. When the environment is composed of a single rational player, the problem is in \exptime{} and \pspaceHard{}. For an environment with at least two players, we show that the problem is \exptimeHard{} but we leave its decidability open. The lower bounds are obtained using an elegant reduction from countdown games~\cite{countdown-game}. We give indications in the paper why the problem is difficult to solve and why classical automata-theoretic methods \emph{may not be sufficient} (if the problem is decidable).

In this paper, we focus on nonnegative weights as we show that considering arbitrary weights leads to undecidability of the synthesis problems. We also focus on NEs instead of SPEs, even if the latter are a better concept to model rational behavior in games played on graphs. Indeed, it is well-known that SPEs pose greater challenges than NEs. So, starting with NEs offers a better initial step for the algorithmic treatment of rational verification and synthesis in quantitative scenarios, an area that remains largely unexplored.

\subparagraph*{Related Work.} 
The survey~\cite{DBLP:conf/fsttcs/BrihayeGMR23} presents several results about different game models and different kinds of objectives related to reachability. Quantitative objectives in \emph{two-player zero-sum games} were largely studied, see e.g.~\cite{DBLP:conf/concur/BrihayeGHM15,DBLP:conf/fsttcs/ChatterjeeDHR10,ehrenfeucht1979positional}, even if exact complexity results are often elusive due to the intricate nature of the problems (e.g. the exact complexity of solving mean-payoff games is still an open problem). 
In multi-player non zero-sum games, the \emph{(constrained) existence} of equilibria is also well studied. The existence of simple NEs was established in~\cite{DBLP:conf/lfcs/BrihayePS13} for mean-payoff and discounted-sum objectives. No decision problem is considered in that paper. The constrained existence of SPEs in quantitative reachability games was proved \pspaceComplete{} in~\cite{DBLP:journals/iandc/BrihayeBGR21}. We prove here that the complexity is lower when we use NEs to model rationality, as we obtain \npComplete{}ness for the related cooperative synthesis problem. Deciding the constrained existence of SPEs was recently solved for quantitative reachability games in~\cite{DBLP:conf/concur/BrihayeBGRB19} and for mean-payoff games in~\cite{DBLP:conf/concur/BriceRB21, DBLP:conf/icalp/BriceRB22}.
The cooperative and non-cooperative rational \emph{synthesis problems} were studied in~\cite{NoncoopSynth_Meanpayoff_2players} for games with mean-payoff and discounted-sum objectives when the environment is composed of a single player. The mean-payoff case was proved to be \npComplete{} and the discounted-sum case was linked to the open target discounted sum problem, which explains the difficulty of solving the problem in this case.

\subparagraph*{Structure of the Paper.} The background is given in \cref{section:background}. The formal definitions of the studied problems and our main complexity results are stated in \cref{section:studied-problems}. The proofs of our results are given for PO rationality in \cref{section:pareto}, and for NE rationality in \cref{section:Nash}. We give a conclusion and future work in \cref{sec:conclusion}.

\section{Background}\label{section:background}

\subparagraph*{Arenas and Plays.}

A (finite) \emph{arena} $\arena{}$ is a tuple $(V,E,\Players,(V_i)_{i\in\Players})$ where $V$ is a finite set of \emph{vertices}, $E \subseteq V \times V$ is a set of \emph{edges}, $\Players$ is a finite set of \emph{players}, and $(V_i)_{i\in\Players}$ is a partition of $V$, where $V_i$ is the set of vertices \emph{owned} by player~$i$. We assume that $v \in V$ has at least one \emph{successor}, i.e., the set $\succc{v}=\{v'\in V\mid (v,v')\in E\}$ is nonempty.

We define a \emph{play} $\pi\in V^\omega$ (resp.\ a \emph{history} $h\in V^*$) as an infinite (resp.\ finite) sequence of vertices $\pi_0\pi_1\dots$ such that $(\pi_i,\pi_{i+1})\in E$ for any two consecutive vertices $\pi_i, \pi_{i+1}$. The \emph{length} $|h|$ of a history $h$ is the number of its vertices. The empty history is denoted $\varepsilon$. Given a play $\pi$ and two indexes $k < k'$, we write $\pi_{\leq k}$ the prefix $\pi_0\dots\pi_k$ of $\pi$, $\pi_{\geq k}$ the suffix $\pi_{k} \pi_{k+1} \dots$ of $\pi$, and $\pi_{[k,k'[}$ for $\pi_k\dots\pi_{k'-1}$. We denote the first vertex of $\pi$ by $\first{\pi}$. These notations are naturally adapted to histories. We also write $\last{h}$ for the last vertex of a history $h \neq \varepsilon$. The set of all plays (resp.\ histories) of an arena $\arena{}$ is denoted $\Plays_{\arena}\subseteq V^\omega$ (resp.\ $\Hist_{\arena}\subseteq V^*$), and we write $\Plays$ (resp.\ $\Hist$) when the context is clear. For $i \in \Players$, the set $\Histsigma{i}\subseteq V^*V_i$ represents all histories ending in a vertex $v\in V_i$. 
That is, $\Histsigma{i}=\{h\in\Hist\mid h\neq\varepsilon \text{ and } \last{h}\in V_i\}$. 

We can \emph{concatenate} two nonempty histories $h_1$ and $h_2$ into a single one, denoted $h_1\cdot h_2$ or $h_1h_2$ if $(\last{h_1},\first{h_2})\in E$. When a history can be concatenated to itself, we call it a \emph{cycle}. Furthermore, a play $\pi= \mu\nu\nu \dots = \mu(\nu)^\omega$ where $\mu\nu\in\Hist$ with $\nu$ a cycle, is called a \emph{lasso}. The \emph{length} of $\pi$ is then the length of $\mu\nu$.\footnote{To have a well-defined length for a lasso $\pi$, we assume that $\pi = \mu(\nu)^\omega$ with $\mu\nu$ of minimal length.} Given a play $\pi$, a \emph{cycle along $\pi$} refers to a sequence $\pi_{[m,n[}$ with $\pi_m = \pi_n$.
We denote $\occ{\pi}=\{v\in V\mid \exists n\in\N,\; v=\pi_n\}$ the set of all vertices occurring along $\pi$, and we say that $\pi$ \emph{visits} or \emph{reaches} a vertex $v\in\occ{\pi}$ or a set $T$ such that $T \cap \occ{\pi} \neq \emptyset$. The previous notions extend to histories.

Given an arena $\arena$, if we fix an \emph{initial vertex} $v_0\in V$, we say that $\arena$ is \emph{initialized} and we denote by $\Plays(v_0)$ (resp.\ $\Hist(v_0)$) all its plays (resp.\ nonempty histories) starting with $v_0$.
An arena is called \emph{weighted} if it is augmented with a non-negative \emph{weight function} $w_i: E \rightarrow \N$ for each player~$i$. We denote by $W$ the greatest weight, i.e., $W = \max\{w_i(e) \mid e \in E, i \in \Players\}$. We extend $w_i$ to any history $h = \pi_0 \dots \pi_n$ such that $w_i(h) = \sum_{j = 1}^n \, w_i((\pi_{j-1},\pi_{j}))$.

\subparagraph*{Reachability Games.}

A \emph{reachability game} is a tuple $\game{}=(\arena, (\reach{i})_{i \in \Players})$ where $\arena$ is a weighted arena and $\reach{i}\subseteq V$ is a \emph{target set} for each $i\in\Players$. We define a \emph{cost function} $\costfunc{i}:\Plays\rightarrow\N\cup\{+\infty\}$ for each player~$i$, such that for all plays $\pi=\pi_0\pi_1\dots\in\Plays$,
$\cost{i}{\pi}=w_i(\pi_{0} \dots \pi_{n})$ with $n$ the smallest index such that $\pi_n \in \reach{i}$, if it exists and $\cost{i}{\pi}=+\infty$ otherwise.

The \emph{reachability objective} of player~$i$ is to \emph{minimize} this cost as much as possible, i.e., given two plays $\pi, \pi'$ such that $\costfunc{i}(\pi) < \costfunc{i}(\pi')$, player~$i$ prefers $\pi$ to $\pi'$. We extend $<$ to tuples of costs as follows: $(\costfunc{i}(\pi))_{i \in \Players} < (\costfunc{i}(\pi'))_{i \in \Players}$
if $\costfunc{i}(\pi) \leq \costfunc{i}(\pi')$ for all $i \in \Players$ and there exists some $j \in \Players$ such that $\costfunc{j}(\pi) < \costfunc{j}(\pi')$. Given a play $\pi$, we denote by $\Visit{\pi}$ the set of players~$i$ such that $\pi$ visits $\reach{i}$, i.e., $\Visit{\pi} = \{i \in \Players \mid \cost{i}{\pi} < + \infty \}$. When for all $i \in \Players$ and $e\in E$, $w_i(e) = 0$, we speak of \emph{qualitative} reachability games, since $\costfunc{i}(\pi) = 0$ if $\occ{\pi} \cap \reach{i} \neq \emptyset$ and $+\infty$ otherwise.

\subparagraph*{Strategies and Mealy Machines.}

Let $\arena{}=(V,E,\Players,(V_i)_{i\in\Players})$ be an arena. A \emph{strategy} $\strategyfor{i}:\Histsigma{i}\rightarrow V$ for player~$i$ maps any history $h\in \Histsigma{i}$ to a vertex $v\in \succc{\last{h}}$, which is the next vertex that player~$i$ chooses to move to after reaching the last vertex in $h$. The set of all strategies of player~$i$ is denoted $\Sigma_i$.
A play
$\pi = \pi_0\pi_1 \dots$ is \emph{consistent} with $\strategyfor{i}$ if $\pi_{k+1} = \strategyfor{i}(\pi_0 \dots \pi_k)$ for all $k \in \N$ such that $\pi_k \in V_i$.
Consistency is naturally extended to histories. A tuple of strategies $\sigma = (\strategyfor{i})_{i\in\Players}$ with $\strategyfor{i} \in \Sigma_i$, is called a \emph{strategy profile}. In an arena initialized at $v_0$, we limit the domain of each strategy $\strategyfor{i}$ to $\Histsigma{i}(v_0)$; the play $\pi$ starting from $v_0$ and consistent with each $\strategyfor{i}$ is denoted $\outcomefrom{\strategyProfile}{v_0}$ and called \emph{outcome}.

Given an initialized arena $\arena$, we can encode a strategy or a set of strategies by a (finite) \emph{nondeterministic Mealy machine}~\cite{DBLP:conf/mfcs/BriceRB23,Pareto-Rational-Verification} $\machine{}=(M,m_0,\delta,\tau)$ on $\arena$, where $M$ is a finite set of \emph{memory states}, $m_0\in M$ is the \emph{initial state}, $\delta:M\times V\rightarrow 2^M$ is the \emph{update} function, and $\tau:M\times V_i\rightarrow 2^V$ is the \emph{next-move} function. Such a machine embeds a (possibly infinite) set of strategies $\strategyfor{i}$ for player~$i$, called \emph{compatible} strategies. Formally, $\strategyfor{i}$ is compatible with $\machine{}$ if there exists a mapping $h \mapsto m_h$ such that $m_{hv} \in \delta(m_h,v)$ for every $hv \in \Hist(v_0)$ (with $m_h = m_0$ when $h$ is empty), and when $v \in V_i$, $\strategyfor{i}(hv) \in \tau(m_h,v)$. An example of such a machine~$\machine{}$ is given in \cref{appendix:exMealyMachine}. We denote by $\llbracket \machine{} \rrbracket$ the set of all strategies compatible with~$\machine{}$. The \emph{memory size} of $\machine{}$ is equal to $|M|$. We say that $\machine{}$ is \emph{deterministic} when the image of both functions $\delta$ and $\tau$ is a singleton. Thus when $\machine{}$ is deterministic, $\llbracket \machine{} \rrbracket = \{\strategyfor{i}\}$ and $\strategyfor{i}$ is called \emph{finite-memory}, and when additionally $|M| = 1$, $\strategyfor{i}$ is called \emph{memoryless}.

\section{Studied Problems}\label{section:studied-problems}

In this section, within the context of rational synthesis and verification, we consider a reachability game $\game = (\arena, (\reach{i})_{i \in \Players})$ with $\arena$ an initialized weighted arena and $\Players = \{0,1, \dots, t\}$ such that player~$0$ is a specific player, often called \emph{system} or \emph{leader}, and the other players~$1,\dots,t$ compose the \emph{environment} and are called \emph{followers}. Player~$0$ announces his strategy $\strategyfor{0}$ at the beginning of the game and is not allowed to change it according to the behavior of the other players. The response of those players to $\strategyfor{0}$ is supposed to be \emph{rational}, where the rationality can be described as the outcome of a \emph{Nash equilibrium}~\cite{Nash50} or as a \emph{Pareto-optimal} play~\cite{Stackelberg-Pareto-Synthesis}.

\subparagraph*{Nash Equilibria.}
A strategy profile for the environment is a Nash equilibrium if no player has an incentive to unilaterally deviate from this profile. In other words, no player can improve his cost by switching to a different strategy, assuming that the other players stick to their current strategies.
Formally, given the initial vertex $v_0$ and a strategy $\strategyfor{0}$ announced by player~$0$, a strategy profile $\strategyProfile = (\strategyfor{i})_{i\in\Players}$ is called a \emph{\fixed{} Nash equilibrium} (\fixedEquilibria{}) if for every player~$i\in\Players\ssetminus\{0\}$ and every strategy $\tau_i \in \Sigma_i$, it holds that $\cost{i}{\outcomefrom{\strategyProfile}{v_0}} \leq \cost{i}{\outcomefrom{\tau_{i},\strategyfor{-i}}{v_0}}$, where $\strategyfor{-i}$ denotes $(\strategyfor{j})_{j \in \Players \setminus \{i\}}$, i.e., $\tau_i$ is not a profitable deviation.
We also say that $\strategyProfile$ is a \emph{\fixedEquilibriaStrategy{}} to emphasize the strategy $\strategyfor{0}$ of player~$0$.

\subparagraph*{Pareto-Optimality.}

When all players collaborate to obtain a best cost for everyone, we need another concept of rationality. In that case, we suppose that the players in $\Players \ssetminus \{0\}$ form a \emph{single} player, player~$1$, that has a \emph{tuple} of targets sets $(\reach{i})_{i \in \{1,\dots,t\}}$. For each play~$\pi \in \Plays(v_0)$, player~$1$ gets a cost tuple $\payoff{\pi} = (\costfunc{i}(\pi))_{i \in \{1,\dots,t\}}$, and prefers $\pi$ to $\pi'$ if $\payoff{\pi} < \payoff{\pi'}$ for the componentwise partial order $<$ over $(\N \cup \{+ \infty\})^t$. Given such a modified game and a strategy $\strategyfor{0}$ announced by player~$0$, we consider the set $C_{\strategyfor{0}}$ of cost tuples of plays consistent with $\strategyfor{0}$ that are \emph{Pareto-optimal} for player~$1$, i.e., minimal with respect to $<$. Hence, $C_{\strategyfor{0}} = \min\{\payoff{\pi} \mid \pi \in \Plays(v_0) \text{ consistent with } \strategyfor{0}\}$. Notice that $C_{\strategyfor{0}}$ is an antichain. A cost tuple $p$ (called cost in the sequel) is said to be $\strategyfor{0}$\emph{-fixed Pareto-optimal} ($\strategyfor{0}$-fixed PO or simply $0$-fixed PO) if $p \in C_{\strategyfor{0}}$. Similarly, a play is said to be $\strategyfor{0}$-fixed PO if its cost is $\strategyfor{0}$-fixed PO.

In some problems studied in this paper, we will have to consider games such that all vertices owned by player~$0$ have only one successor, which means that player~$0$ has no choice but to choose this successor. In this case, we say that \emph{player~$1$ is the only one to play}.

\subparagraph*{Rational Verification.}

We now present the studied decision problems related to the concept of \emph{rational verification}. Given some threshold $c \in \N$, the goal is to verify that a strategy $\strategyfor{0}$ announced by player~$0$ guarantees him a cost $\cost{0}{\pi}\leq c$ whatever the rational response $\pi$ of the environment. By rational response, we mean either a \fixedEquilibriaStrategy{} outcome $\pi$, or a $\strategyfor{0}$-fixed PO play $\pi$.
The strategy $\strategyfor{0}$ is usually given as a deterministic Mealy machine. We can go further: with a nondeterministic Mealy machine, we want to verify whether all strategies $\strategyfor{0} \in \llbracket \mathcal{M} \rrbracket$ are solutions. In the latter case, we speak about \emph{universal} verification.

\begin{problem} \label{problem:verification}
    Given a reachability game $\game$ with an initialized arena, a nondeterministic Mealy machine $\machine{0}$ for player~$0$, and a threshold $c \in \N$,
\begin{itemize}
\item If $\llbracket\machine{0}\rrbracket = \{\strategyfor{0}\}$, the \emph{\probname{}{NC}{N}{V} problem (\prob{NCNV})} asks whether for all \fixedEquilibriaStrategy s $\sigma$, it holds that $\cost{0}{\outcomefrom{\sigma}{v_0}}\leq c$.
\item The
\emph{\probname{U}{NC}{N}{V} problem (\prob{UNCNV})} asks whether for all $\strategyfor{0} \in \llbracket \machine{0}\rrbracket$ and all \fixedEquilibriaStrategy s $\sigma$, it holds that $\cost{0}{\outcomefrom{\sigma}{v_0}}\leq c$.
\item If $\llbracket\machine{0}\rrbracket = \{\strategyfor{0}\}$, the
\emph{\probname{}{NC}{P}{V} problem (\prob{NCPV})} asks whether for all $\strategyfor{0}$-fixed PO plays $\pi$, it holds that $\cost{0}{\pi}\leq c$.
\item The
\emph{\probname{U}{NC}{P}{V} problem (\prob{UNCPV})} asks whether for all $\strategyfor{0} \in \llbracket \machine{0} \rrbracket$ and all $\strategyfor{0}$-fixed PO plays $\pi$, it holds that $\cost{0}{\pi}\leq c$.
\end{itemize}
\end{problem}

\subparagraph*{Rational Synthesis.}

We consider the more challenging problem of \emph{rational synthesis}. 
Given a threshold $c \in \N$, the goal is to synthesize a strategy $\strategyfor{0}$ for player~$0$ (instead of verifying some $\strategyfor{0}$) that guarantees him a cost $\cost{0}{\pi}\leq c$ whatever the rational response $\pi$ of the environment. We also consider the simpler problem where the environment \emph{cooperates} with the leader by proposing \emph{some} rational response $\pi$ that guarantees him a cost $\cost{0}{\pi}\leq c$.

\begin{problem} \label{problem:synthesis}
    Given a reachability game $\game$ with an initialized arena and a threshold $c \in \N$,
    \begin{itemize}
        \item The \emph{\probname{}{C}{N}{S} (\prob{CNS})} problem asks whether there exists $\strategyfor{0} \in \Sigma_0$ and a \fixedEquilibriaStrategy{} $\sigma$ such that $\cost{0}{\outcomefrom{\sigma}{v_0}}\leq c$.
        \item The \emph{\probname{}{NC}{N}{S} (\prob{NCNS})} problem asks whether there exists $\strategyfor{0} \in \Sigma_0$ such that for all \fixedEquilibriaStrategy s $\sigma$, it holds that $\cost{0}{\outcomefrom{\sigma}{v_0}}\leq c$.
        \item The \emph{\probname{}{C}{P}{S} (\prob{CPS})} problem asks whether there exists $\strategyfor{0} \in \Sigma_0$ and a $\strategyfor{0}$-fixed PO play $\pi$ such that $\cost{0}{\pi}\leq c$.
        \item The \emph{\probname{}{NC}{P}{S} (\prob{NCPS})} problem asks whether there exists $\strategyfor{0} \in \Sigma_0$ such that for all $\strategyfor{0}$-fixed PO plays $\pi$, it holds that $\cost{0}{\pi}\leq c$.
    \end{itemize}
\end{problem}

\begin{figure}[t]
  \centering
  \begin{minipage}[b]{0.42\textwidth}
    \centering
    \begin{tikzpicture}[automaton,scale=0.9,every node/.style={scale=0.6},node distance=1.5]
      \node[initial,state,environment] (v0)                     {$v_0$};
      \node[state,system,label=above:{$\playerdiamond$}]         (v1) [below right=.7cm and 1cm of v0] {$v_1$};
      \node[state,environment2]   (v2) [right=2.5cm of v0] {$v_2$};
      \node[state,system,label=below:{$0,\playersquare$}]         (v3) [right=of v1] {$v_3$};
      \node[state,system,label=below:{$0,\playerdiamond$}]         (v4) [below right=.7cm and 1.5cm of v2] {$v_4$};
      \node[state,system,label=above:{$\playersquare$}]         (v5) [right=1.5cm of v2] {$v_5$};

      \path (v0) edge (v1)
                 edge (v2)
            (v1) edge (v3)
                 edge (v2)
            (v2) edge (v4)
                 edge (v5)
            (v3) edge[loop right] (v3)
            (v4) edge[loop right] (v4)
            (v5) edge[loop right] (v5);
    \end{tikzpicture}
    \caption{An example illustrating the two concepts of rational response.}
    \label{fig:example-nash-pareto}
  \end{minipage}
  \hfill
  \begin{minipage}[b]{0.56\textwidth}
    \centering
    \begin{tikzpicture}[automaton,scale=1,every node/.style={scale=0.7},node distance=1.5]
      \node[initial,state,environment] (v0)                     {$v_0$};
      \node[state,system]         (T) [right=of v0] {$v_1$};

      \path (v0) edge             node {$(1,0)$} (T)
                 edge[loop above] node {$(1,0)$} (v0)
            (T)  edge[loop above] node {$(1,0)$} (T);
    \end{tikzpicture}
    \caption{An example showing that PO lasso plays in the \coprob{NCPV} problem may have an exponential length.}
    \label{fig:example-pareto-ncpv-exponential-witness-c}
  \end{minipage}
\end{figure}

\begin{example}

To illustrate these problems, let us study a simple example depicted in \cref{fig:example-nash-pareto} with three players: the system, player~$0$, and two players in the environment, players~$\playersquare$ and~$\playerdiamond$. Player~$0$ owns the circle vertices, player~$\playersquare$ owns the square initial vertex $v_0$, and player~$\playerdiamond$ owns the diamond vertex $v_2$.
Each player~$i$ has a target set, $\reach{0} = \{v_3,v_4\}$, $\reach{\playersquare} = \{v_3,v_5\}$ and $\reach{\playerdiamond} = \{v_1,v_4\}$, and a constant weight $w_i(e) = 1$ for all $e \in E$. When a vertex $v$ is in $\reach{i}$, we depict the symbol of player~$i$ nearby $v$. As the graph is acyclic, the possible player strategies are all memoryless. In the sequel, we thus only indicate the successor chosen by the player.

Let us show that $\strategyfor{0}$ defined by $\strategyfor{0}(v_1) = v_2$ is a solution to the \prob{NCNS} problem with the threshold $c = 3$. Given $\strategyfor{0}$, there exist four distinct strategy profiles $\sigma = (\strategyfor{0},\strategyfor{\playersquare},\strategyfor{\playerdiamond})$. When, for example, $\strategyfor{\playersquare}(v_0)=v_2$ and $\strategyfor{\playerdiamond}(v_2)=v_5$, we abusively denote $\sigma$ as $\{v_0 \rightarrow v_2, v_2 \rightarrow v_5\}$:
    \begin{itemize}
    \item $\{v_0 \rightarrow v_2, v_2 \rightarrow v_5\}$ is not a \fixedEquilibriaStrategy{} because its outcome $\pi_1 = v_0v_2(v_5)^\omega$ has a infinite cost for player~$\playerdiamond$ who will deviate from $v_2$ to $v_4$ to get a cost of $2$; 
    \item similarly, $\{v_0 \rightarrow v_1, v_2 \rightarrow v_5\}$ with outcome $\pi_2 = v_0v_1v_2(v_5)^\omega$ is not a \fixedEquilibriaStrategy{};
    \item the profile $\{v_0 \rightarrow v_1, v_2 \rightarrow v_4\}$ is a \fixedEquilibriaStrategy{}, its outcome is $\pi_3 = v_0v_1v_2(v_4)^\omega$ with $\cost{\playersquare}{\pi_3} = +\infty$, $\cost{\playerdiamond}{\pi_3} = 1$ and $\cost{0}{\pi_3} = 3 \leq c$, so if player~$\playersquare$ deviates from $v_1$ to $v_2$, his cost is still $+\infty$, and player~$\playerdiamond$ has no incentive to deviate since $\cost{\playerdiamond}{\pi_3}$ is already the smallest available;
    \item the profile $\{v_0 \rightarrow v_2, v_2 \rightarrow v_4\}$ with the outcome $\pi_4 = v_0v_2(v_4)^\omega$ is also a \fixedEquilibriaStrategy{} and $\cost{0}{\pi_4} = 2 \leq c$.
\end{itemize}

So, $\strategyfor{0}$ is a solution to the \prob{NCNS} problem with $c = 3$, but not with $c = 2$. It is also a solution for the \prob{CNS} problem.
One can verify that $\strategyfor{0}'$ such that $\strategyfor{0}'(v_1) = v_3$ is a solution to the \prob{NCNS} problem with $c = 2$, since the only $\strategyfor{0}'$-fixed NE outcome is $\pi_5 = v_0v_1(v_3)^\omega$.

We now show that $\strategyfor{0}$ is not a solution to the \prob{CPS} problem with $c = 2$. Let us consider the same four outcomes as before. Their cost for the environment are: $\payoff{\pi_1} = (2,+\infty)$, $\payoff{\pi_2} = (3,1)$, $\payoff{\pi_3} = (+\infty,1)$, and $\payoff{\pi_4} = (+\infty,2)$, meaning that $C_{\strategyfor{0}}= \{(2,+\infty), (3,1)\}$.
Consequently, the only $\strategyfor{0}$-fixed PO plays are $\pi_1$ and $\pi_2$, both giving a cost of $+\infty$ to player~$0$. However, the strategy $\strategyfor{0}'$ is a solution, as there is only one $\strategyfor{0}'$-fixed PO play, $\pi_5 = v_0v_1(v_3)^\omega$, with $\payoff{\pi_5} = (2,1)$ and $\cost{0}{\pi_5} = 2$.
\end{example}

\subparagraph*{Main Results.}

Our main results for Problems~\ref{problem:verification}-\ref{problem:synthesis} are the following ones when the rational responses of the environment are $0$-fixed PO plays. One problem, the \probname{}{NC}{P}{S} problem, was already partially solved in~\cite{gaspard-pareto-quanti-reach}\footnote{In~\cite{gaspard-pareto-quanti-reach}, the authors show the \nexptimeComplete{}ness when there is one single common weight function for all players. We have no result on an \nexptime{} membership for the general case so far.}.

\begin{theorem}\label{theorem:pareto-results}
    \begin{enumerate}[(a)]
       \item \label{thm-NCPV} The \probname{}{NC}{P}{V} problem is \piComplete{}.
        \item \label{thm-UNCPV} The \probname{U}{NC}{P}{V} problem is \pspaceComplete{}.
        \item \label{thm-CPS} The \probname{}{C}{P}{S} problem is \pspaceComplete{}{}.
    \end{enumerate}
\end{theorem}

For \fixedEquilibria{} responses of the environment, we obtain the next main results.

\begin{theorem}\label{theorem:Nash-results}
    \begin{enumerate}[(a)]
        \item \label{thm-NCNV} The \probname{}{NC}{N}{V} problem is \conpComplete{}.
        \item \label{thm-UNCNV} The \probname{U}{NC}{N}{V} problem is \conpComplete{}.
        \item \label{thm-CNS} The \probname{}{C}{N}{S} problem is \npComplete{}.
        \item \label{thm-NCNS} The \probname{}{NC}{N}{S} problem is \exptimeHard{}, already with a two-player environment.
        With a one-player environment, it is in \exptime{} and \pspaceHard{}.
    \end{enumerate}
\end{theorem}

These complexity results depend on the size $|V|$ of the arena, the number $t$ of players~$i$ (resp.\ target sets $\reach{i}$) in case of \fixedEquilibria{} responses (resp.\ $0$-fixed PO responses), the maximal weight $W$ encoded in binary appearing in the functions $w_i$, the threshold $c$ encoded in binary, and the size $|M|$ of the Mealy machine $\machine{0}$ (for the verification problems). Note that for all problems except the NCNS problem, the complexity classes are the same for both qualitative and quantitative frameworks (see \cref{table:results}). Hence, in the case of a \emph{unary} encoding of the weights and the threshold $c$, we get the same complexity classes. Only the most challenging proofs are provided in the paper, while the other proofs or results derived from the literature are deferred in the Appendix.

In this paper, we focus on zero or positive weights, because with negative weights, there are simple examples of one-player games with no NE or no PO plays (thus with no rational responses). Furthermore, considering any weights leads to the undecidability of the \prob{NCNS} and \prob{NCPS} problems. Those results are obtained by reduction from the undecidability of zero-sum multidimensional shortest path games~\cite{Randour-games-on-graphs,DBLP:journals/fmsd/RandourRS17}. See details in~\cref{appendix:Undecidable}.

\begin{restatable}{theorem}{undecidabletheorem}
\label{theorem:undecidableNE}
With integer weight functions, the \probname{}{NC}{N}{S} problem and the \probname{}{NC}{P}{S} problem are undecidable.
\end{restatable}

\section{Pareto-Optimality} \label{section:pareto}

In this section, we provide the proofs of the upper bounds in \cref{theorem:pareto-results}. Recall that the environment is here composed of the sole player~$1$ having $t$ target sets $\reach{i}$, and his rational responses to a strategy $\strategyfor{0}$ announced by player~$0$ are $\strategyfor{0}$-fixed PO plays. The lower bounds are proved in \cref{appendix:lower-bounds-pareto} with reductions from QBF or some of its variants~\cite{STOCKMEYER19761}. All those reductions already hold for qualitative reachability games. We thus obtain the same complexity classes as in \cref{theorem:pareto-results} for this class of games, as indicated in \cref{table:results}.

To solve the two verification problems (\prob{NCPV} and \prob{UNCPV}), we first construct the product game\footnote{The product of a game with a Mealy machine is recalled in \cref{appendix:exMealyMachine}.} $\game \times \machine{0}$ of size polynomial in $\game$ and $\machine{0}$, and we \emph{assume} to directly work with this game, \emph{again denoted $\game$}. Note that in the product game, when $\machine{0}$ is nondeterministic, player~$0$ is able to play any strategy $\strategyfor{0}$ compatible with $\machine{0}$, and when $\machine{0}$ is deterministic, the verification problems are simplified as there is a single compatible strategy $\strategyfor{0}$. The complement of the \prob{(U)NCPV} problem has many similarities with the \prob{CPS} problem:

\begin{problem} \label{problem:complement-PO}
    The \emph{complement of the \prob{(U)NCPV} problem (\coprob{(U)NCPV})} asks whether there exists $\strategyfor{0} \in \Sigma_0$ and a $\strategyfor{0}$-fixed PO play $\pi$ such that $\cost{0}{\pi} > c$.
\end{problem}
Indeed, the statement is the same except that the inequality $\cost{0}{\pi} \leq c$ in the \prob{CPS} problem is here replaced by $\cost{0}{\pi} > c$. To prove the upper bounds of~\Cref{theorem:pareto-results}, we thus have to solve the decision problem ``do there exist $\strategyfor{0} \in \Sigma_0$ and a $\strategyfor{0}$-fixed PO play $\pi$ such that $\cost{0}{\pi} \sim c$ ?'' with $\sim\, \in \{\leq, >\}$.
In short, the algorithm to solve the \prob{CPS} problem and the complement of the \prob{(U)NCPV} problem proceeds through the following steps:
\begin{enumerate}
    \item\label{pareto-synthesis-step1} Guess a play $\pi$ in the form $\pi = \mu(\nu)^\omega$ in polynomial time. The length of the lasso is polynomial or exponential, depending on the studied problem. In the latter case, we will guess a succinct representation of the lasso by using Parikh automata~\cite{parikh,parikh-first}.
    \item\label{pareto-synthesis-step2} Compute in polynomial time $\payoff{\pi}$ and verify in polynomial time that $\cost{0}{\pi} \sim\, c$.
    \item\label{pareto-synthesis-step3} Verify that player~$0$ has a strategy~$\strategyfor{0}$, with $\pi$ consistent with $\strategyfor{0}$, that guarantees that $\payoff{\pi}$ is $\strategyfor{0}$-fixed PO. This last step will be done in \conp{} or in \pspace{}, depending on the studied problem.
\end{enumerate}
Therefore, if a strategy $\strategyfor{0}$ exists as in \cref{pareto-synthesis-step3}, the $\strategyfor{0}$-fixed PO play $\pi$ such that $\cost{0}{\pi} \sim\, c$ is the lasso of \cref{pareto-synthesis-step1}. Let us now provide detailed proofs for these three steps. 

\subsection{Existence of Lassos}

The goal is this section is to prove the next lemma stating that one can always suppose that $\pi$ is a lasso. For that purpose, we use a classical approach consisting of removing cycles~\cite{DBLP:journals/jcss/BrihayeBGT21,DBLP:conf/rp/BriyaheGoeminne23,ConduracheFGR16}.

\begin{lemma}
\label{lem:witness_size_pareto}
    Let $\strategyfor{0} \in \Sigma_0$ and $\pi$ be a $\strategyfor{0}$-fixed PO play $\pi$ such that $\cost{0}{\pi} \sim\, c$. Then there exist $\strategyfor{0}' \in \Sigma_0$ and a $\strategyfor{0}'$-fixed PO play $\pi' = \mu(\nu)^\omega$ such that $\cost{0}{\pi'} \sim\, c$. Moreover, $\Visit{\mu} = \Visit{\mu\nu}$ and
    \begin{itemize}
        \item if $\cost{0}{\pi} \leq c$, then $|\mu| \leq (t+1)|V|$, $|\nu| \leq |V|$, $\payoff{\pi'} \in \{0,1,\dots,B,+\infty\}^t$, with $B = (t+2)|V|W$,
        \item if $\cost{0}{\pi} > c$, then $|\mu| \leq c + (t+1)|V|$, $|\nu| \leq |V|$, $\payoff{\pi'} \in \{0,1,\dots,B,+\infty\}^t$, with $B = (c + (t+2)|V|)W$. 
    \end{itemize}
\end{lemma}

\begin{proof}
    Let $\pi = \pi_0\pi_1\dots$ be a $\strategyfor{0}$-fixed PO play such that $\cost{0}{\pi} \sim\, c$. 

    Suppose that $\cost{0}{\pi} \leq c$. Consider, along $\pi$, any two consecutive first visits to two target sets, say $\reach{i}$ and $\reach{j}$. If there exists $m < n$ such that $\pi_{n} = \pi_{m}$ between these two visits, we remove the cycle $\pi_{[m,n[}$ from $\pi$. We repeat this process until there are less than $|V|$ vertices between the two visits, for any such pair $\reach{i}, \reach{j}$, but also between $\pi_0$ and the first visit to a target set. Let us denote $\pi'$ the resulting play. Consider now along $\pi'$ the last first visit to a target set, say at index $k$. We then seek for the first repeated vertex $\pi'_{\ell_1} = \pi'_{\ell_2}$ with $k \leq \ell_1 < \ell_2$ after $k$. In this way, we obtain $\nu = \pi'_{[\ell_1,\ell_2[}$ with $|\nu| \leq |V|$ and $\mu = \pi'_{[0,\ell_1[}$ with $|\mu| \leq (t+1)|V|$. So, we get the required lasso $\mu(\nu)^\omega$ such that $\Visit{\mu} = \Visit{\mu\nu}$, $\cost{0}{\mu(\nu)^\omega} \leq \cost{0}{\pi} \leq c$, and $\payoff{\mu(\nu)^\omega} \in \{0,1,\dots,B,+\infty\}^t$, with $B = (t+2)|V|W$.

    The case $\cost{0}{\pi} > c$ is treated similarly, except that we cannot remove cycles along the longest prefix $h$ of $\pi$ such that $\cost{0}{h} \leq c$, as this operation might decrease the cost of player~$0$. We thus get $|\mu| \leq c + (t+1)|V|$, $\cost{0}{\mu(\nu)^\omega} > c$, and $\payoff{\mu(\nu)^\omega} \in \{0,1,\dots,B,+\infty\}^t$, with $B = (c+ (t+2)|V|)W$.

    It remains to explain how to construct a strategy $\strategyfor{0}'$ from $\strategyfor{0}$ such that $\pi' = \mu(\nu)^\omega$ is $\strategyfor{0}'$-fixed PO. First, $\strategyfor{0}'$ is built in a way to produce $\pi'$. Second, we have to define $\strategyfor{0}'$ outside $\pi'$, i.e., from any $h'v$, with $v \in V$, such that $h'$ is prefix of $\pi'$ but not $h'v$. Let $h$ be such that the elimination of cycles done in $\pi$, restricted to $h$, leads to $h'$. We then define $\strategyfor{0}'(h'g)= \strategyfor{0}(hg)$ for all histories $g \in \Hist(v)$. Notice that $\strategyfor{0}'$ is the required strategy as the elimination of cycles in a history or a play decreases the costs.
\end{proof}

\begin{example}\label{exemple:pareto-size-expo}
    When $\cost{0}{\pi} > c$, \Cref{lem:witness_size_pareto} provides a bound on $|\mu\nu|$ that is exponential in the binary encoding of $c$. In \cref{fig:example-pareto-ncpv-exponential-witness-c}, we present a small example of a reachability game showing that this is unavoidable. The initial vertex $v_0$ is owned by player~$1$, $v_1$ is owned by player~$0$, and there are two weight functions $w_0$ and $w_1$ (thus $t = 1$). Both players have the same target set: $\reach{0} = \reach{1} = \{v_1\}$. Notice that player~$1$ is the only one to play, and a play $\pi \in \Plays(v_0)$ is PO if and only if visits $T_1$ (and has $\payoff{\pi} = 0$). Hence, given a threshold $c$, any PO play $\pi$ with $\cost{0}{\pi} > c$ is equal to $v_0^{k} (v_1)^\omega$ with $k > c$. The length $|v_0^kv_1|$ is thus greater than $c$. Therefore, \cref{pareto-synthesis-step1} of our decision algorithm for the \coprob{(U)NCPV} cannot guess an explicit representation $\mu(\nu)^\omega$ if we want to stick to a polynomial time algorithm.
\end{example}

\subsection{Particular Zero-sum Games} \label{subsec:particular}

Now that we know we can limit our study to lassos $\pi$, \cref{pareto-synthesis-step3} requires to verify that player~$0$ has a strategy $\strategyfor{0}$ ensuring that $\payoff{\pi}$ is $\strategyfor{0}$-fixed PO. Before going deeper into this step, we need to study some particular two-player zero-sum games.\footnote{We suppose that the reader is familiar with this concept.} Let $\arena = (V,E,\Players,(V_i)_{i \in \Players},(w_i)_{i \in \{1,\dots,t\}})$ be an arena with $\Players = \{\E,\A\}$ and equipped with $t$ weight functions $w_i: E \rightarrow \N$. We suppose that $\arena$ is initialized with $v_0 \in V$. We fix $t$ target sets $\reach{i} \subseteq V$ and $t$ constants $d_i \in \Nzero \cup \{+\infty\}$. We denote by $\game = (\arena,\Omega)$ a zero-sum game whose \emph{objective} $\Omega$ is a Boolean combination of the following objectives:
    \begin{itemize}
        \item $\reachBounded{d_i}{\reach{i}} = \{\pi \in \Plays(v_0) \mid \cost{i}{\pi} < d_i\}$ called \emph{bounded reachability objective}, and
        \item $\safetyBounded{d_i}{\reach{i}} = \Plays(v_0) \setminus \reachBounded{d_i}{\reach{i}}$ called \emph{bounded safety objective}.
        \end{itemize}
\emph{Solving} such a game $\game$ means to decide whether \EE{} has a strategy $\sigma$ such that all plays $\pi \in \Plays(v_0)$ consistent with $\sigma$ belong to the objective $\Omega$. If such a strategy $\sigma$ exists, we say that $\sigma$ is \emph{winning for $\Omega$} and that the initial vertex \emph{$v_0$ is winning for \EE{} for $\Omega$}.

For the PO-check required for \cref{pareto-synthesis-step3}, will see in \cref{subsec:check} that we need to solve the zero-sum games stated in the next two propositions, where the constants $d_i$ are encoded in binary. The second proposition will be used in the general case of nondeterministic Mealy machines $\machine{0}$ while the first one will be used in the deterministic case. Proposition~\ref{theorem:gen-bounded-reach-npcomplete} is a quantitative extension of a result in~\cite{DBLP:journals/tsi/FijalkowH13} about (qualitative) generalized reachability games. 

\begin{proposition}\label{theorem:gen-bounded-reach-npcomplete}
    Let $\game = (\arena,\Omega)$ be a zero-sum game with $\Omega = \bigcap_{1 \leq i \leq t} \reachBounded{d_i}{\reach{i}}$ and \EE{} is the only one to play. Deciding whether $v_0$ is winning for \EE{} is an \npComplete{} problem.
\end{proposition}

\begin{proof}
    We first notice that if \EE{} has a winning strategy from $v_0$, i.e., there exists a play $\pi \in \Omega$, then we can eliminate cycles as in the proof of~\Cref{lem:witness_size_pareto}. Therefore, there exists a lasso $\pi' = \mu(\nu)^\omega \in \Omega$ where $|\mu\nu| \leq (t+2)|V|$. Thus, to get an algorithm in \np{}, we guess such a lasso $\pi'$ and verify that $\cost{i}{\pi'} < d_i$ for each $i \in \{1, \dots, t\}$. This is possible in polynomial time with the costs encoded in binary. It is proved in~\cite{DBLP:journals/tsi/FijalkowH13} that solving (qualitative) generalized reachability games with $V_\A = \emptyset$ is \npComplete{}. Our problem is thus \npHard{} by a reduction from the previous problem with the same arena, the weight functions assigning a null weight to all edges, and by setting $(d_1,\dots,d_t) = (+\infty,\dots,+\infty)$.
\end{proof}

The next proposition, of potential independent interest, is easily extended to any positive Boolean combinations of bounded safety objectives.

\begin{proposition}\label{lem:dev-pareto-zero-sum-pspace}
    Let $\game = (\arena,\Omega)$ be a zero-sum game where
    \[
    \Omega = \left( \bigcap_{1 \leq i \leq t} \safetyBounded{d_i}{\reach{i}} \right) \cup \left( \bigcup_{1 \leq i \leq t} \safetyBounded{d_i+1}{\reach{i}} \right),
    \]
    such that $+\infty + 1 = +\infty$. Then, deciding whether $v_0$ is winning for \EE{} is in \pspace{}.
\end{proposition}

\begin{proof}
    First, let $\Omega^{(1)} = \left( \bigcap_{1 \leq i \leq t} \safetyBounded{d_i}{\reach{i}} \right)$ and $\Omega^{(2)} = \left( \bigcup_{1 \leq i \leq t} \safetyBounded{d_i+1}{\reach{i}} \right)$. We solve the game $(\arena,\Omega)$ by using a recursive algorithm. To know whether $v_0$ is winning for \EE{}, we run a depth-first search over a finite tree rooted at $v_0$ that is the (truncated) unraveling of $\arena$, and we keep track of the accumulated weights along the explored branch as a tuple $(c_i)_{1 \leq i \leq t}$, where each $c_i$ is encoded in binary. Each explored branch $h$ will have its leaf decorated by a boolean $f(h) = \bot$ (\EE{} is losing) or $f(h) = \top$ (\EE{} is winning) according to some rules that we describe below.
    Then the depth-first search algorithm backwardly assigns a boolean to the internal nodes of the tree according to the following rule:
    for any $hv \in V^*V_\E$, we have $f(hv) = \top$ if there exists $v' \in \succc{v}$ such that $f(hvv') = \top$, otherwise $f(hv) = \bot$, while for any $hv \in V^*V_\A$, we have $f(hv) = \top$ if for all $v' \in \succc{v}$, $f(hvv') = \top$, otherwise $f(hv) = \bot$. To have an algorithm executing in polynomial space, the depth of the tree must be polynomial.

    Along a branch, the rules are the following to stop the exploration (the objective $\Omega$ may be modified during the exploration): 
    \begin{itemize}
    \item 
    If for some $i$, the current weight $c_i$ is such that $c_i \geq d_i + 1$ and $\reach{i}$ was not visited, then we can stop the branch $h$ and set $f(h) = \top$. Indeed, $\Omega^{(2)}$ is satisfied, and thus also $\Omega$.
    \item If for some $i$, we have $c_i < d_i$ while visiting $\reach{i}$, then $\Omega^{(1)}$ is not satisfiable anymore, and we continue the exploration with the sole objective $\Omega^{(2)}$ where the $i$-th objective $\safetyBounded{d_i+1}{\reach{i}}$ being ignored (as it is not satisfied).
    \item If for some $i$, we have $c_i = d_i$ while visiting $\reach{i}$, then we continue the exploration with $\Omega$ such that $\safetyBounded{d_i}{\reach{i}}$ is removed from $\Omega^{(1)}$ (as it is satisfied) and $\safetyBounded{d_i+1}{\reach{i}}$ is removed from $\Omega^{(2)}$ (as it is not satisfied).
    \item If $\Omega^{(1)}$ becomes an empty intersection, then we stop the branch $h$ and set $f(h) = \top$.
    \item If $\Omega^{(1)}$ has been removed from $\Omega$ (because it was not satisfiable anymore) and $\Omega^{(2)}$ becomes an empty union, then we stop the branch $h$ and set $f(h) = \bot$. 
    \item There is one more case to stop the branch $h$: when some vertex $v$ is visited twice, i.e., $h = gvg'v$ for some $g,g' \in V^*$. Then we stop the branch and set $f(h) = \top$. Indeed, we stand in a better situation in $gvg'v$ than in $gv$ concerning the accumulated weights, as we consider bounded safety objectives.
    \end{itemize}

    The last case happens as soon as the explored branch has length $|V| + 1$ and the other cases do not occur. Therefore, as there are $t$ bounded safety objectives in both $\Omega^{(1)}$ and $\Omega^{(2)}$, any branch has a length polynomially bounded by $t|V|$. Moreover, the accumulated weights $c_i$ are all bounded by $t|V| W$, thus stored in a polynomial space when encoded in binary. We can thus decide in polynomial space whether $v_0$ is winning for \EE{} for $\Omega$.
\end{proof}

\subsection{Pareto-Optimality Check} \label{subsec:check}

Let us come back to our reachability games. We can now solve \cref{pareto-synthesis-step3} where given a lasso~$\pi$ with $\payoff{\pi} \in \{0,1,\dots,B,+\infty\}^t$ (by~\Cref{lem:witness_size_pareto}), we want to verify whether player~$0$ has a strategy~$\strategyfor{0}$ guaranteeing that $\payoff{\pi}$ is $\strategyfor{0}$-fixed PO. If player~$1$ is the only one to play in the game, it reduces to verify that $\payoff{\pi}$ is PO. The latter problem is in \conp{} as stated in the next lemma, while the former is in \pspace{} as stated in~\Cref{lem:pareto-optimal-check-p1-p2}.

\begin{lemma}\label{lem:pareto-optimal-check-p1-alone}
    Suppose that player~$1$ is the only one to play. Deciding whether a given cost $p \in \{0,1,\dots,B,+\infty\}^t$ is PO is in \conp{}.
\end{lemma}

\begin{proof}
    The cost $p$ is \emph{not} PO if there exists a play $\pi' \in \Plays(v_0)$ such that $\cost{i}{\pi'} \leq p_i$ for all $i \in \{1,\dots,t\}$ and $\cost{j}{\pi'} < p_j$ for some $j \in \{1,\dots,t\}$. That is, if for some $j$, there exists a play $\pi' \in \Omega^{(j)} = \bigcap_{i\neq j} \reachBounded{p_i + 1}{\reach{i}} \cap \reachBounded{p_j}{\reach{j}}$. Solving the zero-sum game $(\arena,\Omega)$ is in \np{} by~\Cref{theorem:gen-bounded-reach-npcomplete}. This concludes the proof.
\end{proof}

\begin{lemma}\label{lem:pareto-optimal-check-p1-p2}
    Given $p = \payoff{\pi} \in \{0,1,\dots,B,+\infty\}^t$ being the cost of a play $\pi$, deciding whether player~$0$ has a strategy $\strategyfor{0}$ ensuring that $p$ is $\strategyfor{0}$-fixed PO is in \pspace{}.
\end{lemma}

\begin{proof}
    To prove the lemma, we first fix a prefix $h$ of $\pi$, with $v \in V$, such that $hv$ is not a prefix of $\pi$ ($hv$ is called a \emph{deviation}), and we study the zero-sum game $(\arena,\Omega^{(hv)})$ with the objective $\Omega^{(hv)}$ equal to $\{\pi' \in \Plays(v) \mid \neg (\payoff{h\pi'} < p)\}$. Let us show that deciding whether $v$ is winning for player~$0$ for $\Omega^{(hv)}$ is in \pspace{}. Notice that for each $i \in \{1,\dots,t\}$ such that $h$ does not visit $\reach{i}$, we have, with $q_i = w_i(hv)$ and $+\infty - q_i = +\infty$: $
    \cost{i}{h\pi'} < p_i$ if and only if $\cost{i}{\pi'} < p_i - q_i$.
    Let us rewrite the condition $\neg (p' < p)$ with $p,p' \in \N^t$ as follows: $
    \left(\forall i\in\{1,\dots,t\}~ p'_i \geq p_i\right) ~\lor~ \left(\exists i\in\{1,\dots,t\}~ p'_i > p_i\right)$.  Hence, the objective $\Omega^{(hv)}$ can be rewritten as
    $
     \left( \bigcap_{\substack{1 \leq i \leq t \\ \occ{h} \cap \reach{i} = \emptyset}} \safetyBounded{p_i - q_i}{\reach{i}} \right) \cup \left( \bigcup_{\substack{1 \leq i \leq t \\ \occ{h} \cap \reach{i} = \emptyset}} \safetyBounded{p_i - q_i + 1}{\reach{i}} \right).
    $
    By~\Cref{lem:dev-pareto-zero-sum-pspace}, given the constants $p_i$ and $q_i$, we can check whether $v$ is winning for player~$0$ in polynomial space. Notice that each $q_i$ can be computed in polynomial space by accumulating the weights, with respect to $w_i$, as long as $\reach{i}$ is not visited (as $q_i \leq p_i$).

    Second, given two deviations $hv, h'v$ ending with the same vertex $v$ and such that $h$ is prefix of $h'$, if $\Visit{h'} = \Visit{h}$ and $v$ is winning for $\Omega^{(hv)}$, then $v$ is also winning for $\Omega^{(h'v)}$ (with the same strategy). Indeed, the constants $q'_i$ for $h'v$ are greater than the constants $q_i$ for~$hv$. We are thus in a ``better situation'' than in $\Omega^{(h'v)}$. So, it suffices to consider polynomially many deviations $hv$, as $\pi$ can visit at most $t$ target sets and there are at most $|V|$ vertices~$v$.

    Finally, deciding whether player~$0$ has a strategy $\strategyfor{0}$ ensuring that $p$ is $\strategyfor{0}$-fixed PO amounts to solving the zero-sum games $(\arena,\Omega^{(hv)})$ for polynomially many deviations $hv$. If player~$0$ has a winning strategy $\strategyfor{hv}$ for all those games, the required strategy $\strategyfor{0}$ is defined as $\strategyfor{0}(g) = \strategyfor{hv}(vg')$ for all histories $g$ such that $g = hvg'$ with the longest prefix $h$ of $\pi$.
\end{proof}

\subsection{Upper Bounds} \label{subsec:UpperBounds}

We are now ready to prove the upper bounds in \Cref{theorem:pareto-results} by providing the announced algorithms for Steps~\ref{pareto-synthesis-step1}-\ref{pareto-synthesis-step3}. The proof is divided according to the considered problem.
We need to recall~\cite{parikh} that a Parikh automaton is a nondeterministic finite automaton (NFA) over an alphabet $\Sigma$ and whose transitions are weighted by tuples in $\N^k$, together with a semilinear set $\mathbf{C} \subseteq \N^k$. It accepts a word $w \in \Sigma^*$ if there exists a run on $w$ ending on an accepting state such that the sum of all encountered weight tuples belongs to $\mathbf{C}$. The non-emptiness problem for Parikh automata is \npComplete{} for numbers encoded in binary~\cite{parikh}.

\begin{proof}[Proof of the upper bounds in \cref{theorem:pareto-results}] 
We begin with the \prob{CPS} problem (\cref{theorem:pareto-results}.\ref{thm-CPS}). Let us give an algorithm in \pspace{} that decides whether there exist $\strategyfor{0} \in \Sigma_0$ and a $\strategyfor{0}$-fixed PO play $\pi$ such that $\cost{0}{\pi} \leq c$. By \cref{lem:witness_size_pareto}, we guess a lasso $\pi = \mu(\nu)^\omega$ with $|\mu\nu| \leq (t+2)|V|$, in time polynomial in $|V|$ and $t$. Then, we compute $p = \payoff{\pi}$ and $\cost{0}{\pi}$ and check whether $\cost{0}{\pi} \leq c$. This can be done in time polynomial in $t$, $|V|$, and the binary encoding of $W$ and $c$ by~\cref{lem:witness_size_pareto}. Finally, by \cref{lem:pareto-optimal-check-p1-p2}, we verify in polynomial space whether player~$0$ has a strategy $\strategyfor{0}$ ensuring that $p$ is $\strategyfor{0}$-fixed PO.

\medskip
For the \prob{NCPV} problem (\cref{theorem:pareto-results}.\ref{thm-NCPV}), recall that we consider its complementary \coprob{NCPV} problem (see \cref{problem:complement-PO}), and that player~$1$ is the only one to play. We begin by giving an algorithm in \np{} for \cref{pareto-synthesis-step1,pareto-synthesis-step2}. \Cref{lem:witness_size_pareto} does not provide a polynomial bound on the length of the lasso $\pi = \mu(\nu)^\omega$ due to the threshold $c$ given in binary. However, we will guess a succinct representation of $\pi$ by using Parikh automata.

The idea is the following one. Along the prefix $\mu$ of the lasso $\pi$, some target sets $\reach{k_1}, \dots, \reach{k_n}$ are visited, with $n \leq t$, such that the first visits are in vertices $\pi_{\ell_ 1}, \dots, \pi_{\ell_n}$ with $\ell_1 < \dots < \ell_n$. And after $\mu$, no more target sets are visited along $\mu\nu$ (see~\Cref{lem:witness_size_pareto}).
We start by guessing a sequence $v_0, v_1, \dots, v_n, v_{n+1}$ of vertices, called \emph{markers}, with the aim that $v_0$ is the initial vertex, $v_i = \pi_{\ell_i}$, $1 \leq i \leq n$, and $v_{n+1} = \first{\nu}$. By~\Cref{lem:witness_size_pareto}, we know that $\payoff{\pi} \in \{0,1,\dots,B,+\infty\}^t$, where $B = (c+(t+2)|V|)W$. We thus guess a tuple $(p_0,p_1, \dots, p_t) \in \{0,1,\dots,B,+\infty\}^t$ with the aim that $(p_1,\dots,p_t) = \payoff{\mu}$ and $p_0 = w_0(\mu)$. We also guess for each \emph{portion} $\pi_{[\ell_i,\ell_{i+1}]}$, $i \leq n$,
\begin{itemize}
\item a weight $q^{(i)}_0 \in \{0,1,\dots,B\}$ for player~$0$ with the aim that $q^{(i)}_0 = w_0(\pi_{[\ell_i,\ell_{i+1}]})$ and $w_0(\mu) = p_0 = \Sigma_i q^{(i)}_0$,
\item a ``useful'' environment weight tuple, i.e., for all $j \in \{1,\dots,t\}$, a weight $q^{(i)}_j \in \{0,1,\dots,B\}$ such that $\pi_{[0,\ell_{i}]}$ does not visit $\reach{j}$, with the aim that $q^{(i)}_j = w_j(\pi_{[\ell_i,\ell_{i+1}]})$ and $\cost{j}{\mu} = p_j = \Sigma_i q^{(i)}_j$.\footnote{If $\pi_{[0,\ell_{i}]}$ visits $\reach{j}$, then $\cost{j}{\pi}$ is already known as $\cost{j}{\pi} = \cost{j}{\pi_{[0,\ell_{i}]}}$.}
\end{itemize}
We can guess in polynomial time the sequence $v_0, v_1, \dots, v_n, v_{n+1}$ and the constants $p_j$, $q^{(i)}_j$ encoded in binary, as $n \leq t$ and $B = (c+(t+2)|V|)W$. We then check in polynomial time that $v_0$ is the initial state, that each $v_i$ belongs to a distinct target set $T_{k_i}$, $1 \leq i \leq n$,
that $p_j = \Sigma_i q^{(i)}_j$ for each $j$, and that $p_0 > c$ for the given threshold~$c$.\footnote{To keep the proof readable, we assume that each $v_i$ belongs to one target set $\reach{k_i}$. In general, it could belong to several target sets. The proof is easily adapted by considering the union of target sets.}

It remains to check the existence of polynomially many paths:
\begin{itemize}
\item For each $i \leq n$, the existence of a path $\rho^{(i)}$ from $v_i$ to $v_{i+1}$ on a subgraph of $\arena$ restricted to some sets $V^{(i)}$ and $E^{(i)}$ of vertices and edges respectively, and to some weight functions, such that $w_j(\rho^{(i)}) = q_j^{(i)}$ for all $j$. 
\item The existence of a path from $v_{n+1}$ to itself (the cycle $\nu$) that visits no new target set with respect to $\reach{k_i}$, $1 \leq i \leq n$.
\end{itemize}
The first check can be done thanks to Parikh automata : 
one can decide in \np{} the existence of a path in a subgraph of $\arena$ between two given vertices and with a given weight tuple $\bar q$ (the subgraph is seen as a Parikh automaton with $\Sigma = \{\#\}$ and $\mathbf{C} = \{\bar q\}$).\footnote{We do not need to use an oracle here. It suffices to plug the \np{} algorithm for Parikh automata in ours as if the required path exists, our algorithm will find it in polynomial time.} The set $V^{(i)}$ is defined as $V \ssetminus \left( \bigcup_{j > i+1} \reach{k_j} \cup \bigcup_{p_j = +\infty} \reach{j} \right )$, and the set $E^{(i)}$ as $(E \cap V^{(i)} \times V^{(i)}) \ssetminus \{(v,v') \mid v \in \reach{k_{i+1}}\}$. Indeed, for the portion $\pi_{[\ell_i,\ell_{i+1}]}$, we do not allow to prematurely visit a target set $\reach{k_j}$, $j \geq i+1$, except $v_{i+1} \in \reach{k_{i+1}}$, and there are target sets that we do not want to visit at all. We also remove the weight function $w_{k_j}$ with $j \in \{1,\dots, i\}$.
The second check can be done thanks to classical automata, by restricting the set of vertices to $V \ssetminus \left( \bigcup_{p_j = +\infty} \reach{j} \right )$.
To show that the \coprob{NCPV} problem is in \sigmaClass, in the previous algorithm in \np{} that guesses a lasso $\pi$ with $\payoff{\pi} = p$, we add an oracle in \conp{} to check whether $p$ is a PO cost thanks to~\Cref{lem:pareto-optimal-check-p1-alone}. As \npConp $=$ \sigmaClass, we get that the \prob{NCPV} problem is in \piClass.

\medskip It remains to show that the \coprob{UNCPV} problem is in \pspace{} to get the upper bound of \cref{theorem:pareto-results}.\ref{thm-UNCPV}). The approach is to guess a cost $p \in \{0,\ldots,B,+\infty\}^t$ and a length $\ell$ for the exponential lasso $\pi$ of~\cref{lem:witness_size_pareto}, whose both encodings in binary use a polynomial space. We guess $\pi$ vertex by vertex, by only storing the current edge $(u,u')$, the current accumulated weight $(c_0,c_1,\dots,c_t)$ on each dimension, and which target sets $\reach{i}$ have already been visited. At any time, the stored information uses a polynomial space. At each guess, we apply the reasoning of \cref{lem:pareto-optimal-check-p1-p2} to check in polynomial space whether player~$0$ can ensure that $p$ is a PO cost from each vertex $v \neq u'$ successor of $u$ (i.e., from any deviation of $\pi$). We also check that for each first visit to a target set $\reach{i}$, we have $c_i = p_i$ if $i \in \{1,\ldots, t\}$, and $c_i > c$ if $i = 0$. At each guess, a counter is incremented until reaching the length $\ell$, where we stop guessing $\pi$ and finally check whether $p_i = +\infty$ for each $\reach{i}$ that has not been visited.

This completes the proof.
\end{proof}

\section{Nash Equilibria} \label{section:Nash}

We now discuss the proofs of \Cref{theorem:Nash-results}. The environment is here composed of $t$ players whose rational responses to a strategy~$\strategyfor{0}$ of player~$0$ are \fixedEquilibriaStrategy{} outcomes. 

The upper bounds for \prob{(U)NCNV} and \prob{CNS} problems given in \cref{theorem:Nash-results}.\ref{thm-NCNV}-\ref{thm-CNS} are proved with the same approach as for Pareto optimality, limited to Steps \ref{pareto-synthesis-step1}-\ref{pareto-synthesis-step2}. There is no need for \cref{pareto-synthesis-step3}, thanks to a well-known characterization of NE outcomes (based on the values of some two-player zero-sum games, see e.g.~\cite{DBLP:journals/jcss/BrihayeBGT21,Bru21} or~\cref{appendix:ne-characterization}) that is directly checked on the lasso guessed in \cref{pareto-synthesis-step1}. We need again Parikh automata to guess a succinct representation of the lasso. The lower bounds for those problems were already known for qualitative reachability games~\cite{grandmont23}. See \cref{appendix:other-results-nash-2p-parikh}.

We thus focus on the \prob{NCNS} problem (\cref{theorem:Nash-results}.\ref{thm-NCNS}). We prove below that this problem is \exptimeHard{}, already for two-player environments. The decidability is left open. This decision problem is a real challenge that cannot be solved by known approaches. Indeed, the technique of tree automata, as used in~\cite{ConduracheFGR16} to show the decidability of several $\omega$-regular objectives, is not applicable in the context of quantitative reachability. This is because, while in the scenario of qualitative reachability, the costs are Boolean and can be encoded within the finite state space of a tree automaton, for quantitative reachability, these costs are now integers that are not bounded and vary according to the strategy~$\strategyfor{0}$ that is being synthesized. Consequently, it is not feasible to directly encode constraints within the states of the automaton in this latter case. Additionally, there is a necessity to enforce constraints related to subtrees, such as comparing (unbounded) costs between two subtrees. Generally, incorporating the capability to enforce subtree constraints in tree automata results in undecidability, with only certain subclasses having a decidable emptiness problem, see e.g.~\cite{DBLP:conf/lics/BargunoCGJV10}. Therefore, addressing the general case would necessitate either advancements in the field of automata theory or an entirely new methodological approach.

However, we are able to solve the practically relevant case of one-player environments for which we prove that the \prob{NCNS} problem is \pspaceHard{} and in \exptime{} in \cref{appendix:exp-algo-2p-ncns}.
The \pspaceHard{}ness is given by a classical reduction from the subset-sum game problem~\cite{game-subset-sum}. The intuition for the \exptime{}-membership is the following: it consists in finding a play $\pi$ where $\cost{0}{\pi} \leq c$ such that when the only component of the environment deviates from $\pi$, either the system inflicts to the deviating play $\pi'$ a cost for the environment such that $\cost{1}{\pi'} > \cost{1}{\pi'}$ meaning that deviating is not profitable, or it ensures a cost for himself such that $\cost{0}{\pi'} \leq c$. Note that this approach only works for one-player environments.

We are also able to solve the \prob{NCNS} problem for any number of players in the environment, for the variant where the rational NE responses of the environment aim to ensure costs bounded by a given threshold rather than minimizing these costs (this is also arguably an interesting model of rationality in practice). This is a perspective studied in \cite{DBLP:conf/ijcai/RajasekaranBansalV23} in
the case of NEs for discounted-sum objectives.
We show in \cref{appendix:2p-ncns-exptime-complete-bounded-reach} that this variant is \exptimeComplete{}.

\begin{restatable}{theorem}{ncnsboundedexptime}
\label{theorem:ncns-bounded-exptime-complete}
    The \probname{}{NC}{N}{S} problem where the objective of each player~$i \in \{1,\dots,t\}$ is a bounded reachability objective $\reachBounded{d_i}{\reach{i}}$ is \exptimeComplete{}, and hardness holds even with a one-player environment.
\end{restatable}

\subparagraph*{Reduction for Two-Player Environments.}

We finally prove that the \prob{NCNS} problem is \exptimeHard{}, already for a two-player environment (lower bound of \cref{theorem:Nash-results}.\ref{thm-NCNS}). The reduction is given from the \emph{countdown game problem}, known to be \exptimeComplete{}~\cite{countdown-game}. Given a threshold $c \in \N$, a countdown game $\mathcal{CG}$ is a two-player zero-sum game played on a directed graph $(V,E)$ where $E \subseteq V \times \Nzero \times V$. A configuration is a pair $(s,k) \in V \times \N$, initially $(s_0,0)$ with $s_0$ an initial vertex, from where player~$0$ chooses $d \in \Nzero$ such that there exists $(s,d,s') \in E$ (we assume that such a $d$ always exists). Player~$1$ then chooses such an $s' \in V$ to reach the configuration $(s',k+d)$. When reaching a configuration $(s,k)$ with $k\geq c$, the game stops and player~$0$ wins if and only if $k = c$.\footnote{Classically, the initial configuration is $(s_0,c)$ and the accumulated weight $k$ decreases until being $\leq 0$.} Player~$0$ wins the game $\mathcal{CG}$ if he has a strategy $\strategyfor{0}$ from $s_0$ that allows him to reach some configuration $(s,c)$, whatever the strategy of player~$1$.

\begin{theorem}\label{theorem:ncns-3players-exptime-hard-countdown-game}
    The \probname{}{NC}{N}{S} problem with a two-player environment is \exptimeHard{}.
\end{theorem}

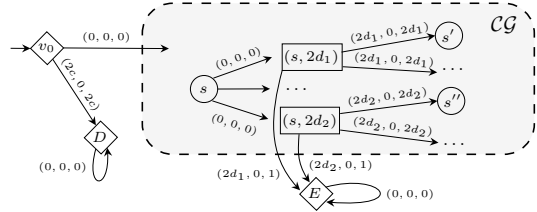
\begin{figure}[t]
    \centering
    \resizebox{0.5\textwidth}{!}{%
    \begin{tikzpicture}[x=0.75pt,y=0.75pt,yscale=-1,xscale=1]
    \draw   (199.96,70.84) -- (210.48,81.04) -- (200.28,91.56) -- (189.76,81.36) -- cycle ;
    \draw    (203.73,88.23) -- (229.1,127.38) ;
    \draw [shift={(230.73,129.9)}, rotate = 237.06] [fill={rgb, 255:red, 0; green, 0; blue, 0 }  ][line width=0.08]  [draw opacity=0] (5.36,-2.57) -- (0,0) -- (5.36,2.57) -- (3.56,0) -- cycle    ;
    \draw  [fill={rgb, 255:red, 155; green, 155; blue, 155 }  ,fill opacity=0.1 ][dash pattern={on 4.5pt off 4.5pt}][line width=0.75]  (259.8,71.91) .. controls (259.8,61.37) and (268.34,52.83) .. (278.87,52.83) -- (486.33,52.83) .. controls (496.86,52.83) and (505.4,61.37) .. (505.4,71.91) -- (505.4,129.13) .. controls (505.4,139.66) and (496.86,148.2) .. (486.33,148.2) -- (278.87,148.2) .. controls (268.34,148.2) and (259.8,139.66) .. (259.8,129.13) -- cycle ;
    \draw    (210.48,81.04) -- (273.48,81.42) ;
    \draw [shift={(276.48,81.44)}, rotate = 180.35] [fill={rgb, 255:red, 0; green, 0; blue, 0 }  ][line width=0.08]  [draw opacity=0] (5.36,-2.57) -- (0,0) -- (5.36,2.57) -- (3.56,0) -- cycle    ;
    \draw   (289.93,106.98) .. controls (289.93,102.27) and (293.76,98.44) .. (298.47,98.44) .. controls (303.18,98.44) and (307.01,102.27) .. (307.01,106.98) .. controls (307.01,111.69) and (303.18,115.52) .. (298.47,115.52) .. controls (293.76,115.52) and (289.93,111.69) .. (289.93,106.98) -- cycle ;
    \draw    (304.34,100.11) -- (338.21,87.89) ;
    \draw [shift={(341.03,86.87)}, rotate = 160.15] [fill={rgb, 255:red, 0; green, 0; blue, 0 }  ][line width=0.08]  [draw opacity=0] (5.36,-2.57) -- (0,0) -- (5.36,2.57) -- (3.56,0) -- cycle    ;
    \draw    (303.68,114.44) -- (335.88,124.97) ;
    \draw [shift={(338.73,125.9)}, rotate = 198.1] [fill={rgb, 255:red, 0; green, 0; blue, 0 }  ][line width=0.08]  [draw opacity=0] (5.36,-2.57) -- (0,0) -- (5.36,2.57) -- (3.56,0) -- cycle    ;
    \draw    (307.01,106.98) -- (336.83,106.88) ;
    \draw [shift={(339.83,106.87)}, rotate = 179.8] [fill={rgb, 255:red, 0; green, 0; blue, 0 }  ][line width=0.08]  [draw opacity=0] (5.36,-2.57) -- (0,0) -- (5.36,2.57) -- (3.56,0) -- cycle    ;
    \draw   (347.43,78.47) -- (385.01,78.47) -- (385.01,95.27) -- (347.43,95.27) -- cycle ;
    \draw   (346.3,118.47) -- (383.34,118.47) -- (383.34,135.27) -- (346.3,135.27) -- cycle ;
    \draw    (384.57,83.67) -- (440.41,73.79) ;
    \draw [shift={(443.37,73.27)}, rotate = 169.97] [fill={rgb, 255:red, 0; green, 0; blue, 0 }  ][line width=0.08]  [draw opacity=0] (5.36,-2.57) -- (0,0) -- (5.36,2.57) -- (3.56,0) -- cycle    ;
    \draw    (385.37,92.9) -- (441.57,96.29) ;
    \draw [shift={(444.57,96.47)}, rotate = 183.45] [fill={rgb, 255:red, 0; green, 0; blue, 0 }  ][line width=0.08]  [draw opacity=0] (5.36,-2.57) -- (0,0) -- (5.36,2.57) -- (3.56,0) -- cycle    ;
    \draw    (383.77,122.47) -- (441.79,114.89) ;
    \draw [shift={(444.77,114.5)}, rotate = 172.56] [fill={rgb, 255:red, 0; green, 0; blue, 0 }  ][line width=0.08]  [draw opacity=0] (5.36,-2.57) -- (0,0) -- (5.36,2.57) -- (3.56,0) -- cycle    ;
    \draw    (383.77,132.87) -- (442.39,140.48) ;
    \draw [shift={(445.37,140.87)}, rotate = 187.4] [fill={rgb, 255:red, 0; green, 0; blue, 0 }  ][line width=0.08]  [draw opacity=0] (5.36,-2.57) -- (0,0) -- (5.36,2.57) -- (3.56,0) -- cycle    ;
    \draw    (347.43,95.27) .. controls (340.06,123.2) and (337.8,148.07) .. (356.84,169) ;
    \draw [shift={(358.67,170.93)}, rotate = 225.45] [fill={rgb, 255:red, 0; green, 0; blue, 0 }  ][line width=0.08]  [draw opacity=0] (5.36,-2.57) -- (0,0) -- (5.36,2.57) -- (3.56,0) -- cycle    ;
    \draw    (358,135.53) .. controls (357.39,145.37) and (356.47,152.82) .. (363.05,163.47) ;
    \draw [shift={(364.67,165.93)}, rotate = 235.22] [fill={rgb, 255:red, 0; green, 0; blue, 0 }  ][line width=0.08]  [draw opacity=0] (5.36,-2.57) -- (0,0) -- (5.36,2.57) -- (3.56,0) -- cycle    ;
    \draw    (230.63,144.87) .. controls (223.65,170.63) and (239.5,172.1) .. (238.06,146.91) ;
    \draw [shift={(237.83,144.07)}, rotate = 84.22] [fill={rgb, 255:red, 0; green, 0; blue, 0 }  ][line width=0.08]  [draw opacity=0] (5.36,-2.57) -- (0,0) -- (5.36,2.57) -- (3.56,0) -- cycle    ;
    \draw    (374.67,168.93) .. controls (410.13,159.08) and (428.36,182.52) .. (377.38,178.47) ;
    \draw [shift={(375,178.27)}, rotate = 5.28] [fill={rgb, 255:red, 0; green, 0; blue, 0 }  ][line width=0.08]  [draw opacity=0] (5.36,-2.57) -- (0,0) -- (5.36,2.57) -- (3.56,0) -- cycle    ;
    \draw   (443.37,73.27) .. controls (443.37,68.55) and (447.19,64.73) .. (451.9,64.73) .. controls (456.62,64.73) and (460.44,68.55) .. (460.44,73.27) .. controls (460.44,77.98) and (456.62,81.8) .. (451.9,81.8) .. controls (447.19,81.8) and (443.37,77.98) .. (443.37,73.27) -- cycle ;
    \draw   (444.77,114.5) .. controls (444.77,109.79) and (448.59,105.96) .. (453.3,105.96) .. controls (458.02,105.96) and (461.84,109.79) .. (461.84,114.5) .. controls (461.84,119.21) and (458.02,123.04) .. (453.3,123.04) .. controls (448.59,123.04) and (444.77,119.21) .. (444.77,114.5) -- cycle ;
    \draw    (177.28,81.27) -- (186.76,81.34) ;
    \draw [shift={(189.76,81.36)}, rotate = 180.41] [fill={rgb, 255:red, 0; green, 0; blue, 0 }  ][line width=0.08]  [draw opacity=0] (5.36,-2.57) -- (0,0) -- (5.36,2.57) -- (3.56,0) -- cycle    ;
    \draw   (233.66,127.24) -- (244.18,137.44) -- (233.98,147.96) -- (223.46,137.76) -- cycle ;
    \draw   (368.62,162.84) -- (379.14,173.04) -- (368.94,183.56) -- (358.42,173.36) -- cycle ;
    
    \draw (193.33,77.2) node [anchor=north west][inner sep=0.75pt]  [font=\scriptsize] [align=left] {$v_{0}$};
    \draw (227.33,133.07) node [anchor=north west][inner sep=0.75pt]  [font=\scriptsize] [align=left] {$D$};
    \draw (477.02,59.8) node [anchor=north west][inner sep=0.75pt]  [font=\normalsize] [align=left] {$\mathcal{CG}$};
    \draw (294.33,104.2) node [anchor=north west][inner sep=0.75pt]  [font=\scriptsize] [align=left] {$s$};
    \draw (348.3,105.4) node [anchor=north west][inner sep=0.75pt]  [font=\footnotesize] [align=left] {$\dotsc $};
    \draw (347.23,81.07) node [anchor=north west][inner sep=0.75pt]  [font=\scriptsize] [align=left] {$(s,2d_{1})$};
    \draw (346.1,121.13) node [anchor=north west][inner sep=0.75pt]  [font=\scriptsize] [align=left] {$(s,2d_{2})$};
    \draw (446.23,92.33) node [anchor=north west][inner sep=0.75pt]  [font=\footnotesize] [align=left] {$\dotsc $};
    \draw (445.93,68.13) node [anchor=north west][inner sep=0.75pt]  [font=\scriptsize] [align=left] {$s'$};
    \draw (447.1,139.67) node [anchor=north west][inner sep=0.75pt]  [font=\footnotesize] [align=left] {$\dotsc $};
    \draw (447.27,109.07) node [anchor=north west][inner sep=0.75pt]  [font=\scriptsize] [align=left] {$s''$};
    \draw (362.2,168.07) node [anchor=north west][inner sep=0.75pt]  [font=\scriptsize] [align=left] {$E$};
    \draw (387.96,71.59) node [anchor=north west][inner sep=0.75pt]  [font=\tiny,rotate=-350.21] [align=left] {$(2d_{1} ,0,2d_{1})$};
    \draw (386.5,111.07) node [anchor=north west][inner sep=0.75pt]  [font=\tiny,rotate=-352.4] [align=left] {$(2d_{2} ,0,2d_{2})$};
    \draw (303.07,88.9) node [anchor=north west][inner sep=0.75pt]  [font=\tiny,rotate=-340.4] [align=left] {$(0,0,0)$};
    \draw (220.9,69.8) node [anchor=north west][inner sep=0.75pt]  [font=\tiny] [align=left] {$(0,0,0)$};
    \draw (305.1,156) node [anchor=north west][inner sep=0.75pt]  [font=\tiny] [align=left] {$(2d_{1} ,0,1)$};
    \draw (362.63,150.8) node [anchor=north west][inner sep=0.75pt]  [font=\tiny] [align=left] {$(2d_{2} ,0,1)$};
    \draw (411.17,168.6) node [anchor=north west][inner sep=0.75pt]  [font=\tiny] [align=left] {$(0,0,0)$};
    \draw (193.57,152.07) node [anchor=north west][inner sep=0.75pt]  [font=\tiny] [align=left] {$(0,0,0)$};
    \draw (213.72,84.4) node [anchor=north west][inner sep=0.75pt]  [font=\tiny,rotate=-57.41] [align=left] {$(2c,0,2c)$};
    \draw (303.28,119.93) node [anchor=north west][inner sep=0.75pt]  [font=\tiny,rotate=-18.6] [align=left] {$(0,0,0)$};
    \draw (391.5,82.43) node [anchor=north west][inner sep=0.75pt]  [font=\tiny,rotate=-3.59] [align=left] {$(2d_{1} ,0,2d_{1})$};
    \draw (391.12,122.49) node [anchor=north west][inner sep=0.75pt]  [font=\tiny,rotate=-7.86] [align=left] {$(2d_{2} ,0,2d_{2})$};
    \end{tikzpicture}
    }%
    \caption{Reduction from the countdown game problem to the \prob{NCNS} problem (two-player env.).}
    \label{fig:NCNS-3players-countdown-game}
\end{figure}

\begin{proof}
Given a countdown game $\mathcal{CG}$ and a threshold $c$, we build a reachability game $\game{}$ as depicted in \cref{fig:NCNS-3players-countdown-game} with three players, player~$0$ (owning the circle vertices of $\mathcal{CG}$), player~$1$ (owning the square vertices of $\mathcal{CG}$), and player~$2$ (owning the initial vertex $v_0$ and vertices $D, E$). The three weight functions are indicated on the edges, with a null weight on all edges for player~$1$. The initial vertex $v_0$ has two outgoing edges, one towards vertex~$D$ and the other one to the initial vertex $s_0$ of $\mathcal{CG}$. Inside $\mathcal{CG}$, players~$0$ and $1$ are simulating the countdown game. The target sets are $\reach{0} = \reach{2} = \{D,E\}$ and $\reach{1} = V$. Thus, for any play, player~$1$ gets a cost of $0$ and will never have the incentive to deviate from his strategy.
The $\mathcal{CG}$ part of the figure contains a slight modification of the given countdown: players~$0$ and~$1$ act as in $\mathcal{CG}$, player~$1$ can exit it by taking the edge to vertex~$E$, the weights $d$ are multiplied by~$2$. More precisely, player~$0$ first selects a transition from a vertex $s$ to some vertex $(s,2d)$, with $d \in \Nzero$, then player~$1$ responds with a successor $s'$ such that $(s,d,s')$ is an edge in the initial countdown game. At any point $(s,2d)$, player~$1$ can exit the $\mathcal{CG}$ by going to $E$, adding $2d$ to the cost of player~$0$ and $1$ to the cost of player~$2$, i.e., it gives the cost tuple $(2k+2d,0,2k+1)$ where $2k$ is the accumulated weight inside $\mathcal{CG}$ before exiting it.

Let us show that a strategy $\strategyfor{0} \in \Sigma_0$ is a solution to the \prob{NCNS} problem with the threshold $2c$ if and only if it is winning in the given countdown game and threshold $c$. We first suppose that $\strategyfor{0}$ is a winning strategy for player~$0$ in the countdown game. We consider this strategy in $\game$ and enumerate all possible plays consistent with $\strategyfor{0}$:
\begin{itemize}
\item The play $v_0(D)^\omega$ gives the cost $2c$ to player~$0$, thus satisfying the threshold $2c$, 
\item No play staying infinitely often in $\mathcal{CG}$ is the outcome of a \fixedEquilibriaStrategy{}, as it gives an infinite cost to player~$2$ while player~$2$ could deviate in $v_0$ to get a cost of $2c < +\infty$, 
\item Any play $\pi$ ultimately reaching $E$ has $\cost{0}{\pi} = 2k+2d$ and $\cost{2}{\pi} = 2k+1$, for some $k \in \N$. If $2k+2d \leq 2c$, then $\cost{0}{\pi}$ satisfies the threshold constraint. Otherwise, $2k+2d > 2c$, but as $\strategyfor{0}$ is winning in the initial countdown game, this means that there was a previous configuration where the costs of both players~$0$ and~$2$ were exactly $2c$. This means that $\cost{2}{\pi} = 2k + 1 \geq 2c+1$, i.e., $\pi$ is again not a \fixedEquilibriaStrategy{} outcome.
\end{itemize}

Assume now that $\strategyfor{0}$ is not winning in the countdown game. Hence, there exists a losing play consistent with $\strategyfor{0}$ in this game, that leads to a play $\pi$ in the grey part of \Cref{fig:NCNS-3players-countdown-game} such that in none of its vertices, the accumulated weight is exactly~$2c$, i.e., there are two consecutive steps where the accumulated weight is $2k < 2c$ and then $2k+2d > 2c$. So, player~$1$ can exit between these two steps to reach $E$. The resulting play $\pi'$ has $\cost{0}{\pi'} = 2k+2d > 2c$ and $\cost{2}{\pi'} = 2k+1 < 2c + 1$, thus $\cost{2}{\pi'} < 2c$. Consequently, $\pi'$ is a \fixedEquilibriaStrategy{} outcome but $\cost{0}{\pi} > 2c$. It follows that $\strategyfor{0}$ is not a solution to the \prob{NCNS} problem. 
\end{proof}

\section{Conclusion} \label{sec:conclusion}

In this paper, we have determined the exact complexity class for several rational verification and synthesis problems in quantitative reachability games, considering both NE and PO rational behaviors of the environment. However, for the \prob{NCNS} problem, while we have solved the important one-player environment case, we have left open the multi-player environment case. We believe this latter case poses a significant challenge that may require new advances in automata techniques to be solved.

There are several interesting future works to investigate. (1) We intend to study the FPT complexity of the studied problems. Notice that some of our lower bounds results already hold for one-player environments (see the CNS and UNCNV problems in Section 5). (2)~Instead of one reachability objective, player 0 could have several ones and a threshold on these objectives that he wants to see satisfied. (3) The concept of NE could be replaced by SPE or by strong NE (that allows collaborations between the players during deviations). Still, it is important to note that strategies $\sigma_0$ that are solutions to the non-cooperative synthesis problems under NE rationality are also solutions under SPE (resp. strong NE) rationality,
as SPEs (resp. strong NEs) constitute a subset of NEs. (4) The result of \nexptimeComplete{}ness of the \prob{NCPS} problem in~\cite{gaspard-pareto-quanti-reach} could be generalized for any weight function, and not a common one used for each player. Or, we could find the more accurate complexity class where it lies if it is not in \nexptime{}.

\bibliography{bibliography}

\appendix

\section{Example of a Nondeterministic Mealy Machine and Product Game} \label{appendix:exMealyMachine}

We first provide an example of a nondeterministic Mealy machine and the way it encodes strategies.

\begin{example}\label{ex:mealy-product}
Consider the arena in \cref{fig:ex-arena} and the nondeterministic Mealy machine $\machine{0}$ of player~$0$ illustrated in \cref{fig:ex-arena-mealy}, formally defined as $\machine{0}=(M,m_0,\delta,\tau)$ such that
\begin{itemize}
    \item $M=\{m_0,m_1\}$,
    \item $\delta(m_0,v_3) = \{m_0,m_1\}$ and $\delta(m,v)=\{m\}$ for every $(m,v)\neq(m_0,v_3)$,
    \item $\tau(m_0,v)=\begin{cases}
        \{v_1, v_3\} & \text{if } v=v_1\\
        \{v_3\} & \text{if } v=v_2
    \end{cases}$, and $\tau(m_1,v)= \{v_2\}$ if $v=v_1$ or $v=v_2$.
\end{itemize}
The idea is to start and stay in the memory state $m_0$ and then, once $v_3$ has been visited, to nondeterministically switch to the memory state $m_1$, or continue staying in the memory state $m_0$. The memory state defines which edge player~$0$ is able to choose from $v_1$: either a nondeterministic choice between $v_1$ and $v_3$ in $m_0$, or $v_2$ in $m_1$.

\begin{figure}[htbp]
  \centering
  \begin{minipage}[b]{0.48\textwidth}
    \centering
    \begin{tikzpicture}[automaton,scale=0.8,every node/.style={scale=0.8}]
      \node[environment2] (v0)              {$v_0$};
      \node[system]       (v1)[right of=v0] {$v_1$};
      \node[system]       (v2)[right of=v1] {$v_2$};
      \node[environment]  (v3)[below of=v1] {$v_3$};
    
      \path (v0) edge             (v1)
                 edge[bend right] (v3)
            (v1) edge[bend right] (v3)
                 edge[loop above] (v1)
                 edge             (v2)
            (v2) edge             (v3)
                 edge[loop below] (v2)
            (v3) edge[bend right=20] (v1);
    \end{tikzpicture}
    \caption{An arena with player~$0$, player~$\playersquare$, and player~$\playerdiamond$, with no weight displayed.}
    \label{fig:ex-arena}
  \end{minipage}
  \hfill
  \begin{minipage}[b]{0.48\textwidth}
    \centering
    \begin{tikzpicture}[automaton,scale=0.8,every node/.style={scale=0.8},node distance=2.5cm]
      \node[system,initial,initial where=above] (m0)              {$m_0$};
      \node[system]                             (m1)[right of=m0] {$m_1$};
    
      \path (m0) edge[loop left] node[align=center] {$v_0 \mid *$\\$v_3 \mid *$\\$v_1 \mid v_1$\\$v_1 \mid v_3$\\$v_2 \mid v_3$} (m0)
                 edge             node[align=center] {$v_3 \mid *$}   (m1)
            (m1) edge[loop right] node[align=center] {$v_0 \mid *$\\$v_3 \mid *$\\$v_1 \mid v_2$\\$v_2 \mid v_2$}   (m1);
    \end{tikzpicture}
    \caption{A nondeterministic Mealy machine of player~$0$. The notation $v \mid v'$ on the transitions $(m,m')$ indicates that $m' \in \delta(m,v)$, and if $v \in V_0$, that $v' \in \tau(m,v)$, otherwise $v' = *$.}
    \label{fig:ex-arena-mealy}
  \end{minipage}
\end{figure}
\end{example}

We now formally define the notion of product arena. Let $\arena{}=(V,E,\Players,(V_i)_{i\in\Players},(w_i)_{i\in\Players})$ be a weighted arena and $\machine{j}=(M,m_0,\delta,\tau)$ be a (nondeterministic) Mealy machine for player~$j\in\Players$. Then, the \emph{product arena} $\arena{}\times \machine{j}$ is the weighted arena $\arena{}\times \machine{j}=(V',E',\Players,(V'_i)_{i\in\Players},(w'_i)_{i\in\Players})$ where
\begin{itemize}
    \item $V'= (V \times M) \cup (V \times V \times M)$,
    \item $V'_i = V_i\times M$ for all $i \in \Players \ssetminus \{j\}$, and $V'_j = (V_j \times M) \cup (V \times V \times M)$,
    \item $E'$ is the set of edges defined as
    \begin{itemize}
        \item $(v,m) \rightarrow (v,v',m)$ if $(v,v') \in E$, and when $v \in V_j$, it must hold that $v' \in \tau(m,v)$,
        \item $(v,v',m) \rightarrow (v',m')$ if $m' \in \delta(m,v)$,
    \end{itemize}
    \item For the edges $e' \in E'$ of the form $(v,m) \rightarrow (v,v',m)$, $w'_i(e') = w_i((v,v'))$, while for the edges $e'$ of the form $(v,v',m) \rightarrow (v',m')$, $w'_i(e') = 0$, for all players~$i \in \Players$.
\end{itemize}
Intuitively, in vertices $(v,v',m)$, it is player~$j$ who decides how to update the memory state $m$ according to $\delta$.

When $\arena$ is initialized with $v_0$ as initial vertex, then the product arena is also initialized with $(v_0,m_0)$ as initial vertex. Given a reachability game $\game = (\arena, (\reach{i})_{i \in \Players})$, we also define the \emph{product game} $\game \times \machine{j}$ as the reachability game $(\arena \times \machine{j}, (\reach{i}')_{i \in \Players})$ such that $\reach{i}' = \reach{i} \times M$ for all $i \in \Players$.

Back to \cref{ex:mealy-product}, the product arena $\arena{}'=\arena{}\times\machine{0}$ is depicted in \cref{fig:ex-arena-product}. 
We can see that player~$0$ has several strategies $\strategyfor{0}\in\strategiesmachine{0}$ whose behavior changes according to the memory state $m_0$ or $m_1$.

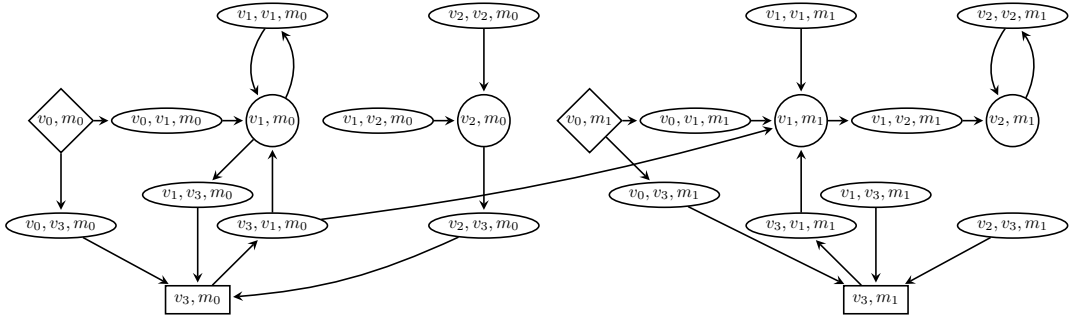
\begin{figure}
\centering
\begin{tikzpicture}[automaton,scale=0.65,every node/.style={scale=0.65},node distance=2.14cm]
  \node[environment2,inner sep=1pt] (v0m0)                {$v_0,m_0$};

  \node[system2] (v0v1m0)[right of=v0m0]  {$v_0,v_1,m_0$};

  \node[system,inner sep=1pt]       (v1m0)[right of=v0v1m0] {$v_1,m_0$};

  \node[system2] (v0v3m0)[below of=v0m0]  {$v_0,v_3,m_0$};
  \node[system2] (v1v1m0)[above of=v1m0]  {$v_1,v_1,m_0$};
  \node[system2] (v1v2m0)[right of=v1m0]  {$v_1,v_2,m_0$};
  \node[system2] (v1v3m0)[below left of=v1m0]  {$v_1,v_3,m_0$};

  \node[system,inner sep=1pt]       (v2m0)[right of=v1v2m0] {$v_2,m_0$};
  \node[environment,inner sep=5pt]  (v3m0)[below of=v1v3m0] {$v_3,m_0$};

  \node[system2] (v2v2m0)[above of=v2m0]  {$v_2,v_2,m_0$};
  \node[system2] (v2v3m0)[below of=v2m0]  {$v_2,v_3,m_0$};
  \node[system2] (v3v1m0)[above right of=v3m0]  {$v_3,v_1,m_0$};

  \node[environment2,inner sep=1pt] (v0m1)[right of=v2m0] {$v_0,m_1$};

  \node[system2] (v0v1m1)[right of=v0m1]  {$v_0,v_1,m_1$};

  \node[system,inner sep=1pt]       (v1m1)[right of=v0v1m1] {$v_1,m_1$};

  \node[system2] (v0v3m1)[below right of=v0m1]  {$v_0,v_3,m_1$};
  \node[system2] (v1v1m1)[above of=v1m1]  {$v_1,v_1,m_1$};

  \node[system2] (v1v2m1)[right of=v1m1]  {$v_1,v_2,m_1$};
  \node[system2] (v1v3m1)[below right of=v1m1]  {$v_1,v_3,m_1$};

  \node[system,inner sep=1pt]       (v2m1)[right of=v1v2m1] {$v_2,m_1$};
  \node[environment,inner sep=5pt]  (v3m1)[below of=v1v3m1] {$v_3,m_1$};
  \node[system2] (v2v2m1)[above of=v2m1]  {$v_2,v_2,m_1$};
  \node[system2] (v2v3m1)[below of=v2m1]  {$v_2,v_3,m_1$};
  \node[system2] (v3v1m1)[above left of=v3m1]  {$v_3,v_1,m_1$};

  \path (v0m0) edge             (v0v1m0)
               edge             (v0v3m0)
        (v1m0) edge[bend right] (v1v1m0)
               edge             (v1v3m0)
        (v2m0) edge             (v2v3m0)
        (v3m0) edge             (v3v1m0)

        (v0m1) edge             (v0v1m1)
               edge             (v0v3m1)
        (v1m1) edge             (v1v2m1)
        (v2m1) edge[bend right] (v2v2m1)
        (v3m1) edge             (v3v1m1)

        (v0v1m0) edge (v1m0)
        (v0v3m0) edge (v3m0)
        (v1v1m0) edge[bend right] (v1m0)
        (v1v2m0) edge (v2m0)
        (v1v3m0) edge (v3m0)
        (v2v2m0) edge (v2m0)
        (v2v3m0) edge[bend left=10] (v3m0)
        (v3v1m0) edge (v1m0)
                 edge[bend right=4] (v1m1)

        (v0v1m1) edge (v1m1)
        (v0v3m1) edge (v3m1)
        (v1v1m1) edge (v1m1)
        (v1v2m1) edge (v2m1)
        (v1v3m1) edge (v3m1)
        (v2v2m1) edge[bend right] (v2m1)
        (v2v3m1) edge (v3m1)
        (v3v1m1) edge (v1m1);

\end{tikzpicture}
\caption{The product arena of the arena in \cref{fig:ex-arena} and the Mealy machine in \cref{fig:ex-arena-mealy}.}
\label{fig:ex-arena-product}
\end{figure}

\section{Undecidability} \label{appendix:Undecidable}

This section aims at proving the undecidability results described in \cref{theorem:undecidableNE} for the \prob{NCNS} and \prob{NCPS} problems.

\undecidabletheorem*

For that purpose, we use a reduction from zero-sum multidimensional
shortest path games, that we recall here. We consider an arena $\arena = (V,E,\{\EE{},\AA{}\},V_\E,V_\A,(w_i)_{i \in \{1,\dots,t\}})$ with two players, \EE{} and \AA{} and $k$ weight functions such that for all $i \in \{1,\dots,t\}$, $w_i: E \rightarrow \Z$. We then consider a zero-sum $\game = (\arena,\Omega)$ whose objective $\Omega$ (for \EE{}) is to reach a target set\footnote{This target set is unique, contrarily to the particular games introduced in~\cref{subsec:particular}.} $\reach{} \subseteq V$. Such a game is called a \emph{zero-sum multidimensional
shortest path game}~\cite{Randour-games-on-graphs}. For a play $\pi = \pi_0\pi_1 \ldots$, we define $\cost{}{\pi} = (+\infty)^t$ if $\pi$ does not visit $\reach{}$, and $\cost{}{\pi} = (\cost{1}{\pi},\dots,\cost{t}{\pi})$, where $\cost{i}{\pi} = \sum_{k=1}^{\ell} w_i(\pi_{k-1},\pi_{k})$ and $\ell = \inf\{n \in \N \mid \pi_n \in \reach{}\}$ otherwise. The \emph{Shortest Path problem} (\prob{SP} problem) (resp. \emph{Equality Shortest Path problem} (\prob{ESP} problem)) asks, given such a multidimensional shortest path game, an initial vertex $v_0$, and a threshold $\bar{c} \in \Z^k$, whether there exists a strategy $\strategyfor{\E}$ of \EE{} such that for every strategy $\strategyfor{\A}$ of \AA{}, $\cost{}{\outcomefrom{\strategyfor{\E},\strategyfor{\A}}{v_0}} \leq \bar{c}$ (resp. $\cost{}{\outcomefrom{\strategyfor{\E},\strategyfor{\A}}{v_0}} = \bar{c}$).\footnote{Recall that we consider here the componentwise partial order over ($\Z \cup \{+\infty\})^t$.} The \prob{SP} problem is \pspaceComplete{}~\cite{DBLP:conf/rp/BriyaheGoeminne23} with weight functions $E \rightarrow \N$ and thresholds $\bar{c} \in \N^k$. It is in general undecidable~\cite{Randour-games-on-graphs,DBLP:journals/fmsd/RandourRS17}, even with $\bar{c} = (0,\dots,0)$, 
by reduction from the halting problem of two-counter machines, known to be undecidable~\cite{minsky1961recursive}. We additionally show that the \prob{ESP} problem is undecidable by a reduction from the \prob{SP} problem.

\begin{figure}
    \centering
    \begin{tikzpicture}[x=0.75pt,y=0.75pt,yscale=-1,xscale=1]
    \draw   (225.33,81.16) .. controls (225.33,78.73) and (227.3,76.77) .. (229.73,76.77) -- (237.3,76.77) .. controls (239.73,76.77) and (241.69,78.73) .. (241.69,81.16) -- (241.69,88.25) .. controls (241.69,90.68) and (239.73,92.65) .. (237.3,92.65) -- (229.73,92.65) .. controls (227.3,92.65) and (225.33,90.68) .. (225.33,88.25) -- cycle ;
    \draw  [fill={rgb, 255:red, 155; green, 155; blue, 155 }  ,fill opacity=0.1 ][dash pattern={on 4.5pt off 4.5pt}][line width=0.75]  (164.97,80.11) .. controls (164.97,72.11) and (171.45,65.63) .. (179.44,65.63) -- (285.19,65.63) .. controls (293.19,65.63) and (299.67,72.11) .. (299.67,80.11) -- (299.67,123.53) .. controls (299.67,131.52) and (293.19,138) .. (285.19,138) -- (179.44,138) .. controls (171.45,138) and (164.97,131.52) .. (164.97,123.53) -- cycle ;
    \draw    (241.69,81.16) .. controls (265.84,75.67) and (296.49,83.15) .. (314.06,100.57) ;
    \draw [shift={(315.92,102.51)}, rotate = 227.54] [fill={rgb, 255:red, 0; green, 0; blue, 0 }  ][line width=0.08]  [draw opacity=0] (5.36,-2.57) -- (0,0) -- (5.36,2.57) -- (3.56,0) -- cycle    ;
    \draw   (315,108.77) .. controls (315,104.39) and (318.55,100.83) .. (322.94,100.83) -- (323.42,100.83) .. controls (327.81,100.83) and (331.36,104.39) .. (331.36,108.77) -- (331.36,108.77) .. controls (331.36,113.16) and (327.81,116.71) .. (323.42,116.71) -- (322.94,116.71) .. controls (318.55,116.71) and (315,113.16) .. (315,108.77) -- cycle ;
    \draw   (242,118.16) .. controls (242,115.73) and (243.97,113.77) .. (246.39,113.77) -- (253.97,113.77) .. controls (256.39,113.77) and (258.36,115.73) .. (258.36,118.16) -- (258.36,125.25) .. controls (258.36,127.68) and (256.39,129.65) .. (253.97,129.65) -- (246.39,129.65) .. controls (243.97,129.65) and (242,127.68) .. (242,125.25) -- cycle ;
    \draw    (258.36,125.25) .. controls (275.99,126.46) and (294.75,127.98) .. (314.28,117.13) ;
    \draw [shift={(316.72,115.71)}, rotate = 148.71] [fill={rgb, 255:red, 0; green, 0; blue, 0 }  ][line width=0.08]  [draw opacity=0] (5.36,-2.57) -- (0,0) -- (5.36,2.57) -- (3.56,0) -- cycle    ;
    \draw    (235.06,92.65) -- (233.4,107.25) ;
    \draw [shift={(233.07,110.23)}, rotate = 276.46] [fill={rgb, 255:red, 0; green, 0; blue, 0 }  ][line width=0.08]  [draw opacity=0] (5.36,-2.57) -- (0,0) -- (5.36,2.57) -- (3.56,0) -- cycle    ;
    \draw    (253.97,113.77) -- (260.78,99.35) ;
    \draw [shift={(262.07,96.63)}, rotate = 115.3] [fill={rgb, 255:red, 0; green, 0; blue, 0 }  ][line width=0.08]  [draw opacity=0] (5.36,-2.57) -- (0,0) -- (5.36,2.57) -- (3.56,0) -- cycle    ;
    \draw  [color={rgb, 255:red, 211; green, 33; blue, 33 }  ,draw opacity=1 ][line width=0.75] [line join = round][line cap = round] (261.67,103.92) .. controls (258.1,104.52) and (256.35,106.13) .. (254.45,108.61) ;
    \draw  [color={rgb, 255:red, 211; green, 33; blue, 33 }  ,draw opacity=1 ][line width=0.75] [line join = round][line cap = round] (254.45,103.55) .. controls (256.43,105.55) and (258.93,107.33) .. (260.09,109.73) ;
    \draw  [color={rgb, 255:red, 211; green, 33; blue, 33 }  ,draw opacity=1 ][line width=0.75] [line join = round][line cap = round] (237.95,102.47) .. controls (234.96,101.94) and (233.49,100.51) .. (231.9,98.3) ;
    \draw  [color={rgb, 255:red, 211; green, 33; blue, 33 }  ,draw opacity=1 ][line width=0.75] [line join = round][line cap = round] (231.9,102.8) .. controls (233.56,101.02) and (235.66,99.44) .. (236.63,97.3) ;
    \draw    (319.52,101.58) .. controls (317.61,79.32) and (328.83,78.54) .. (327.3,98.32) ;
    \draw [shift={(327,101.25)}, rotate = 277.32] [fill={rgb, 255:red, 0; green, 0; blue, 0 }  ][line width=0.08]  [draw opacity=0] (5.36,-2.57) -- (0,0) -- (5.36,2.57) -- (3.56,0) -- cycle    ;
    \draw    (331.92,109.95) -- (340.92,110.07) ;
    \draw [shift={(343.92,110.11)}, rotate = 180.76] [fill={rgb, 255:red, 0; green, 0; blue, 0 }  ][line width=0.08]  [draw opacity=0] (5.36,-2.57) -- (0,0) -- (5.36,2.57) -- (3.56,0) -- cycle    ;
    \draw    (357.52,110.11) -- (366.52,109.99) ;
    \draw [shift={(369.52,109.95)}, rotate = 179.24] [fill={rgb, 255:red, 0; green, 0; blue, 0 }  ][line width=0.08]  [draw opacity=0] (5.36,-2.57) -- (0,0) -- (5.36,2.57) -- (3.56,0) -- cycle    ;
    \draw   (369.33,109.24) .. controls (369.33,104.85) and (372.89,101.3) .. (377.27,101.3) -- (377.75,101.3) .. controls (382.14,101.3) and (385.69,104.85) .. (385.69,109.24) -- (385.69,109.24) .. controls (385.69,113.63) and (382.14,117.18) .. (377.75,117.18) -- (377.27,117.18) .. controls (372.89,117.18) and (369.33,113.63) .. (369.33,109.24) -- cycle ;
    \draw   (425.4,107.57) .. controls (425.4,103.19) and (428.95,99.63) .. (433.34,99.63) -- (433.82,99.63) .. controls (438.21,99.63) and (441.76,103.19) .. (441.76,107.57) -- (441.76,107.57) .. controls (441.76,111.96) and (438.21,115.51) .. (433.82,115.51) -- (433.34,115.51) .. controls (428.95,115.51) and (425.4,111.96) .. (425.4,107.57) -- cycle ;
    \draw    (385.52,108.75) -- (422.12,108.75) ;
    \draw [shift={(425.12,108.75)}, rotate = 180] [fill={rgb, 255:red, 0; green, 0; blue, 0 }  ][line width=0.08]  [draw opacity=0] (5.36,-2.57) -- (0,0) -- (5.36,2.57) -- (3.56,0) -- cycle    ;
    \draw    (373.52,101.58) .. controls (371.61,79.32) and (382.83,78.54) .. (381.3,98.32) ;
    \draw [shift={(381,101.25)}, rotate = 277.32] [fill={rgb, 255:red, 0; green, 0; blue, 0 }  ][line width=0.08]  [draw opacity=0] (5.36,-2.57) -- (0,0) -- (5.36,2.57) -- (3.56,0) -- cycle    ;
    \draw    (429.52,99.18) .. controls (427.61,76.92) and (438.83,76.14) .. (437.3,95.92) ;
    \draw [shift={(437,98.85)}, rotate = 277.32] [fill={rgb, 255:red, 0; green, 0; blue, 0 }  ][line width=0.08]  [draw opacity=0] (5.36,-2.57) -- (0,0) -- (5.36,2.57) -- (3.56,0) -- cycle    ;
    
    \draw (228.13,79.63) node [anchor=north west][inner sep=0.75pt]  [font=\footnotesize] [align=left] {$T$};
    \draw (169.85,70.6) node [anchor=north west][inner sep=0.75pt]  [font=\normalsize] [align=left] {$\mathcal{H}$};
    \draw (315,104.17) node [anchor=north west][inner sep=0.75pt]  [font=\scriptsize] [align=left] {$D_{1}$};
    \draw (301.71,74.56) node [anchor=north west][inner sep=0.75pt]  [font=\tiny,rotate=-0.09] [align=left] {$(1,0,... ,0)$};
    \draw (244.72,117.13) node [anchor=north west][inner sep=0.75pt]  [font=\footnotesize] [align=left] {$T$};
    \draw (344.33,108.53) node [anchor=north west][inner sep=0.75pt]  [font=\scriptsize] [align=left] {$...$};
    \draw (369.33,104.7) node [anchor=north west][inner sep=0.75pt]  [font=\scriptsize] [align=left] {$D_{t}$};
    \draw (425.4,101.83) node [anchor=north west][inner sep=0.75pt]  [font=\scriptsize] [align=left] {$v_{T} '$};
    \draw (386.12,96.32) node [anchor=north west][inner sep=0.75pt]  [font=\tiny,rotate=-0.19] [align=left] {$(0,... ,0)$};
    \draw (358.91,74.16) node [anchor=north west][inner sep=0.75pt]  [font=\tiny,rotate=-0.09] [align=left] {$(0,... ,0,1)$};
    \draw (261.12,115.32) node [anchor=north west][inner sep=0.75pt]  [font=\tiny,rotate=-0.19] [align=left] {$(0,... ,0)$};
    \draw (262.96,80.4) node [anchor=north west][inner sep=0.75pt]  [font=\tiny,rotate=-13.08] [align=left] {$(0,... ,0)$};
    \end{tikzpicture}
    \caption{Reduction from the \prob{SP} problem to the \prob{ESP} problem.}
    \label{fig:reduction-sp-esp-undecidable}
\end{figure}
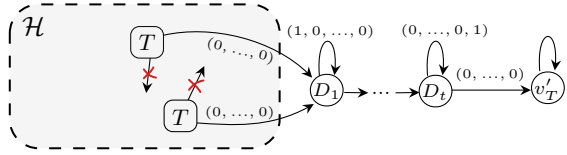

\begin{theorem}\label{theorem:equality-shortest-path-integer-undecidable}
    The Equality Shortest Path problem is undecidable.
\end{theorem}

\begin{proof}
    We show the undecidability of the \prob{ESP} problem by reduction from the \prob{SP} problem with the threshold $\bar c = \bar 0$. Let $\mathcal{H} = (\arena,\Omega)$ be a multidimensional shortest path game, with the arena $\arena = (V,E,\{\EE{},\AA{}\},V_\E,V_\A,(w_i)_{i \in \{1,\dots,t\}})$, the target set $\reach{}$, the initial vertex $v_0 \in V$, and the threshold $\bar{0}$. We construct a new shortest path game $\mathcal{H}'= (\arena',\Omega')$ with the arena $\arena' = (V',E',\{\EE{},\AA{}\},V'_\E,V'_\A,(w'_i)_{i \in \{1,\dots,t\}})$ illustrated in \cref{fig:reduction-sp-esp-undecidable}, the target set $\reach{}'$, the same initial vertex $v_0 \in V$, and the same threshold $\bar{c}$, where
    \begin{itemize}
        \item $V' = V \cup \{D_1,\dots,D_t,v_{T'}\}$,
        \item $V'_\E = V_\E \cup \{D_1,\dots,D_t,v_{T'}\}$, $V'_\A = V_\A$,
        \item the set $E'$ is composed of the edges of $E$, except that $\succc{v} = \{D_1\}$ for every vertex $v \in \reach{}$; $E'$ also contains the new edges $(D_{i},D_{i+1})$ for $i \in \{1,\dots,t-1\}$, $(D_{i},D_{i})$ for $i \in \{1,\dots,t\}$, $(D_{t},v_{T'})$, and $(v_{T'},v_{T'})$,
        \item for every $i \in \{1,\dots,t\}$, $w'_i(e) = w_i(e)$ if $e \in E$, and $w'_i$ assigns $0$ to each added edge, except $w_i((D_i,D_i)) = 1$,
        \item $\reach{}' = \{v_{T'}\}$.
    \end{itemize}
    Let us show that \EE{} is winning for the \prob{SP} problem in $\mathcal{H}$ with the threshold $\bar 0$ if and only if she is winning for the \prob{ESP} problem in $\mathcal{H}'$ with the same threshold.

    Suppose that \EE{} has a winning strategy $\strategyfor{\E}$ in $\mathcal{H}$ for the \prob{SP} problem. Thus, every play consistent with $\strategyfor{\E}$ has a prefix $h$ ending in $\reach{}$, with $\cost{}{h} \leq \bar 0$. We define a strategy for \EE{} in $\mathcal{H}'$ by simulating $\sigma_\E$ in $\mathcal{H}$ and by extending it after each such history $h$ in a way that the resulting consistent play $\pi'$ satisfies $\cost{}{\pi'} = \bar 0$ (for each $i \in \{1,\dots,t\}$, loop on $D_i$ exactly $- \cost{i}{h}$ times).

    Conversely, suppose that \EE{} has a winning strategy $\strategyfor{\E}'$  to win in $\mathcal{H}'$, and let us simulate this strategy in $\mathcal{H}$ for histories both in $\mathcal{H}$ and $\mathcal{H}'$. By hypothesis, every play $\pi'$ consistent with $\strategyfor{\E}'$ is such that $\cost{}{\pi'} = \bar 0$. So, by construction, there is a prefix $h'$ of $\pi'$ ending in $\reach{}$ such that, in $\mathcal{H}$, $\cost{}{h'} \leq \bar 0$ holds. Once $T$ is reached by $h'$, we extend the new strategy arbitrarily in $\mathcal H$.
\end{proof}

\begin{figure}[htbp]
  \begin{minipage}[b]{0.46\textwidth}
    \centering
    \resizebox{0.95\textwidth}{!}{%
    \begin{tikzpicture}[x=0.75pt,y=0.75pt,yscale=-1,xscale=1]
    \draw   (118.96,67.18) -- (129.48,77.37) -- (119.28,87.89) -- (108.76,77.7) -- cycle ;
    \draw    (225.94,82.89) -- (191,110.78) ;
    \draw [shift={(188.65,112.65)}, rotate = 321.41] [fill={rgb, 255:red, 0; green, 0; blue, 0 }  ][line width=0.08]  [draw opacity=0] (5.36,-2.57) -- (0,0) -- (5.36,2.57) -- (3.56,0) -- cycle    ;
    \draw  [fill={rgb, 255:red, 155; green, 155; blue, 155 }  ,fill opacity=0.1 ][dash pattern={on 4.5pt off 4.5pt}][line width=0.75]  (243.22,73.38) .. controls (243.22,64.93) and (250.06,58.08) .. (258.51,58.08) -- (370.59,58.08) .. controls (379.04,58.08) and (385.88,64.93) .. (385.88,73.38) -- (385.88,119.26) .. controls (385.88,127.7) and (379.04,134.55) .. (370.59,134.55) -- (258.51,134.55) .. controls (250.06,134.55) and (243.22,127.7) .. (243.22,119.26) -- cycle ;
    \draw    (177.72,123.87) .. controls (172.66,143.38) and (187.02,145.78) .. (185.27,126.01) ;
    \draw [shift={(184.92,123.07)}, rotate = 81.31] [fill={rgb, 255:red, 0; green, 0; blue, 0 }  ][line width=0.08]  [draw opacity=0] (5.36,-2.57) -- (0,0) -- (5.36,2.57) -- (3.56,0) -- cycle    ;
    \draw    (96.28,77.61) -- (105.76,77.67) ;
    \draw [shift={(108.76,77.7)}, rotate = 180.41] [fill={rgb, 255:red, 0; green, 0; blue, 0 }  ][line width=0.08]  [draw opacity=0] (5.36,-2.57) -- (0,0) -- (5.36,2.57) -- (3.56,0) -- cycle    ;
    \draw   (181.07,106.16) -- (191.59,116.35) -- (181.39,126.87) -- (170.87,116.68) -- cycle ;
    \draw   (170.29,66.59) -- (180.81,76.79) -- (170.61,87.31) -- (160.09,77.11) -- cycle ;
    \draw    (129.48,77.37) -- (157.09,77.14) ;
    \draw [shift={(160.09,77.11)}, rotate = 179.51] [fill={rgb, 255:red, 0; green, 0; blue, 0 }  ][line width=0.08]  [draw opacity=0] (5.36,-2.57) -- (0,0) -- (5.36,2.57) -- (3.56,0) -- cycle    ;
    \draw    (180.81,76.79) -- (192.65,76.92) ;
    \draw [shift={(195.65,76.95)}, rotate = 180.62] [fill={rgb, 255:red, 0; green, 0; blue, 0 }  ][line width=0.08]  [draw opacity=0] (5.36,-2.57) -- (0,0) -- (5.36,2.57) -- (3.56,0) -- cycle    ;
    \draw   (230.62,66.84) -- (241.14,77.04) -- (230.94,87.56) -- (220.42,77.36) -- cycle ;
    \draw    (208.65,77.2) -- (217.42,77.32) ;
    \draw [shift={(220.42,77.36)}, rotate = 180.79] [fill={rgb, 255:red, 0; green, 0; blue, 0 }  ][line width=0.08]  [draw opacity=0] (5.36,-2.57) -- (0,0) -- (5.36,2.57) -- (3.56,0) -- cycle    ;
    \draw    (241.14,77.04) -- (285.4,76.72) ;
    \draw [shift={(288.4,76.7)}, rotate = 179.59] [fill={rgb, 255:red, 0; green, 0; blue, 0 }  ][line width=0.08]  [draw opacity=0] (5.36,-2.57) -- (0,0) -- (5.36,2.57) -- (3.56,0) -- cycle    ;
    \draw    (172.15,86.15) -- (178.59,103.14) ;
    \draw [shift={(179.65,105.95)}, rotate = 249.25] [fill={rgb, 255:red, 0; green, 0; blue, 0 }  ][line width=0.08]  [draw opacity=0] (5.36,-2.57) -- (0,0) -- (5.36,2.57) -- (3.56,0) -- cycle    ;
    \draw    (123.65,82.65) -- (173.21,111.64) ;
    \draw [shift={(175.8,113.15)}, rotate = 210.32] [fill={rgb, 255:red, 0; green, 0; blue, 0 }  ][line width=0.08]  [draw opacity=0] (5.36,-2.57) -- (0,0) -- (5.36,2.57) -- (3.56,0) -- cycle    ;
    \draw   (289.5,74.36) .. controls (289.5,71.93) and (291.47,69.97) .. (293.89,69.97) -- (301.47,69.97) .. controls (303.89,69.97) and (305.86,71.93) .. (305.86,74.36) -- (305.86,81.45) .. controls (305.86,83.88) and (303.89,85.85) .. (301.47,85.85) -- (293.89,85.85) .. controls (291.47,85.85) and (289.5,83.88) .. (289.5,81.45) -- cycle ;

    \draw (112,70.87) node [anchor=north west][inner sep=0.75pt]  [font=\scriptsize] [align=left] {$v'_{1}$};
    \draw (175.42,112.32) node [anchor=north west][inner sep=0.75pt]  [font=\scriptsize] [align=left] {$D$};
    \draw (359.85,67.22) node [anchor=north west][inner sep=0.75pt]  [font=\normalsize] [align=left] {$\mathcal{H}$};
    \draw (163.33,70.62) node [anchor=north west][inner sep=0.75pt]  [font=\scriptsize] [align=left] {$v'_{2}$};
    \draw (194.73,75.33) node [anchor=north west][inner sep=0.75pt]  [font=\scriptsize] [align=left] {$...$};
    \draw (224.33,70.53) node [anchor=north west][inner sep=0.75pt]  [font=\scriptsize] [align=left] {$v'_{t}$};
    \draw (127.72,89.74) node [anchor=north west][inner sep=0.75pt]  [font=\tiny,rotate=-29.06] [align=left] {$(0,... ,0)$};
    \draw (254.25,104.58) node [anchor=north west][inner sep=0.75pt]  [font=\footnotesize] [align=left] {Player~$0$ vs player~$t+1$};
    \draw (163.07,138.88) node [anchor=north west][inner sep=0.75pt]  [font=\tiny] [align=left] {$(0,... ,0)$};
    \draw (193.58,112.39) node [anchor=north west][inner sep=0.75pt]  [font=\tiny,rotate=-321.61] [align=left] {$(0,... ,0)$};
    \draw (246.15,64.47) node [anchor=north west][inner sep=0.75pt]  [font=\tiny] [align=left] {$(0,...,0)$};
    \draw (123.98,63.38) node [anchor=north west][inner sep=0.75pt]  [font=\tiny] [align=left] {$(0,... ,0)$};
    \draw (290.88,73.83) node [anchor=north west][inner sep=0.75pt]  [font=\footnotesize] [align=left] {$v_{0}$};
    \end{tikzpicture}
    }
    \caption{Reduction from the \prob{SP} problem to the \prob{NCNS} problem.}
    \label{fig:undecidability-ncns}
  \end{minipage}
  \hfill
  \begin{minipage}[b]{0.53\textwidth}
    \centering
    \resizebox{0.95\textwidth}{!}{%
    \begin{tikzpicture}[x=0.75pt,y=0.75pt,yscale=-1,xscale=1]
    \draw    (258.61,140.51) -- (222.44,162.14) ;
    \draw [shift={(219.86,163.68)}, rotate = 329.12] [fill={rgb, 255:red, 0; green, 0; blue, 0 }  ][line width=0.08]  [draw opacity=0] (5.36,-2.57) -- (0,0) -- (5.36,2.57) -- (3.56,0) -- cycle    ;
    \draw  [fill={rgb, 255:red, 155; green, 155; blue, 155 }  ,fill opacity=0.1 ][dash pattern={on 4.5pt off 4.5pt}][line width=0.75]  (280.88,76.69) .. controls (280.88,67.89) and (288.02,60.75) .. (296.83,60.75) -- (437.72,60.75) .. controls (446.53,60.75) and (453.67,67.89) .. (453.67,76.69) -- (453.67,124.52) .. controls (453.67,133.33) and (446.53,140.47) .. (437.72,140.47) -- (296.83,140.47) .. controls (288.02,140.47) and (280.88,133.33) .. (280.88,124.52) -- cycle ;
    \draw    (219.67,82.6) -- (287.67,82.79) ;
    \draw [shift={(290.67,82.8)}, rotate = 180.16] [fill={rgb, 255:red, 0; green, 0; blue, 0 }  ][line width=0.08]  [draw opacity=0] (5.36,-2.57) -- (0,0) -- (5.36,2.57) -- (3.56,0) -- cycle    ;
    \draw   (290.5,79.36) .. controls (290.5,76.93) and (292.47,74.97) .. (294.89,74.97) -- (302.47,74.97) .. controls (304.89,74.97) and (306.86,76.93) .. (306.86,79.36) -- (306.86,86.45) .. controls (306.86,88.88) and (304.89,90.85) .. (302.47,90.85) -- (294.89,90.85) .. controls (292.47,90.85) and (290.5,88.88) .. (290.5,86.45) -- cycle ;
    \draw   (293.17,101.76) .. controls (293.17,99.57) and (294.94,97.8) .. (297.12,97.8) -- (303.51,97.8) .. controls (305.7,97.8) and (307.47,99.57) .. (307.47,101.76) -- (307.47,108.22) .. controls (307.47,110.41) and (305.7,112.18) .. (303.51,112.18) -- (297.12,112.18) .. controls (294.94,112.18) and (293.17,110.41) .. (293.17,108.22) -- cycle ;
    \draw    (307.48,105.04) -- (407.67,104.94) ;
    \draw [shift={(410.67,104.93)}, rotate = 179.94] [fill={rgb, 255:red, 0; green, 0; blue, 0 }  ][line width=0.08]  [draw opacity=0] (5.36,-2.57) -- (0,0) -- (5.36,2.57) -- (3.56,0) -- cycle    ;
    \draw   (411.5,102.09) .. controls (411.5,99.9) and (413.27,98.13) .. (415.46,98.13) -- (421.84,98.13) .. controls (424.03,98.13) and (425.8,99.9) .. (425.8,102.09) -- (425.8,108.56) .. controls (425.8,110.74) and (424.03,112.51) .. (421.84,112.51) -- (415.46,112.51) .. controls (413.27,112.51) and (411.5,110.74) .. (411.5,108.56) -- cycle ;
    \draw    (213,64.1) -- (212.97,73.01) ;
    \draw [shift={(212.96,76.01)}, rotate = 270.21] [fill={rgb, 255:red, 0; green, 0; blue, 0 }  ][line width=0.08]  [draw opacity=0] (5.36,-2.57) -- (0,0) -- (5.36,2.57) -- (3.56,0) -- cycle    ;
    \draw    (149.85,140.12) .. controls (126.17,149.08) and (135.1,159.7) .. (151.87,145.73) ;
    \draw [shift={(154,143.85)}, rotate = 137.25] [fill={rgb, 255:red, 0; green, 0; blue, 0 }  ][line width=0.08]  [draw opacity=0] (5.36,-2.57) -- (0,0) -- (5.36,2.57) -- (3.56,0) -- cycle    ;
    \draw    (164.18,143.25) -- (202.67,162.35) ;
    \draw [shift={(205.36,163.68)}, rotate = 206.4] [fill={rgb, 255:red, 0; green, 0; blue, 0 }  ][line width=0.08]  [draw opacity=0] (5.36,-2.57) -- (0,0) -- (5.36,2.57) -- (3.56,0) -- cycle    ;
    \draw    (273.35,137.45) .. controls (295.63,148.91) and (282.35,157.63) .. (270.36,143.3) ;
    \draw [shift={(268.68,141.12)}, rotate = 54.57] [fill={rgb, 255:red, 0; green, 0; blue, 0 }  ][line width=0.08]  [draw opacity=0] (5.36,-2.57) -- (0,0) -- (5.36,2.57) -- (3.56,0) -- cycle    ;
    \draw    (205.36,90.51) -- (166.38,126.63) ;
    \draw [shift={(164.18,128.67)}, rotate = 317.18] [fill={rgb, 255:red, 0; green, 0; blue, 0 }  ][line width=0.08]  [draw opacity=0] (5.36,-2.57) -- (0,0) -- (5.36,2.57) -- (3.56,0) -- cycle    ;
    \draw    (219.86,90.51) -- (256.4,123.91) ;
    \draw [shift={(258.61,125.93)}, rotate = 222.43] [fill={rgb, 255:red, 0; green, 0; blue, 0 }  ][line width=0.08]  [draw opacity=0] (5.36,-2.57) -- (0,0) -- (5.36,2.57) -- (3.56,0) -- cycle    ;
    \draw    (212.75,90.52) -- (212.52,122.6) ;
    \draw [shift={(212.5,125.6)}, rotate = 270.41] [fill={rgb, 255:red, 0; green, 0; blue, 0 }  ][line width=0.08]  [draw opacity=0] (5.36,-2.57) -- (0,0) -- (5.36,2.57) -- (3.56,0) -- cycle    ;
    \draw    (204.89,169.4) .. controls (181.98,166.07) and (180.5,176.31) .. (201.95,175.95) ;
    \draw [shift={(204.75,175.85)}, rotate = 176.86] [fill={rgb, 255:red, 0; green, 0; blue, 0 }  ][line width=0.08]  [draw opacity=0] (5.36,-2.57) -- (0,0) -- (5.36,2.57) -- (3.56,0) -- cycle    ;
    \draw   (205.36,75.93) .. controls (205.36,75.93) and (205.36,75.93) .. (205.36,75.93) -- (219.86,75.93) .. controls (219.86,75.93) and (219.86,75.93) .. (219.86,75.93) -- (219.86,90.51) .. controls (219.86,90.51) and (219.86,90.51) .. (219.86,90.51) -- (205.36,90.51) .. controls (205.36,90.51) and (205.36,90.51) .. (205.36,90.51) -- cycle ;
    \draw   (258.61,125.93) .. controls (258.61,125.93) and (258.61,125.93) .. (258.61,125.93) -- (273.11,125.93) .. controls (273.11,125.93) and (273.11,125.93) .. (273.11,125.93) -- (273.11,140.51) .. controls (273.11,140.51) and (273.11,140.51) .. (273.11,140.51) -- (258.61,140.51) .. controls (258.61,140.51) and (258.61,140.51) .. (258.61,140.51) -- cycle ;
    \draw   (149.68,128.67) .. controls (149.68,128.67) and (149.68,128.67) .. (149.68,128.67) -- (164.18,128.67) .. controls (164.18,128.67) and (164.18,128.67) .. (164.18,128.67) -- (164.18,143.25) .. controls (164.18,143.25) and (164.18,143.25) .. (164.18,143.25) -- (149.68,143.25) .. controls (149.68,143.25) and (149.68,143.25) .. (149.68,143.25) -- cycle ;
    \draw   (205.36,163.68) .. controls (205.36,163.68) and (205.36,163.68) .. (205.36,163.68) -- (219.86,163.68) .. controls (219.86,163.68) and (219.86,163.68) .. (219.86,163.68) -- (219.86,178.26) .. controls (219.86,178.26) and (219.86,178.26) .. (219.86,178.26) -- (205.36,178.26) .. controls (205.36,178.26) and (205.36,178.26) .. (205.36,178.26) -- cycle ;
    
    \draw (150.08,129.12) node [anchor=north west][inner sep=0.75pt]  [font=\scriptsize] [align=left] {$v'_{1}$};
    \draw (205.5,165.98) node [anchor=north west][inner sep=0.75pt]  [font=\scriptsize] [align=left] {$D$};
    \draw (425.85,65.55) node [anchor=north west][inner sep=0.75pt]  [font=\normalsize] [align=left] {$\mathcal{H}$};
    \draw (257.67,125.53) node [anchor=north west][inner sep=0.75pt]  [font=\scriptsize] [align=left] {$v'_{2t}$};
    \draw (172.43,133.23) node [anchor=north west][inner sep=0.75pt]  [font=\tiny,rotate=-26.24] [align=left] {$(0,...,0)$};
    \draw (313.92,120.92) node [anchor=north west][inner sep=0.75pt]  [font=\footnotesize] [align=left] {Player~$0$ vs player~$1$};
    \draw (219.06,150.34) node [anchor=north west][inner sep=0.75pt]  [font=\tiny,rotate=-328.5] [align=left] {$(0,...,0)$};
    \draw (234.48,72.47) node [anchor=north west][inner sep=0.75pt]  [font=\tiny] [align=left] {$(0,...,0)$};
    \draw (106.65,153.22) node [anchor=north west][inner sep=0.75pt]  [font=\tiny] [align=left] {$(0,0,-1,...,-1)$};
    \draw (290.88,79.17) node [anchor=north west][inner sep=0.75pt]  [font=\footnotesize] [align=left] {$v_{0}$};
    \draw (295.55,100.17) node [anchor=north west][inner sep=0.75pt]  [font=\footnotesize] [align=left] {$v$};
    \draw (412.55,98.5) node [anchor=north west][inner sep=0.75pt]  [font=\footnotesize] [align=left] {$v'$};
    \draw (308.82,92.8) node [anchor=north west][inner sep=0.75pt]  [font=\tiny] [align=left] {$(0,d_{1},-d_{1},...,d_{t},-d_{t})$};
    \draw (205.5,76.12) node [anchor=north west][inner sep=0.75pt]  [font=\scriptsize] [align=left] {$v'_{0}$};
    \draw (205.82,131.42) node [anchor=north west][inner sep=0.75pt]  [font=\scriptsize] [align=left] {$...$};
    \draw (258.13,151.43) node [anchor=north west][inner sep=0.75pt]  [font=\tiny] [align=left] {$(0,-1,...,-1,0)$};
    \draw (229.39,84.93) node [anchor=north west][inner sep=0.75pt]  [font=\tiny,rotate=-41.08] [align=left] {$(0,...,0,1)$};
    \draw (157.58,119.33) node [anchor=north west][inner sep=0.75pt]  [font=\tiny,rotate=-316.21] [align=left] {$(0,1,0,...,0)$};
    \end{tikzpicture}
    }
    \caption{Reduction from the \prob{ESP} problem to the \prob{NCPS} problem.}
    \label{fig:undecidability-ncps}
  \end{minipage}
\end{figure}

\begin{proof}[Proof of \cref{theorem:undecidableNE}]
    (1) We first show the undecidability of the \prob{NCNS} problem by reduction from the complement of the \prob{SP} problem. Let $\mathcal{H}= (\arena,\Omega)$ be a multidimensional shortest path game, with the arena $\arena = (V,E,\{\EE{},\AA{}\},V_\E,V_\A,(w_i)_{i \in \{1,\dots,t\}})$, the target set $T$, and the initial vertex $v_0 \in V$. We construct a reachability game $\game' = ((V', E', \Players, (V'_i)_{i\in\Players}, (w'_i)_{i \in \Players}), (\reach{i}')_{i \in \Players})$ with $t+2$ players as illustrated in \cref{fig:undecidability-ncns}, where:
    \begin{itemize}
        \item $\Players = \{0,\dots,t+1\}$, 
        \item $V' = V \cup \{v'_1,\dots,v'_t,D\}$,
        \item $V'_i = \{v'_i\}$ for all $i \in \Players \ssetminus \{0,t+1\}$, $V'_0 = V_\A$, and $V'_{t+1} = V_\E$,
        \item $E' = E \cup \{(v_{i}',v_{i+1}'),(v_{i},D) \mid i \in \{1,\dots,t-1\}\} \cup \{(v_t',D),(D,D),(v_t',v_0)\}$,
        \item the functions $w'_i$, for all $i \in \Players \ssetminus \{0,t+1\}$, are the same as $w_i$ in $\mathcal{H}$ and assign $0$ on the added edges, and the functions $w_0 = w_{t+1}$ both assign $0$ to every edge of $E'$,
        \item $\reach{i}' = \reach{} \cup \{D\}$ for all $i \in \Players \ssetminus \{0,t+1\}$, $\reach{0}' = \{D\}$ and $\reach{t+1}' = V'$.
    \end{itemize}
    From the initial vertex $v'_1$, each player~$1,\dots,t$ has the choice to go to $D$ (and thus giving a zero cost to everybody) or to go to the game $\mathcal{H}$. In $\mathcal{H}$, player~$0$ has the role of \AA{} and always gets an infinite cost (he gets a zero cost outside $\mathcal H$), while player~$t+1$ has the role of \EE{} and always gets a zero cost (even outside $\mathcal{H}$). Let us prove that there exists a solution to the \prob{NCNS} problem in $\game'$ with the threshold $c' = 0$ if and only if \AA{} has a winning strategy in $\mathcal{H}$ with the threshold $\bar c = \bar 0$.

    If there exists a strategy $\strategyfor{0}$ solution to the \prob{NCNS} problem, it means that for every \fixedEquilibriaStrategy{} outcome $\pi'$, $\cost{0}{\pi'} \leq 0$ holds. From $\strategyfor{0}$, we can construct a strategy $\strategyfor{\A}$ for \AA{}. Let $\strategyfor{\E}$ be any strategy for \EE{}. These two strategies generate a play $\pi$ in $\mathcal{H}$ (from $v_0$), leading to the play $\pi' = v'_1\dots v'_t \cdot\pi$ in $\game'$. By construction, this play $\pi'$ is not a \fixedEquilibria{} outcome as $\pi'$ is consistent with $\strategyfor{0}$ and $\cost{0}{\pi'} = +\infty$. Therefore, some player, say player~$i$, has a profitable deviation along $\pi'$. As player~$t+1$ has no profitable deviation (he gets a zero cost for all plays), it follows that $i \in \{1,\dots,t\}$. That deviation must occur in vertex $v'_i$ as it is the only vertex controlled by player~$i$ along $\pi'$. Therefore, the deviating play gives a zero cost for player~$i$, meaning that $\cost{i}{\pi'} > 0$. Hence, we get $\cost{i}{\pi} > 0$ in $\mathcal{H}$. This shows that $\strategyfor{\A}$ is a winning strategy for \AA{} in $\mathcal H$.

    Suppose now that $\strategyfor{\A}$ is a winning strategy for \AA{} in $\mathcal H$. Let $\strategyfor{0}$ be a strategy of player~$0$ in $\game'$, which corresponds to the strategy $\strategyfor{\A}$. By contradiction, suppose that $\strategyfor{0}$ is not a solution to the \prob{NCNS} problem, i.e., there exists a \fixedEquilibriaStrategy{} outcome $\pi'$ such that $\cost{0}{\pi'} > 0$. As $\reach{0}' = \{D\}$, it means that $\pi' = v'_1\dots v'_t \cdot \pi$, with $\pi$ a play in $\mathcal H$ consistent with $\strategyfor{\A}$, thus such that $\cost{i}{\pi} > 0$ for some $i \in \{1,\ldots,t\}$. This is in contradiction with $\pi'$ being a \fixedEquilibriaStrategy{} outcome, as player~$i$ has a profitable deviation in~$v'_i$.

    \medskip
    (2) We then show the undecidability of the \prob{NCPS} problem by reduction from the complement of the \prob{ESP} problem. Let again $\mathcal{H} = (\arena,\Omega)$ be a multidimensional shortest path game. We now construct a reachability game $\game' = ((V', E', V'_0, V'_1, (w'_i)_{i \in \{0,\dots,2t\}}),(\reach{i}')_{i \in \{0,\dots,2t\}})$ with two players~$0$ and~$1$, and $2t$ reachability objectives from player~$1$, as illustrated in \cref{fig:undecidability-ncps}, where:
    \begin{itemize}
        \item $V' = V \cup \{v_0',v_1',\dots,v_{2t}',D\}$,
        \item $E' = E \cup \{(v_0',v_0),(D,D)\} \cup \{(v_0',v_i'),(v_i',v_i'),(v_i',D) \mid i \in \{1,\dots,2t\}\}$,
        \item $V'_0 = V_\A$, $V'_1 = V_\E \cup \{v_0',v_1',\dots,v_{2t}',D\}$,
        \item the weight function $w'_0$ is null; for each $i \in \{1,\dots,2t\}$, we have
        \begin{itemize}
            \item for each $e \in E$, $w_i'(e) = w_{(i+1)/2}(e)$ if $i$ is odd, and $w_i'(e) = -w_{i/2}(e)$ if $i$ is even,
            \item  $w'_i((v_j',D)) = w'_i((D,D)) = 0$,
            \item and \[w_i'((v_0',v_j')) = \begin{cases}
            1 & \text{ if $j = i$,} \\
            0 & \text{ if $j \neq i$,}
        \end{cases} \quad w_i'((v_j',v_j')) = \begin{cases}
            0 & \text{ if $j = i$,} \\
            -1 & \text{ if $j \neq i$,}
        \end{cases}\] 
        \end{itemize}
        \item $\reach{0}' = \{D\}$, $\reach{i}' = \reach{} \cup \{D\}$ for each $i \in \{1,\dots,2t\}$.
    \end{itemize}
    We denote by $\costfuncprim{}$ the cost function in $\mathcal{G'}$.
    Hence, a weight tuple $(d_1,d_2, \ldots, d_t)$ labeling an edge in $\mathcal H$ is replaced by the weight tuple $(0,d_1,-d_1,d_2,-d_2, \ldots, d_t,-d_t)$ in the new game. Moreover, a play $\pi$ in $\mathcal H$ with $\cost{}{\pi} = (c_1,c_2,\ldots,c_t) \in (\Z \cup \{+\infty\})^{t}$ has a cost for player~$1$ in $\mathcal{G}'$ equal to $\payoffprim{\pi} = (c_1,-c_1,c_2,-c_2, \ldots, c_t,-c_t)$ if $\pi$ visits $\reach{}$ and $(+\infty, \ldots, \infty)$ otherwise. Note that any play $\pi'$ in the left part of $\mathcal{G}'$ that visits $D$ has a cost for player~$1$ of the form $(-c,\ldots,-c,1,-c, \ldots, -c)$ where $c \in \\N$ can be arbitrarily large.
     
    Let us prove that there exists a solution to the \prob{NCPS} problem in $\game'$ with the initial vertex $v_0'$ and the threshold $c = 0$ if and only if \AA{} has a winning strategy in $\mathcal{H}$ for the \prob{ESP} problem with the threshold $\bar c = \bar 0$.

    If \AA{} has a winning strategy $\strategyfor{\A}$ in $\mathcal{H}$, then player~$0$ can construct a strategy $\strategyfor{0}$ in $\mathcal{G}'$ that simulates $\strategyfor{\A}$. Hence, every play $\pi$ consistent with $\strategyfor{0}$ going to the right part of $\mathcal{G}'$ satisfies $\payoffprim{\pi} \neq (0,\dots,0)$, i.e., either $\payoffprim{\pi} = (+\infty,\ldots,+\infty)$ or there exists two consecutive indices, say $i$ and $i+1$, such that $\costprim{i}{\pi},\costprim{i+1}{\pi} \in \Z$ and $\costprim{i+1}{\pi} = -\costprim{i}{\pi}$.
    W.l.o.g, suppose that $\costprim{i}{\pi} > 0$.
    Consequently, there is a play $\pi' = v_0' (v_i')^c D^\omega$ for some large enough $c \in \N$ such that $\payoffprim{\pi'} = (-c,\ldots,-c,1,-c, \ldots, -c) < \payoffprim{\pi}$. It means that there is no $\strategyfor{0}$-fixed PO play going to the right part of $\mathcal{G}'$. Therefore, $\strategyfor{0}$ is solution to the \prob{NCPS} problem since if there exists a $\strategyfor{0}$-fixed PO play, it must visit $D$, leading to a zero cost to player~$0$.\footnote{Note that there is no $\strategyfor{0}$-fixed PO play, making $\sigma_0$ trivially a solution to the \prob{NCPS} problem.}

    Conversely, if \EE{} has a winning strategy $\strategyfor{\E}$ for the \prob{ESP} problem in $\mathcal H$, then for all strategies $\strategyfor{0}$ of player~$0$, player~$1$ has a strategy $\strategyfor{1}$ that mimics $\strategyfor{\E}$ to generate a play $\pi$ going to the right part of $\mathcal{G}'$ with $\payoffprim{\pi} = (0,\dots,0)$. As for every play $\pi'$ visiting $D$ (in the left part of $\mathcal{G}'$), there exists $i \in \{1,\dots,2t\}$ such that $\costprim{i}{\pi'} = 1$, it follows that $\pi$ is a $\strategyfor{0}$-fixed PO play, but does not visit $\reach{0}$. Consequently, the instance of the \prob{NCPS} is false.
\end{proof}

\section{Lower Bounds for Pareto Optimality}\label{appendix:lower-bounds-pareto}

We prove the lower bounds of \cref{theorem:pareto-results} for \emph{qualitative} reachability. As a corollary, for the problems studied in this paper, we get the same complexity class for quantitative reachability as for qualitative reachability.

In the following proofs, it is important to recall that in qualitative reachability games, the weight of all edges is $(0,\dots,0)$. Hence, $\cost{i}{\pi} = 0$ if the play $\pi$ visits the target set $\reach{i}$, and $+\infty$ otherwise. In particular, we can always set the threshold $c$ to $0$. In the sequel, instead of using the notation $\cost{i}{\pi} = 0$, we will say that $\reach{i}$ is visited.

\subsection{Lower bound for the NCPV problem}

We begin with the lower bound of \cref{theorem:pareto-results}.\ref{thm-NCPV}.

\begin{theorem} \label{theorem:piHard}
    The \probname{}{NC}{P}{V} problem is \piHard{}{}.
\end{theorem}

We show that there exists a polynomial reduction from a variant of the \emph{\sigmaQBF{} problem} to the complementary of the \prob{NCPV} problem. Given a Boolean formula $\varphi$ in 3DNF on the set of variables $X \cup Y$, with $X$ and $Y$ being disjoint, the \sigmaQBF{} problem asks whether there exists a valuation $val_X$ of variables of $X$ such that for every valuation $val_Y$ of variables of $Y$, the valuation $(val_X,val_Y)$ satisfies $\varphi$. In other words,
\[
\exists X \; \forall Y ~ [\varphi(X,Y) = 1].
\]
This problem is known to be \sigmaComplete{}~\cite{STOCKMEYER19761}. The variant, that we call \emph{\sigmaQBFneg{} problem}, asks whether $\varphi$ in 3CNF is not satisfied, i.e.,
\[
\exists X \; \forall Y ~ [\varphi(X,Y) = 0].
\]

\begin{corollary}
    The \sigmaQBFneg{} problem with a 3CNF formula $\varphi$ is \sigmaComplete{}.
\end{corollary}

We are going to show that the \sigmaQBFneg{} problem reduces to the \coprob{NCPV} problem (with threshold $c=0$). Hence the \prob{NCPV} problem will be \piHard{}. Recall that in this case, player~$1$ is the only one to play. Our goal is then to show that there exists a PO play $\pi$ not visiting $\reach{0}$ if and only if the \sigmaQBFneg{} instance is satisfied. Let $X = \{x_1,\dots,x_n\}$, $Y = \{y_1,\dots,y_m\}$ be two sets of distinct variables and $\varphi$ be a 3CNF Boolean formula on $X \cup Y$, i.e., in the form
\[
\varphi = C_1 \land \dots \land C_k, \quad \text{and} \quad \forall i, \; C_i = (\ell_{i,1} \lor \ell_{i,2} \lor \ell_{i,3}),
\]
where $\ell_{i,j}$ is a literal, representing either a variable $v \in X \cup Y$, or its negation $\neg v$.

\begin{figure}[ht]
    \centering
    \resizebox{0.9\textwidth}{!}{%
        \begin{tikzpicture}

        \node[draw, rectangle, minimum size=0.8cm] (o) at (-2.5,-0.25){$v_0$};


        \draw[dashed, rounded corners] (2.9,-1.8) rectangle (10.1, 1.3)  {};
        \node[] at (10.1 - 0.4, 1.3 - 0.3) {$\arena_1$};
         
        \node[draw, rectangle, minimum size=0.8cm] (b1b) at (3.5,-0.25){$v_1$};
        \node[draw, rectangle, minimum size=0.8cm,inner sep=0.05pt] (x11b) at (4.5,-1.25){$\neg x^1_1$};
        \node[draw, rectangle, minimum size=0.8cm] (nx11b) at (4.5,0.75){$x^1_1$};
        \node[draw, rectangle, minimum size=0.8cm] (b2b) at (5.5,-0.25){};
        \node[minimum size=0.8cm] (b2i1b) at (6.5,-1.25){$\dots$};
        \node[minimum size=0.8cm] (b2i2b) at (6.5,0.75){$\dots$};
        \node[draw, rectangle, minimum size=0.8cm] (bmb) at (7.5,-0.25){};
        \node[draw, rectangle, minimum size=0.8cm, inner sep=0.05pt] (x1mb) at (8.5,-1.25){$\neg x^1_m$};
        \node[draw, rectangle, minimum size=0.8cm] (nx1mb) at (8.5,0.75){$x^1_m$};
        \node[draw, rectangle, minimum size=0.8cm] (end1b) at (9.5,-0.25){};
        
        \draw[-stealth, shorten >=1pt,auto] (b1b) to [] node []{} (x11b);
        \draw[-stealth, shorten >=1pt,auto] (b1b) to [] node []{} (nx11b);
        \draw[-stealth, shorten >=1pt,auto] (x11b) to [] node []{} (b2b);
        \draw[-stealth, shorten >=1pt,auto] (nx11b) to [] node []{} (b2b);
        \draw[-stealth, shorten >=1pt,auto] (b2b) to [] node []{} (b2i1b);
        \draw[-stealth, shorten >=1pt,auto] (b2b) to [] node []{} (b2i2b);
        \draw[-stealth, shorten >=1pt,auto] (b2i1b) to [] node []{} (bmb);
        \draw[-stealth, shorten >=1pt,auto] (b2i2b) to [] node []{} (bmb);
        \draw[-stealth, shorten >=1pt,auto] (bmb) to [] node []{} (x1mb);
        \draw[-stealth, shorten >=1pt,auto] (bmb) to [] node []{} (nx1mb);
        \draw[-stealth, shorten >=1pt,auto] (x1mb) to [] node []{} (end1b);
        \draw[-stealth, shorten >=1pt,auto] (nx1mb) to [] node []{} (end1b);

        \path (end1b) edge[loop below] (end1b);


        \draw[dashed, rounded corners] (-3.1,-5.3) rectangle (10.1, -2.2) {};
        \node[] at (10.1 - 0.4,-2.2 - 0.3) {$\arena_2$};

        \node[draw, rectangle, minimum size=0.8cm] (s1) at (-2.5,-3.75){$v_2$};
        \node[draw, rectangle, minimum size=0.8cm, inner sep=0.05pt] (s2) at (-1.5,-4.75){$\neg x^2_1$};
        \node[draw, rectangle, minimum size=0.8cm] (s3) at (-1.5,-2.75){$x^2_1$};
        \node[draw, rectangle, minimum size=0.8cm] (s4) at (-0.5,-3.75){};
        \node[minimum size=0.8cm] (s5) at (0.5,-4.75){$\dots$};
        \node[minimum size=0.8cm] (s6) at (0.5,-2.75){$\dots$};
        \node[draw, rectangle, minimum size=0.8cm] (s7) at (1.5,-3.75){};
        \node[draw, rectangle, minimum size=0.8cm, inner sep=0.05pt] (s8) at (2.5,-4.75){$\neg x^2_n$};
        \node[draw, rectangle, minimum size=0.8cm] (s9) at (2.5,-2.75){$x^2_n$};

        \node[draw, rectangle, minimum size=0.8cm] (d1) at (3.5,-3.75){};
        \node[draw, rectangle, minimum size=0.8cm, inner sep=0.05pt] (x31) at (4.5,-4.75){$\neg y_1$};
        \node[draw, rectangle, minimum size=0.8cm] (nx31) at (4.5,-2.75){$y_1$};
        \node[draw, rectangle, minimum size=0.8cm] (d2) at (5.5,-3.75){};
        \node[minimum size=0.8cm] (d2i1) at (6.5,-4.75){$\dots$};
        \node[minimum size=0.8cm] (d2i2) at (6.5,-2.75){$\dots$};
        \node[draw, rectangle, minimum size=0.8cm] (dm) at (7.5,-3.75){};
        \node[draw, rectangle, minimum size=0.8cm, inner sep=0.05pt] (x3m) at (8.5,-4.75){$\neg y_m$};
        \node[draw, rectangle, minimum size=0.8cm] (nx3m) at (8.5,-2.75){$y_m$};
        \node[draw, rectangle, minimum size=0.8cm] (eng3) at (9.5,-3.75){};
        
        \draw[-stealth, shorten >=1pt,auto] (s1) to [] node []{} (s2);
        \draw[-stealth, shorten >=1pt,auto] (s1) to [] node []{} (s3);
        \draw[-stealth, shorten >=1pt,auto] (s2) to [] node []{} (s4);
        \draw[-stealth, shorten >=1pt,auto] (s3) to [] node []{} (s4);
        \draw[-stealth, shorten >=1pt,auto] (s4) to [] node []{} (s5);
        \draw[-stealth, shorten >=1pt,auto] (s4) to [] node []{} (s6);
        \draw[-stealth, shorten >=1pt,auto] (s5) to [] node []{} (s7);
        \draw[-stealth, shorten >=1pt,auto] (s6) to [] node []{} (s7);
        \draw[-stealth, shorten >=1pt,auto] (s7) to [] node []{} (s8);
        \draw[-stealth, shorten >=1pt,auto] (s7) to [] node []{} (s9);
        \draw[-stealth, shorten >=1pt,auto] (s8) to [] node []{} (d1);
        \draw[-stealth, shorten >=1pt,auto] (s9) to [] node []{} (d1);
        
        \draw[-stealth, shorten >=1pt,auto] (d1) to [] node []{} (x31);
        \draw[-stealth, shorten >=1pt,auto] (d1) to [] node []{} (nx31);
        \draw[-stealth, shorten >=1pt,auto] (x31) to [] node []{} (d2);
        \draw[-stealth, shorten >=1pt,auto] (nx31) to [] node []{} (d2);
        \draw[-stealth, shorten >=1pt,auto] (d2) to [] node []{} (d2i1);
        \draw[-stealth, shorten >=1pt,auto] (d2) to [] node []{} (d2i2);
        \draw[-stealth, shorten >=1pt,auto] (d2i1) to [] node []{} (dm);
        \draw[-stealth, shorten >=1pt,auto] (d2i2) to [] node []{} (dm);
        \draw[-stealth, shorten >=1pt,auto] (dm) to [] node []{} (x3m);
        \draw[-stealth, shorten >=1pt,auto] (dm) to [] node []{} (nx3m);
        \draw[-stealth, shorten >=1pt,auto] (x3m) to [] node []{} (eng3);
        \draw[-stealth, shorten >=1pt,auto] (nx3m) to [] node []{} (eng3);

        \path (eng3) edge[loop below] (eng3);

        \draw[-stealth, shorten >=1pt,auto] (o) to [] node []{} (s1);
        \draw[-stealth, shorten >=1pt,auto] (o) to [] node []{} (b1b);

        \end{tikzpicture}
        }%
    \caption{The single-player arena $\arena$ used in the reduction from the \sigmaQBFneg{} problem.}
    \label{fig:PRV-fig-QBF}
\end{figure}
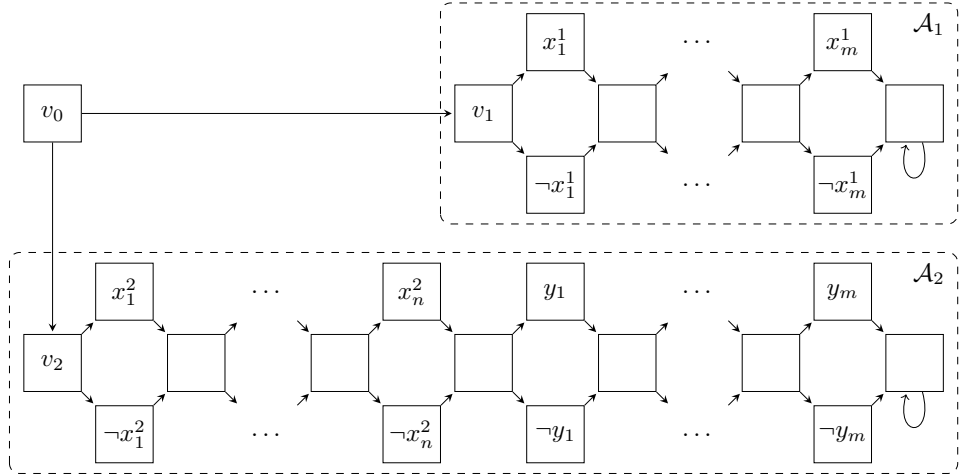

The intuition of the reduction is the following one (see \cref{fig:PRV-fig-QBF}). From $\varphi$, we build two subarenas, arena $\arena_1$ where player~$1$ is ``happy'' except for one particular target $\reach{1}$ and where player~$0$ never visits his target, and arena $\arena_2$ where player~$0$ always visits his target and player~$1$ visits some targets according to which clauses of $\varphi$ are satisfied. There are two types of targets: $(i)$ some to force a specific $val_X$ to be considered and $(ii)$ some to check if a clause $C_i$ is satisfied. In $\arena_1$, any play represents a valuation of $X$, and in $\arena_2$, any play represents a valuation of $X \cup Y$. The idea is to make one clause always false when a valuation $val_X$ is taken in $\arena_2$, in a way that we get an incomparable cost with the same valuation in $\arena_1$.

\begin{proof}[Proof of \cref{theorem:piHard}]
Let us first describe the arena of the game and its target sets. All vertices are owned by player~$1$ who is thus the only one to play. The initial vertex is $v_0$ with an edge to $v_1$ to enter $\arena_1$ and an edge to $v_2$ to enter $\arena_2$.
There are $1 + 2n + k$ target sets $(\reach{1},\reach{x_1},\reach{\neg x_1},\dots,\reach{x_n},\reach{\neg x_n},\reach{C_1},\dots,\reach{C_k})$ for player~$1$ and the target set $\reach{0}$ for player~$0$:
\begin{itemize}
    \item $\reach{0} = \reach{1} = \{v_2\}$,
    \item for all $i \leq n$, $\reach{x_i} = \{x^1_i,x^2_i\}$, $\reach{\neg x_i} = \{\neg x^1_i, \neg x^2_i\}$,
    \item for all $i \leq k$, $\reach{C_i} = \{v_1,\ell_{i,1}^2,\ell_{i,2}^2,\ell_{i,3}^2\}$, where $\ell_{i,j}^2$ is the literal vertex in $\arena_2$ representing the $j$-th literal $\ell_{i,j}$ of $C_i$.
\end{itemize}
From the definition of target sets $\reach{C_i}$, any play $\pi$ in $\arena_2$ corresponds to a valuation of the variables of $\varphi$, where $\reach{C_i}$ is visited if and only if the clause $C_i$ is satisfied.

Let us assume that the instance of the \sigmaQBF{} problem is positive. That is, there exists a valuation $val_X$ of the variables in $X$ such that whatever the valuation $val_Y$ of the remaining variables in $Y$, the formula $\varphi$ is false. Let $\pi$ be the play in $\arena_1$ which corresponds to the valuation $val_X$. It does not visit $\reach{0}$, and $\payoff{\pi} = (+\infty, val_X, 0, \dots, 0)$, where $val_X$ is here expressed as a vector of $2n$ values in $\{0,+\infty\}$ for the targets $\reach{x_1}, \reach{\neg x_1}, \dots, \reach{x_n}, \reach{\neg x_n}$.\footnote{Recall: $0$ in case of a target visit, and $+ \infty$ otherwise.} Let us show that $\pi$ is PO, i.e., no play in $\arena$ has a cost strictly smaller than $\payoff{\pi}$. To avoid plays with incomparable costs, it is enough to consider plays $\pi'$ in $\arena_2$ which correspond to the same valuation $val_X$. Such a play $\pi'$ represents $val_X$ and then some valuation $val_Y$. Since $val_X$ is a positive instance to the \sigmaQBFneg{} problem, for each such play $\pi'$, there exists $i \leq k$ such that $\pi'$ does not visit $\reach{C_i}$, as no valuation of $Y$ together with $val_X$ satisfies $\varphi$. It follows that $\payoff{\pi'}$ is incomparable to $\payoff{\pi}$, as $\pi$ visits $\reach{C_i}$ but $\pi'$ does not, and $\pi'$ visits $\reach{1}$ but $\pi$ does not. We conclude that $\pi$ is PO play that does not visit $\reach{0}$, i.e., with $\cost{0}{\pi} > c =0$.

Let us now assume that there exists a PO play $\pi$ in $\arena$ not visiting $\reach{0}$. It is thus a play in $\arena_1$ that corresponds to a valuation $val_X$ of the variables of $X$ in $\varphi$. And there is no play in $\arena_2$ with a strictly smaller cost. It follows that all play in $\arena_2$ corresponding to the same valuation $val_X$ do not visit at least one target $\reach{C_i}$. Therefore, $val_X$ is a positive instance of the \sigmaQBFneg{} problem, as for all valuations $val_Y$ of $Y$, together $val_X$ and $val_Y$ do not satisfy $\varphi$.
\end{proof}

\subsection{Lower bound for the CPS and UNCPV problems}

We here prove that both \prob{UNCPV} and \prob{CPS} problems are \pspaceHard{} (lower bound of \cref{theorem:pareto-results}.\ref{thm-UNCPV}-\ref{thm-CPS}). 

\begin{theorem} \label{theorem:pspaceHard}
    The \probname{}{C}{P}{S} problem and the \probname{U}{NC}{P}{V} problem are both \pspaceHard{}{}.
\end{theorem}

We first prove the theorem for the \prob{CPS} problem. We then explain what to modify to get the reduction for the \coprob{UNCPV} problem.
For this purpose, we need to define a variant of the \QBF{} problem: the \emph{\coQBF{} problem with 3CNF formulas}. An instance of this problem is in the form
\[
    Q_1 x_1 \dots Q_n x_n ~ [\varphi = 0]
\]
where $\varphi$ is a 3CNF formula over the variables $\{x_1,\dots,x_n\}$, and $Q_i \in \{\forall,\exists\}$ for all $i$.

\begin{lemma}
    The \coQBF{} problem with 3CNF formulas is \pspaceComplete{}.
\end{lemma}

\begin{proof}
    The \pspace{}-membership is clear by the same algorithm as for \QBF{}. The \pspaceHard{}ness is obtained through two consecutive reductions. For the first one, given some \QBF{} instance $Q_1 x_1 \dots Q_n x_n ~ [\varphi = 1]$ with $\varphi$ a 3CNF Boolean formula, we simply construct the negation of this instance:
    \[ Q'_1 x_1 \dots Q'_n x_n ~[\neg\varphi = 1], \]
    where $Q'_i = \forall$ (resp.\ $Q'_i = \exists$) if $Q_i = \exists$ (resp.\ $Q_i = \forall$). Notice that $\neg\varphi$ is a 3DNF formula. For the second reduction from the previous instance, we use the property ``$\varphi$ is true if and only if $\neg \varphi$ is false'', hence we construct 
    \[ Q'_1 x_1 \dots Q'_n x_n ~[\neg(\neg\varphi) = 0].\]
    This concludes the proof, by $\neg(\neg\varphi) = \varphi$ and by construction.
\end{proof}

\begin{proof}[Proof of \Cref{theorem:pspaceHard}]

We begin with a reduction from the \coQBF{} problem to the \prob{CPS} problem (with threshold $c=0$).
Let $Q_1 x_1 \dots Q_n x_n \; [\varphi = 0]$ be a \coQBF{} instance with $\varphi = C_1 \wedge \dots \wedge C_n$ a 3CNF formula on the set of variables $\{x_1,\dots,x_n\}$ such that $C_i = (\ell_{i,1} \lor \ell_{i,2} \lor \ell_{i,3})$ for all $i \leq n$. We construct an arena $\arena$ as shown in \cref{fig:qbf-coop-pareto-synthesis} with $n+1$ target sets $(\reach{1},\reach{C_1},\dots,\reach{C_n})$ for the environment: 
\begin{itemize}
    \item $V_0 = \{Q_i \mid Q_i = \exists\}$, $V_1 = V \ssetminus V_0$, and $v_0$ is the initial vertex,
    \item $\reach{0} = \{v_1\}$,
    \item $\reach{1} = \{Q_1\}$, $\reach{C_i} = \{v_1,\ell_{i,1},\ell_{i,2},\ell_{i,3}\}$ for all $i \leq n$.
\end{itemize}

\begin{figure}
    \centering
    \resizebox{.7\textwidth}{!}{%
        \begin{tikzpicture}
        
        \node[draw, rectangle, minimum size=0.8cm] (o) at (1.7,5){$v_0$};
        
        \node[draw, rectangle, minimum size=0.8cm] (b1bb) at (0,5){$v_1$};

        \draw[dashed, rounded corners] (2.9,3.25) rectangle (10.1, 6.75)  {};
        \node[] at (10.1 - 0.4, 6.75 - 0.3) {$\arena_1$};

        \node[draw, rounded corners, minimum size=0.8cm] (b1b) at (3.5,5){$Q_1$};
        \node[draw, rectangle, minimum size=0.8cm] (x11b) at (4.5,4){$\neg x_1$};
        \node[draw, rectangle, minimum size=0.8cm] (nx11b) at (4.5,6){$x_1$};
        \node[draw, rounded corners, minimum size=0.8cm] (b2b) at (5.5,5){$Q_2$};
        \node[minimum size=0.8cm] (b2i1b) at (6.5,4){$\dots$};
        \node[minimum size=0.8cm] (b2i2b) at (6.5,6){$\dots$};
        \node[draw, rounded corners, minimum size=0.8cm] (bmb) at (7.5,5){$Q_n$};
        \node[draw, rectangle, minimum size=0.72cm, inner sep=0.05pt] (x1mb) at (8.5,4){$\neg x_n$};
        \node[draw, rectangle, minimum size=0.8cm] (nx1mb) at (8.5,6){$x_n$};
        \node[draw, rectangle, minimum size=0.8cm] (end1b) at (9.5,5){};

        \draw[-stealth, shorten >=1pt,auto] (b1b) to [] node []{} (x11b);
        \draw[-stealth, shorten >=1pt,auto] (b1b) to [] node []{} (nx11b);
        \draw[-stealth, shorten >=1pt,auto] (x11b) to [] node []{} (b2b);
        \draw[-stealth, shorten >=1pt,auto] (nx11b) to [] node []{} (b2b);
        \draw[-stealth, shorten >=1pt,auto] (b2b) to [] node []{} (b2i1b);
        \draw[-stealth, shorten >=1pt,auto] (b2b) to [] node []{} (b2i2b);
        \draw[-stealth, shorten >=1pt,auto] (b2i1b) to [] node []{} (bmb);
        \draw[-stealth, shorten >=1pt,auto] (b2i2b) to [] node []{} (bmb);
        \draw[-stealth, shorten >=1pt,auto] (bmb) to [] node []{} (x1mb);
        \draw[-stealth, shorten >=1pt,auto] (bmb) to [] node []{} (nx1mb);
        \draw[-stealth, shorten >=1pt,auto] (x1mb) to [] node []{} (end1b);
        \draw[-stealth, shorten >=1pt,auto] (nx1mb) to [] node []{} (end1b);

        \path (end1b) edge[loop below] (end1b);

        \draw[-stealth, shorten >=1pt,auto] (o) to [] node []{} (b1bb); 
        \draw[-stealth, shorten >=1pt,auto] (b1bb) to [loop left] node []{} (b1bb);

        \draw[-stealth, shorten >=1pt,auto] (o) to [] node []{} (b1b);

        \end{tikzpicture}
        }%
    \caption{The arena used in the reduction from \coQBF{}.}
    \label{fig:qbf-coop-pareto-synthesis}
\end{figure}
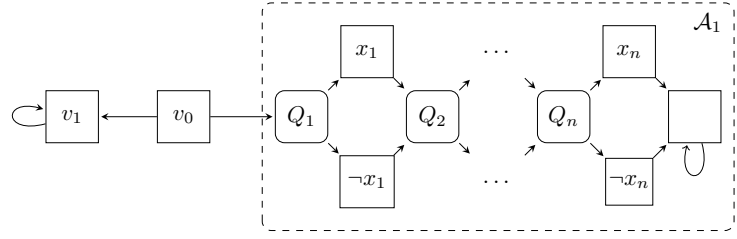

Suppose first that the \coQBF{} instance is positive. Let player~$0$ mimic the behavior of the existential player as a strategy $\strategyfor{0}$ in $\arena_1$ and let us show that the play $\pi = v_0(v_1)^\omega$ with $\cost{0}{\pi} = 0$ and $\payoff{\pi} = (+\infty,0,\dots,0)$, is $\strategyfor{0}$-fixed PO. Any play $\pi'$ in $\arena_1$ consistent with $\strategyfor{0}$ represents a valuation of $\varphi$. By hypothesis, it makes $\varphi$ unsatisfiable, i.e., there is an unsatisfied clause $C_i$, meaning that $\reach{C_i}$ is not visited by $\pi'$. The cost $\payoff{\pi'}$ is incomparable to $\payoff{\pi}$, as $\reach{C_i}$ is visited in $\pi$ and not in $\pi'$, while $\reach{1}$ is visited in $\pi'$ and not in $\pi$.

Let us suppose that the \prob{CPS} instance is positive. Then there exists $\strategyfor{0} \in \Sigma_0$ and a $\strategyfor{0}$-fixed PO play $\pi$ where $\cost{0}{\pi} = 0$. Necessarily, $\pi = v_0(v_1)^\omega$, as this is the only play with a zero cost for player~$0$. By hypothesis on $\strategyfor{0}$, for any play $\pi'$ consistent with $\strategyfor{0}$ and ending in $\arena_1$, it is not the case that $\payoff{\pi'} < \payoff{\pi}$. But since $\reach{1}$ is visited by $\pi'$ but not by $\pi$, $\payoff{\pi'}$ must be incomparable with $\payoff{\pi}$, i.e., it cannot visit $\reach{C_i}$, for some $i$. Therefore, for all such plays $\pi$', i.e., for all responses to the universal player in $\varphi$, there exists $i \leq n$ such that $C_i$ is not satisfied, thus making $\varphi$ unsatisfied.

Finally, let us adapt the previous proof to the \coprob{UNCPV} problem. We construct the same game except that $\reach{0} = \{Q_1\}$. The proof remains the same as the play $\pi = v_0(v_1)^\omega$ has now $\cost{0}{\pi} = +\infty$.
\end{proof}

\section{Characterization of Nash Equilibria outcomes}\label{appendix:ne-characterization}

For some results of this paper, we need to recall a characterization of NE outcomes for quantitative reachability~\cite{DBLP:journals/jcss/BrihayeBGT21}. This section aims to recall this characterization based on values of some particular two-player zero-sum games. See the survey~\cite{Bru21} for more details about this characterization for different kinds of objectives.

Given a reachability game $\game$ and a fixed strategy $\sigma_0$ of player~$0$, we denote by \ValStarFunc{}$: V \rightarrow \N \cup \{+\infty\}$ the function giving, for any vertex $v \in V_i$, $i \in \Players \ssetminus \{0\}$, the \emph{value} of player~$i$ in the zero-sum setting of player~$i$ opposed to the \emph{coalition} of the other players, which is the lowest cost, with resepct to $w_i$, that player~$i$ can ensure against this coalition. These values can be computed in polynomial time through a value iteration algorithm, and belong to $\{0,\dots,|V|W,+\infty\}$~\cite{DBLP:conf/lfcs/BrihayePS13,DBLP:journals/acta/BrihayeGHM17,DBLP:journals/mst/KhachiyanBBEGRZ08}.

Formally, to define the value $\ValPlayer{i}{v}$ of a vertex $v$ for player~$i$ against the coalition $\Players \ssetminus \{0,i\}$ denoted $-i$, we first need to define the \emph{lower value} $\lowerVal{i}{v}$ and the \emph{upper value} $\upperVal{i}{v}$ of a vertex $v$ in the following way:
\[\lowerVal{i}{v}=\sup_{\strategyfor{-i}}\; \inf_{\strategyfor{i}}\; \cost{i}{\outcomefrom{\sigma_0,\strategyfor{i},\strategyfor{-i}}{v}},
\\
\upperVal{i}{v}=\inf_{\strategyfor{i}}\; \sup_{\strategyfor{-i}}\; \cost{i}{\outcomefrom{\sigma_0,\strategyfor{i},\strategyfor{-i}}{v}}.
\]
Then, when for every vertex $v$, $\lowerVal{i}{v} = \upperVal{i}{v}$ holds, the game is said to be \emph{determined} and the value is $\ValPlayer{i}{v} = \upperVal{i}{v} = \lowerVal{i}{v}$. This is the case for our reachability games~\cite{DBLP:journals/jcss/BrihayeBGT21}. Then, for each vertex $v$, we define $\ValStar{v} = \ValPlayer{i}{v}$ for the player~$i$ such that $v \in V_i$.

\begin{definition}\label{def:visit-lambda-consistency}
A play $\pi=\pi_0\pi_1\dots$ is \emph{Visit \ValStarFunc{}-consistent for a set of players~$I \subseteq \Players$} if for all $i\in I$ and for all $n\in\N$, we have $\left(\pi_n \in V_i \land i \not\in \Visit{\pi_{<n}}\right) \Rightarrow \cost{i}{\pi_{\geq n}} \leq \ValStar{\pi_n}.$
\end{definition}

The next result is a direct corollary of \cite[Theorem~15]{DBLP:journals/jcss/BrihayeBGT21} stating the characterization of NE outcomes. We simply ignore values of player~$0$ to get the equivalent result for \fixedEquilibria{}s.

\begin{proposition}[\cite{DBLP:journals/jcss/BrihayeBGT21}]\label{proposition:fixedEquilibriaChar}
Let $\pi$ be a play. Then $\pi$ is the outcome of a \fixedEquilibria{} if and only if $\pi$ is Visit \ValStarFunc{}-consistent for the set of players~$\Players \ssetminus \{0\}$.
\end{proposition}

\section{Complexity of the CNS and (U)NCNV Problems}\label{appendix:other-results-nash-2p-parikh}

We study in this section the \prob{(U)NCNV} and \prob{CNS} problems (\cref{theorem:Nash-results}.\ref{thm-NCNV}-\ref{thm-CNS}).

Some results are easily obtained as there already exist similar approaches for closely related problems. This is the case for the \prob{CNS} problem. For the particular case of qualitative reachability, it is proved in~\cite{ConduracheFGR16} that the \prob{CNS} problem is \npComplete{} (by using automata techniques). This result is extended in~\cite{grandmont23} for arenas where the weights on the edges are all equal to $1$, by using~\cref{proposition:fixedEquilibriaChar}.

We also show two proofs of hardness for the \prob{CNS} and \prob{UNCNV} problems, that already hold when the environment is limited to one player. Then, we see how to adapt the work from~\cite{DBLP:conf/mfcs/BriceRB23}, where the authors proved the \conpComplete{}ness of the \prob{(U)NCNV} problems when all weights on the edges are equal to~$1$. The adaptation is nontrivial and requires the use of Parikh automata for the same reason as in~\cref{section:pareto}.

\begin{theorem}\label{theorem:cns-np-complete}
    The \probname{}{C}{N}{S} problem for reachability games is \npComplete{}, even with one-player environments.
\end{theorem}

To prove this theorem, we will use the characterization of $0$-fixed NE outcomes from \cref{appendix:ne-characterization}.

\begin{proof}[Proof of the upper bound of \Cref{theorem:cns-np-complete}]
    We use the same approach as in \cref{lem:witness_size_pareto} to show that, given a Visit \ValStarFunc{}-consistent play $\pi$ for the set of players~$\Players \ssetminus \{0\}$, we can build a Visit \ValStarFunc{}-consistent play $\pi' = \mu(\nu)^\omega$ for~$\Players \ssetminus \{0\}$ where $|\mu\nu|$ is polynomial in $|V|$ and $|\Players|$. Thus, we first compute in polynomial time the value of each vertex. Then we guess a lasso $\pi' = \mu(\nu)^\omega$ as above and use \cref{proposition:fixedEquilibriaChar} to check in time polynomial in $|V|$ and $|\Players|$ whether there exists a strategy $\strategyfor{0}$ such that $\pi$ is a \fixedEquilibriaStrategy{} since values and weights along $\pi$ are given in binary.
\end{proof}

We now prove the lower bound of \Cref{theorem:cns-np-complete}, that already holds for one-player environments, with a reduction from the \emph{bi-partition problem} which is known to be \npComplete{}~\cite{Karp1972}. Given $S = \{a_1,\dots,a_n\}$ a subset of $\N$, the bi-partition problem asks whether there exists $A,B \subseteq S$ such that $A$ and $B$ is a partition of $S$ and $\sum_{a_i \in A} a_i = \sum_{a_i \in B} a_i$.
Notice that if the answer to this problem is positive, then $T = \sum_{a_i \in S} a_i$ is even. Therefore we can suppose that $T$ is always even, which allows us to work with $T/2 \in \N$. 

\begin{figure}[htbp]
  \centering
  \begin{minipage}[b]{0.48\textwidth}
    \centering
    \begin{tikzpicture}[automaton,scale=0.6,every node/.style={scale=0.6},node distance=1.2]
      \node[initial,initial where=above,state,environment] (v0)           {$v_{0}$};
      \node[state,system]      (T1)  [left=1.6cm of v0]    {$L$};
      \node[state,system]      (v1)  [right=of v0]  {$v_{1}$};
      \node                    (dots)[right=of v1]  {$\dots$};
      \node[state,system]      (T)   [right=of dots] {$R$};

      \path (T1)  edge[loop above] node[above] {}       (T1)
            (v0)  edge             node[above] {$(0,T/2)$}     (T1)
                  edge             node[above] {$(0,0)$} (v1)
            (v1)  edge[bend left]  node[above] {$(a_1,0)$} (dots)
                  edge[bend right] node[below] {$(0,a_1)$} (dots)
            (dots)edge[bend left]  node[above] {$(a_n,0)$} (T)
                  edge[bend right] node[below] {$(0,a_n)$} (T)
            (T)   edge[loop above] node[above] {}          (T);
    \end{tikzpicture}
    \caption{Reduction for the \npHard{}ness for the \prob{CNS} problem with a one-player environment}
    \label{fig:np-hardness-NCS-weighted-2p}
  \end{minipage}
  \hfill
  \begin{minipage}[b]{0.48\textwidth}
    \centering
    \begin{tikzpicture}[automaton,scale=0.6,every node/.style={scale=0.55},node distance=1.2]
      \node[initial,initial where=above,state,environment] (v0)           {$v_{0}$};
      \node[state,system]      (T1)  [left=2.1cm of v0]    {$L$};
      \node[state,system]      (v1)  [right=of v0]  {$v_{1}$};
      \node                    (dots)[right=of v1]  {$\dots$};
      \node[state,system]      (T)   [right=of dots] {$R$};

      \path (T1)  edge[loop above] node[above] {}       (T1)
            (v0)  edge             node[above] {$(0,T/2 \cdot (2n-1))$}     (T1)
                  edge             node[above] {$(0,0)$} (v1)
            (v1)  edge[bend left]  node[above] {$(a_1,T)$} (dots)
                  edge[bend right] node[below] {$(0,T-a_1)$} (dots)
            (dots)edge[bend left]  node[above] {$(a_n,T)$} (T)
                  edge[bend right] node[below] {$(0,T-a_n)$} (T)
            (T)   edge[loop above] node[above] {}          (T);
    \end{tikzpicture}
    \caption{Reduction for the \npHard{}ness for the \coprob{UNCNV} problem with a one-player environment}
    \label{fig:np-hardness-UNCNV-weighted-2p}
  \end{minipage}
\end{figure}

\begin{proof}[Proof of the lower bound of \cref{theorem:cns-np-complete}]
\Cref{fig:np-hardness-NCS-weighted-2p} shows the arena used for the reduction from the bi-partition problem to the \prob{CNS} problem with the threshold $c = T/2$. The target sets for player~$0$ and player~$1$ are respectively $\reach{0} = \{R\}$ and $\reach{1} = \{L,R\}$.

Suppose that the instance of the bi-partition problem is positive. Player~$0$ can set a strategy $\strategyfor{0}$ such that the play $\pi = v_0v_1\dots (R)^\omega$ satisfies $\cost{0}{\pi} = \cost{1}{\pi} = T/2$. Notice that $\cost{0}{\pi} \leq c$ and that $\pi$ is a \fixedEquilibriaStrategy{} outcome, because if player~$1$ deviates at $v_0$, then the resulting $\pi' = v_0 (L)^\omega$ has $\cost{1}{\pi'} = T/2$.

Suppose now that there exists a \fixedEquilibria{} outcome $\pi$ with $\cost{0}{\pi} \leq T/2$. Necessarily, it goes to the right and reaches $R$, since otherwise $\cost{0}{\pi} = +\infty$. Moreover, $\cost{1}{\pi} \geq T/2$ as $\cost{0}{\pi} + \cost{1}{\pi} = T$. Finally, since $\pi$ is a \fixedEquilibria{} outcome, player~$1$ has no incentive to deviate to the left from $v_0$, i.e., $\cost{1}{\pi} \leq T/2$. It follows that  $\cost{0}{\pi} = \cost{1}{\pi} = T/2$.
\end{proof}

With small adaptations to the previous reduction, we get the \conpHard{}ness of the \prob{UNCNV} problem.

\begin{corollary}
    The \probname{U}{NC}{N}{V} problem for reachability games is \conpHard{}, even with one-player environments.
\end{corollary}

\begin{proof}
    We use the same kind of reduction as for the \prob{CNS} problem. The game used for the reduction to the \coprob{UNCNV} problem is illustrated in \cref{fig:np-hardness-UNCNV-weighted-2p},  such that $c = T/2 - 1$ and $\reach{0} = \reach{1} = \{L,R\}$.

    If there exists a partition $\{A,B\}$ of $S$ such that $\sum_{A} a_i = \sum_{B} a_i = T/2$, then we have a play $\pi$ going to the right where $\cost{0}{\pi} = T/2$, $\cost{1}{\pi} = nT - T/2 = T/2 (2n-1)$, i.e., $\pi$ is a \fixedEquilibria{} outcome, so we have finished. For the other direction, if there exists a \fixedEquilibria{} outcome $\pi$, this play must go to the right otherwise the cost of player~$0$ would be $0 \leq c$. Thus, $\cost{0}{\pi} \geq T/2$ and by definition of NE, $\cost{1}{\pi} \leq T/2 (2n-1)$. Let $a,b$ such that $\cost{0}{\pi} = a$ and $\cost{1}{\pi} = nT - b$, clearly we have $a + b = T$ by construction, and $a \geq T/2$ and $b \geq T/2$ by hypothesis, which means that $a = b = T/2$.
\end{proof}

The next result generalizes a result from~\cite{DBLP:conf/mfcs/BriceRB23}, where the authors showed the \conpComplete{}ness for the \prob{(U)NCNV} problem when the weights on the edges are all equal to $1$. Having integer weights encoding in binary induces a new difficulty: we cannot characterize \fixedEquilibriaStrategy{} outcomes by using polynomial-size lassos. Indeed, \cref{exemple:pareto-size-expo}, used with PO rationality, is also applicable in the context of NE rationality, with one-player environments.

To prove the result, we will guess a succinct representation of a solution $\pi$ of exponential size as we did in \Cref{section:pareto} thanks to Parikh automata. We will not dive into all tricky technical details and stick to an intuition based on the lasso-sufficiency result of such a solution as presented in~\cite{DBLP:conf/mfcs/BriceRB23}.

\begin{theorem}
    The (Universal) \probname{}{NC}{N}{V} problem for reachability games is \conpComplete{}.
\end{theorem}

\begin{proof}[Proof (sketch).]
    The lower bound is already proved in from~\cite{DBLP:conf/mfcs/BriceRB23}. Let us consider the upper bound. As for other verification problems, we use the complementary problem and show that the \coprob{(U)NCNV} problem belongs to \np{}. We suppose in the following that we already have a game $\game{}$ which is the product of the given game and the (non)deterministic Mealy machine of player~$0$.

    The main goal is to guess a play $\pi$ outcome of a \fixedEquilibriaStrategy{} such that $\cost{0}{\pi} > c$. By~\cite{DBLP:conf/mfcs/BriceRB23}, we deduce that it suffices to guess a lasso $\pi = \mu(\nu)^\omega$ such that $\cost{i}{\pi} \in \{0,\dots,(2t+1)|V|W,+\infty\}$ for every player~$i$. We use again Parikh automata to guess such a lasso in polynomial time (by using markers and portions as in the proofs of \Cref{subsec:UpperBounds}). 

    We need two kinds of markers. The first kind of markers are vertices that aim to represent the first occurrence of a target set along $\pi$ (including $\reach{0}$ if $\pi$ visits $\reach{0}$). To verify that $\pi$ is a \fixedEquilibriaStrategy{}, we also need markers that are vertices with a finite value, with the aim that every time $\pi_m \in V_i$ is a marker with $\ValStar{\pi_m} < +\infty$, we must have $\cost{i}{\pi_{\geq m}} \leq \ValStar{\pi_m}$ (see \Cref{proposition:fixedEquilibriaChar}). The total number of markers is polynomial as it is bounded by $|\Players| + |V|$.\footnote{Notice that, as in \Cref{subsec:UpperBounds}, we also need markers for the initial vertex and the first vertex of the cycle $\nu$.}

    The algorithm works as follows. First, it computes in polynomial time the value of each vertex. Then, it guesses whether the lasso will visit $\reach{0}$ or not, and guess its markers as explained above. For each portion between any two consecutive markers, it also guesses a tuple of weights in $\{0,\dots,(2t+1)|V|W,+\infty\}$ that aim to represent the player weights for this portion. Finally, the algorithm performs a series of checks in polynomial time: It verifies that the markers are vertices with finite values or belonging to target sets (including $\reach{0}$ if it guessed a visit to $\reach{0}$). It verifies the existence of the portions with the guessed weight tuples (excluding the visit to certain vertices to remain consistent with the markers). It finally checks that the characterization of \cref{proposition:fixedEquilibriaChar} holds and that the cost for player~$0$ is at least $c$ in case of a visit to $\reach{0}$.
\end{proof}

\section{Complexity of the NCNS Problem with One-Player Environments}\label{appendix:exp-algo-2p-ncns}

This section aims at proving \cref{theorem:ncns-2players-pspace-hard-subset-sum}, i.e., showing that the \prob{NCNS} problem is \pspaceHard{} and in \exptime{} when the environment is composed of only one player, player~$1$ (\cref{theorem:Nash-results}.\ref{thm-NCNS} with one-player environments). Imposing this restriction drastically simplifies the problem.

\begin{restatable}{theorem}{ncnstwoplayers}
\label{theorem:ncns-2players-pspace-hard-subset-sum}
    The \prob{NCNS} problem with a one-player environment is \pspaceHard{} and belongs to \exptime{}.
\end{restatable}

\subsection{PSPACE-hardness}

To show the \pspaceHard ness result stated in \Cref{theorem:ncns-2players-pspace-hard-subset-sum}, we use a reduction from a problem called the \emph{subset-sum game problem}. An instance of this problem is a formula $\psi$ as
\[
\forall P_{1} \in \{A_1,B_1\},\, \exists P_{2} \in \{E_1,F_1\}, \dots, \forall P_{2n-1} \in \{A_n,B_n\},\, \exists P_{2n} \in \{E_n,F_n\},\; \sum_{i = 1}^{2n} P_i = T,
\]
where $T,A_i,B_i,E_i$ and $F_i$ are natural integers encoded in binary, for all $i \in \{1,\dots,n\}$. We can view this formula as a game with two players, the existential and the universal players, where the existential player wins if the formula is satisfied, i.e., if the valuations of each picked variable $P_i$ sum up to $T$. This problem is \pspaceComplete{}~\cite{game-subset-sum}.

Notice that we can suppose w.l.o.g.\ that every integer $A_i,B_i,E_i$ and $F_i$ is $\leq T$. If there exists $i$ such that either $A_i > T$ or $B_i > T$, or $E_i > T$ and $F_i > T$, then the formula is never satisfied. Indeed either the universal player selects a valuation exceeding $T$ for his variable, or the existential player has no choice but to exceed $T$, and thus $\sum_{i = 1}^{2n} P_i > T$. And when $E_i > T$ and $F_i \leq T$ (resp.\ $E_i \leq T$ and $F_i > T$), the existential player must select $F_i$ (resp.\ $E_i$) to hope to win, in which case we modify $\psi$ by replacing $E_i$ by $F_i$ (resp.\ $F_i$ by $E_i$).

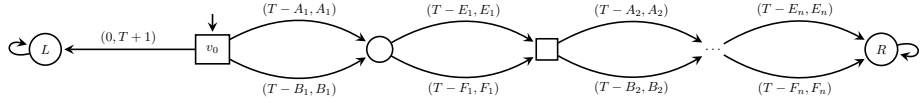
\begin{figure}
\centering
\begin{tikzpicture}[automaton,scale=0.6,every node/.style={scale=0.5},node distance=2.2]
  \node[initial,initial where=above,state,environment] (v00)           {$v_0$};
  \node[state,system]      (L)  [left=of v00]    {$L$};
  \node[state,system]      (v01) [right=of v00]  {~~~};
  \node[state,environment] (v10) [right=of v01]  {};
  \node                    (dots)[right=of v10]  {$\dots$};
  \node[state,system]      (R)   [right=of dots] {$R$};

  \path (L)   edge[loop left] node[above] {}       (L)
        (v00) edge             node[above] {$(0,T+1)$}     (L)
              edge[bend left]  node[above] {$(T-A_1,A_1)$} (v01)
              edge[bend right] node[below] {$(T-B_1,B_1)$} (v01)
        (v01) edge[bend left]  node[above] {$(T-E_1,E_1)$} (v10)
              edge[bend right] node[below] {$(T-F_1,F_1)$} (v10)
        (v10) edge[bend left]  node[above] {$(T-A_2,A_2)$} (dots)
              edge[bend right] node[below] {$(T-B_2,B_2)$} (dots)
        (dots)edge[bend left]  node[above] {$(T-E_n,E_n)$} (R)
              edge[bend right] node[below] {$(T-F_n,F_n)$} (R)
        (R)   edge[loop right] node[above] {}       (R);
\end{tikzpicture}
\caption{Reduction from the subset-sum game problem to the \prob{NCNS} problem (one-player env.).}
\label{fig:pspace-hardness-weighted}
\end{figure}

\begin{proof}[Proof of~\Cref{theorem:ncns-2players-pspace-hard-subset-sum} (Hardness)]
    We use a reduction from the subset-sum game problem and construct a game $\game$ as in \cref{fig:pspace-hardness-weighted} and set the threshold $c = (2n-1)T$. Player~$0$, owning the circle vertices, represents the existential player of the subset-sum game, while player~$1$, owning the square vertices, represents the universal player. The initial vertex is $v_0$ and the target sets are $\reach{0} = \{R\}$ and $\reach{1} = \{R,L\}$. The two weight functions are indicated on the edges. Notice that each weight is positive or zero, from our previous remark. At the initial vertex, player~$1$ either chooses the edge towards $L$ and gets a cost of $T+1$ while imposing an infinite cost to player~$0$, or he simulates with player~$0$ a play of the subset-sum game by alternatively picking an integer, until reaching the vertex $R$.

    If the existential player has a strategy to satisfy $\psi$, then player~$0$ can simulate this strategy in $\game{}$, which is a solution $\strategyfor{0}$ to the \prob{NCNS} problem. Indeed, any play consistent with $\strategyfor{0}$ either reaches $L$ with a cost of $T+1$ for player~$1$ or reaches $R$ with a cost of $T$ for player~$1$. So, the $\strategyfor{0}$-fixed NE outcomes are those reaching $R$, ensuring a cost of $(2n-1)T = c$ for player~$0$.

    For the other implication, notice that for any play $\pi$ reaching $R$, we have $\cost{0}{\pi} + \cost{1}{\pi} = 2nT$. Let $\strategyfor{0}$ be a solution to the \prob{NCNS} problem, that is, such that the outcome $\pi$ of every \fixedEquilibriaStrategy{} satisfies $\cost{0}{\pi} \leq (2n-1)T$. Then, by our remark, $\cost{1}{\pi} \geq T$. However, we cannot have $\cost{1}{\pi} \geq T+1$, since otherwise the play $\pi$ visiting $L$ would be the outcome of a \fixedEquilibriaStrategy{} with $\cost{0}{\pi} = +\infty$, in contradiction with $\strategyfor{0}$ being a solution to the \prob{NCNS} problem. We conclude that $\cost{1}{\pi}=T$, and thus $\cost{0}{\pi} = (2n-1)T$, i.e., the existential player can simulate the strategy $\strategyfor{0}$ to satisfy the formula $\psi$.
\end{proof}

\subsection{EXPTIME-Membership}

We now prove the \exptime{}-membership result stated in \cref{theorem:ncns-2players-pspace-hard-subset-sum}.

The heart of the decidability of the \prob{NCNS} problem with player~$0$, the system, and player~$1$, the environment, comes with the following result.

\begin{lemma}\label{lem:2p-ne-same-cost}
    Let $\strategyfor{0} \in \Sigma_0$ and $d = \min\{\cost{1}{\pi} \mid \pi$ play consistent with $\strategyfor{0}\}$. Then for all play $\pi$ consistent with $\strategyfor{0}$, $\pi$ is the outcome of a \fixedEquilibriaStrategy{} if and only if $\cost{1}{\pi} = d$.

    Moreover, when $d = +\infty$, every play consistent with $\strategyfor{0}$ is a \fixedEquilibriaStrategy{}.
\end{lemma}

\begin{proof}
    Let $\pi$ be the outcome of a \fixedEquilibriaStrategy{}. By minimality of $d$, we have $\cost{1}{\pi} \geq d$, and by definition of an NE, $\cost{1}{\pi} \leq \cost{1}{\pi'}$ for all plays $\pi'$ consistent with $\strategyfor{0}$. In particular, it is also verified for a play $\pi'$ such that $\cost{1}{\pi'} = d$. Therefore, we get $\cost{1}{\pi} = d$.

    Conversely, let $\pi$ be such that $\cost{1}{\pi} = d$. By minimality of $d$, we have $\cost{1}{\pi'} \geq \cost{1}{\pi}$ for all plays $\pi'$ consistent with $\strategyfor{0}$, which means that $\pi$ is the outcome of an NE.

    The remark when $d = +\infty$ is a direct consequence of the minimality of $d$.
\end{proof}

\cref{lem:2p-ne-same-cost} indicates that we can associate a unique minimum $d \in \N\cup\{+\infty\}$ to each strategy $\strategyfor{0}$ of player~$0$. From now on, we will often do this link between such a strategy $\strategyfor{0}$ and its minimum cost $d$ for player~$1$, written explicitly $d_{\strategyfor{0}}$ when needed.

Another important result is a characterization of solutions to the \prob{NCNS} problem, thanks to the concept of \emph{witness}. This idea already appears in~\cite{NoncoopSynth_Meanpayoff_2players} for mean-payoff games with a one-player environment.

\begin{definition}[Witness]\label{def:2p-witness}
    Let $c\in \N$. A \emph{$c$-witness} $\pi$ is a play such that $\cost{0}{\pi} \leq c$ and for every deviation\footnote{We recall that $hv$ is a deviation of $\pi$ if $h$ is a prefix of $\pi$ but not $hv$.} $hv$ of $\pi$, $v$ is winning for player~$0$ in the zero-sum game $(\arena{},\Omega^{(hv)})$ such that
    \[
    \Omega^{(hv)} = \{\pi' \in \Plays(v) \mid \cost{0}{h\pi'} \leq c \text{ or } \cost{1}{h\pi'} > d\}
    \]
    with $d = \cost{1}{\pi}$.
\end{definition}
In other words, $\pi$ is a $c$-witness if $\cost{0}{\pi} \leq c$ and when player~$1$ is deviating from $\pi$, player~$0$ can either ensure that his own cost is smaller than $c$ or punish player~$1$ by imposing him a cost strictly greater than $d$.

In other words, $\pi$ is a $c$-witness if $\cost{0}{\pi} \leq c$ and when player~$1$ is deviating from $\pi$, player~$0$ can either ensure that his own cost is smaller than $c$ or punish player~$1$ by imposing him a cost strictly greater than $d$.\footnote{In the objective $\Omega^{hv}$, we need to subtract the weights $w_0(hv)$ and $w_1(hv)$ as we start to play from $v$, after history $h$, in the game $(\arena{},\Omega^{hv})$.}

\begin{lemma}\label{lem:2p-witness-characterize-solutions} Let $c \in \N$ be a threshold. Then there exists a solution $\strategyfor{0}$ to the \prob{NCNS} problem with the threshold $c$ if and only if there exists a $c$-witness.
\end{lemma}

\begin{proof}
Suppose that there exists a $c$-witness $\pi$. We can construct a strategy $\strategyfor{0}$ such that it produces $\pi$ (i.e., $\pi$ is consistent with $\strategyfor{0}$), and for all deviations $hv$ of $\pi$, $\strategyfor{0}$ acts as a winning strategy of player~$0$ in the zero-sum game $(\arena{},\Omega^{(hv)})$ from the history $hv$. Let us show that this strategy $\strategyfor{0}$ is a solution to the \prob{NCNS} problem. For each \fixedEquilibriaStrategy{} $(\strategyfor{0},\strategyfor{1})$, we have by definition of NE that $\cost{1}{\outcomefrom{\strategyfor{0},\strategyfor{1}}{v_0}} \leq \cost{1}{\pi}$. Since $\pi$ is a $c$-witness, it follows that $\cost{0}{\outcomefrom{\strategyfor{0},\strategyfor{1}}{v_0}} \leq c$.

    Let us now suppose that there exists a solution $\strategyfor{0}$ to the \prob{NCNS} problem. We extract by \Cref{lem:2p-ne-same-cost} a \fixedEquilibriaStrategy{} outcome $\pi$ such that $\cost{0}{\pi} = d_{\strategyfor{0}}$ and claim that $\pi$ is a $c$-witness. First, notice that $\cost{0}{\pi} \leq c$ as $\strategyfor{0}$ is solution. Then, for all strategies $\strategyfor{1}$ of player~$1$ such that $\pi' = \outcomefrom{\strategyfor{0},\strategyfor{1}}{v_0}$ is different from $\pi$, we have $\cost{1}{\pi'} \geq d_{\strategyfor{0}}$ (again by \Cref{lem:2p-ne-same-cost}). So either $\cost{1}{\pi'} > d_{\strategyfor{0}}$ or $\cost{1}{\pi'} = d_{\strategyfor{0}}$, i.e., $\strategyfor{1}$ is a \fixedEquilibriaStrategy{}, and thus, $\cost{0}{\pi'} \leq c$ as $\strategyfor{0}$ is solution. Thus for any deviation $hv$ of $\pi$, $\strategyfor{0}$ is a winning strategy for player~$0$ in the game $(\arena,\Omega^{(hv)})$.
\end{proof}

Despite this useful characterization, notice that a $c$-witness is not necessarily the outcome of an NE, because player~$1$ might have a smaller cost elsewhere, but in that case, we keep a cost of player~$0$ smaller than $c$ along any other play profitable for player~$1$.

\begin{lemma}\label{lem:lasso-sufficiency-generic}
    Let $c \in \N$ and let $\pi$ be a $c$-witness. Then there exists a $c$-witness $\pi'$ such that either $\cost{1}{\pi'} = +\infty$ or $\cost{1}{\pi'} \leq 2|V|W$.
\end{lemma}

\begin{proof}
Let $\pi=\pi_0\pi_1\dots$ be a $c$-witness such that $\cost{1}{\pi} = d$. Let $k_i = \inf\{n \in \N \mid \pi_n \in \reach{i}\}$ for $i \in \{0,1\}$, where $k_0 < +\infty$ by hypothesis. If $\cost{1}{\pi} = +\infty$, then $\pi' = \pi$. Otherwise, we have $k_1 < +\infty$.

Suppose first that $k_0 \leq k_1$. 
\begin{itemize}
    \item If there is a cycle $\pi_{[n,m[}$ before $k_0$, i.e., $m < k_0$, then we can remove it from $\pi$ and still keep a $c$-witness $\pi'$. Indeed, $\cost{i}{\pi'} = \cost{i}{\pi} - w_i(\pi_{[n,m]})$, thus $\cost{0}{\pi'} \leq c$ and $\cost{1}{\pi'} = d - w_1(\pi_{[n,m]})$. Moreover, consider any deviation $hv$ of $\pi$ such that $h$ contains the cycle $\pi_{[n,m[}$. Then, from a winning strategy for player~$0$ from $v$ in $(\arena,\Omega^{(hv)})$ with $\Omega^{(hv)} = \{\rho \in \Plays(v) \mid \cost{0}{h\rho} \leq c \text{ or } \cost{1}{h\rho} > d \}$ and for $h'v$ the deviation of $\pi'$ obtained by removing $\pi_{[n,m[}$ from $h$, we get that $v$ is winning in $(\arena,\Omega^{(h'v)})$ with $\Omega^{(h'v)} = \left\{\rho' \in \Plays(v) \mid \cost{0}{h'\rho'} \leq c - w_0(\pi_{[n,m]})\text{ or } \cost{1}{h'\rho'} > d - w_1(\pi_{[n,m]})\right\}$. 
    \item If there exists a cycle between $k_0$ and $k_1$, we can similarly remove it and again keep a $c$-witness. In this case, notice that $\Omega^{(hv)}$ is always satisfied as $h$ visits $\reach{0}$ and thus $\cost{0}{h\rho} \leq c$.
\end{itemize}
Suppose now that $k_1 < k_0$, then we remove cycles before $k_1$ as we did previously with cycles before $k_0$.

In both cases, by repeatedly removing cycles as explained above, we obtain a $c$-witness $\pi'$ such that $\cost{1}{\pi'} \leq 2|V|W$.
\end{proof}

Witnesses are useful, nevertheless, it would be easier to deal with qualitative objectives instead of quantitative ones. We will show that working with bounded objectives as $\Omega^{(hv)}$ amounts to working with qualitative objectives in a specific arena, called \emph{extended}, where we add accumulated weights in the vertices as well as players who visited their target set. This extended arena is defined for any number $t$ of players in the environment, but used with $t=1$ in the proof of the \exptime{} membership of~\Cref{theorem:ncns-2players-pspace-hard-subset-sum}.

\begin{definition}[Extended arena]
    Let $\game{}=(\arena, (\reach{i})_{i \in \Players})$ be a reachability game with $\arena = (V,E,\Players,(V_i)_{i\in\Players},(w_i)_{i\in\Players})$ and $B = (B_i)_{i\in \Players} \in \N^{t+1}$ be a family of bounds, one for each player. The extended arena $\arena{}_{B}$ of $\game$ is the tuple $(V_B,E_B,\Players,(V_{B,i})_{i \in \Players})$, where
    \begin{itemize}
        \item a vertex $v' \in V_B$ is in the form $v' = (v,c_0,\dots,c_t,F)$ with $v \in V$, $c_i \in \{0,1,\dots,B_i,+\infty\}$ for all $i \in \{0,\dots, t\}$, and $F \subseteq \Players$,
        \item $((v,c_0,\dots,c_t,F),(v',c'_0,\dots,c'_t,F'))\in E_B$ if
        \begin{itemize}
            \item $(v,v')\in E$,
            \item $c_i =
            \begin{cases}
                c_i               & \text{ if $i \in F$ or $c_i = +\infty$}\\
                c_i + w_i((v,v')) & \text{ if $i \not\in F$, $c_i < +\infty$ and $c_i + w_i((v,v')) \leq B_i$}\\
                +\infty           & \text{ otherwise,}
            \end{cases}$
            \item $F' = F \cup \{i \in \Players \mid v' \in \reach{i}\}$,
        \end{itemize}
        \item $V_{B,i} = \{(v,c_0,\dots,c_t,F)\in V_B \mid v\in V_i\}$ for all $i \in \{0,\dots,t\}$.
    \end{itemize}
    Moreover, if $v_0$ is the initial vertex of $\arena$, then $v_0'=(v_0,0,\dots,0,F)$ is the initial vertex of $\arena_B$ with $F = \{i \in \Players \mid v_0 \in \reach{i}\}$.
\end{definition}

In this definition, 
the component $c_i$ is the cumulative weight of player~$i$. It is set to $+\infty$ when it exceeds $B_i$. The set $F$ is the set of all players who visited their target set. When $i \in F$, the component $c_i$ is frozen and no longer changes. Notice that the new arena $\arena_B$ is not weighted.

We clearly have a bijection between $\Plays_\arena(v_0)$ and $\Plays_{\arena_B}(v'_0)$.
The useful difference is that the extended arena can express bounded objectives from the original game as qualitative objectives in the extended arena, e.g., saying that a play $\pi \in \Plays_\arena(v_0)$ has $\cost{i}{\pi}\leq d$, when $B \geq d$, is equivalent to saying that its corresponding play in $\Plays_{\arena_B}(v'_0)$ visits some vertex $(v,c_0,\dots,c_t,F)$ where $c_i \leq d$ and $i \in F$. For a play $\pi\in\Plays_{\arena_B}(v'_0)$, we denote $\pi_{|V}$ the corresponding play in $\Plays_\arena(v_0)$ by taking its projection on the $V$-component. We also denote by $|\arena_B|$ the size of the extended arena equal to $|V| \cdot 2^{t+1} \cdot \prod_{i\in \Players}(B_i+2)$.

We are now able to prove~\Cref{theorem:ncns-2players-pspace-hard-subset-sum}.

\begin{proof}[Proof of \Cref{theorem:ncns-2players-pspace-hard-subset-sum} (Easiness)]
By~\Cref{lem:2p-witness-characterize-solutions}, it is enough to check the existence of a $c$-witness $\pi$ in \exptime{}. Moreover, by \cref{lem:lasso-sufficiency-generic}, we can assume that $\cost{1}{\pi} \in \{0,\dots,2|V|W,+\infty\}$.

We first treat the particular case where there exists a strategy $\strategyfor{0}$ for player~$0$ such that all plays $\rho \in \Plays_\arena(v_0)$ consistent with $\strategyfor{0}$ satisfy $\cost{0}{\rho} < c+1$. The existence of such a strategy (winning in a two-player zero-sum game with a bounded reachability game) can be checked in polynomial time~\cite{DBLP:journals/mst/KhachiyanBBEGRZ08}.

If this algorithm concludes the nonexistence of such a strategy $\strategyfor{0}$, we know by~\Cref{lem:2p-ne-same-cost} that there is no solution $\strategyfor{0}$ to the \prob{NCNS} problem such that $d_{\strategyfor{0}} = +\infty$.
In that case, we consider the next second algorithm: Iterate for each $d \in \{0,\dots,2|V|W\}$ the following steps
\begin{enumerate}
    \item Compute the extended arena $\arena{}_{(c,d)}$ and build two objectives:
    \begin{itemize}
        \item $\Omega_0=\{\pi\in\Plays_{\arena_{(c,d)}}(v'_0) \mid \pi \text{ visits a state } (v,c_0,c_1,F) \text{ with } c_0 \leq c \text{ and } 0 \in F\}$,
        \item $\Omega_1=\{\pi\in\Plays_{\arena_{(c,d)}}(v'_0) \mid \pi \text{ visits a state } (v,c_0,c_1,F) \text{ with } c_1 > d\}$.
    \end{itemize}
    \item Consider the zero-sum game $(\arena_{(c,d)},\Omega)$ with $\Omega = \Omega_0 \cup \Omega_1$, and compute the set $W_0$ of vertices that are winning for player~$0$ for $\Omega$. As $\Omega$ is a (qualitative) reachability game, the set $W_0$ can be computed in time polynomial in $|\arena_{(c,d)}|$ (see e.g. \cite{DBLP:conf/dagstuhl/2001automata}), i.e., in time polynomial in $|V|$ and exponential in $c$ and $W$ since $c$ and $W$ are encoded in binary.
    \item Construct the subarena of $\arena_{(c,d)}$ restricted to $W_0$, and check the existence of a path $\pi$ (the required $c$-witness) visiting a vertex $(v,c_0,c_1,F)$ where $c_0 \leq c$ and $0 \in F$ as well as a vertex $(v',c_0',c_1',F')$ with $c_1' = d$ and $1 \in F'$. This can be done by a graph traversal algorithm in time polynomial in $|W_0|$, so in exponential time. If such a path exists, stop the algorithm, otherwise proceed to the next iteration.
\end{enumerate}

This completes the proof since every step is in \exptime{} and there is an exponential number of iterations.
\end{proof}

Notice that with a little more work, we could find the smallest threshold $c$ for which there exists a solution $\strategyfor{0}$ to the \prob{NCNS} problem if such a solution exists.

\section{Variant of the NCNS Problem}\label{appendix:2p-ncns-exptime-complete-bounded-reach}

To better understand how difficult the \prob{NCNS} problem is, we look at the variant where the rational NE responses of the multi-player environment aim to ensure costs bounded by a given threshold rather than minimizing these costs. This is a perspective studied in~\cite{DBLP:conf/ijcai/RajasekaranBansalV23} in the case of NEs for discounted-sum objectives. The authors call those objectives \emph{satisficing objectives}, establishing a duality with \emph{optimization objectives}. For some reward functions such as discounted-sum, players might deviate for arbitrarily minuscule rewards. Hence, satisficing objectives prevent this behavior since a player is either happy or unhappy, there is no more optimization.

We prove that, when the players have such bounded reachability objectives, the \prob{NCNS} problem becomes \exptimeComplete{}. This result holds already for a one-player environment.

\ncnsboundedexptime*

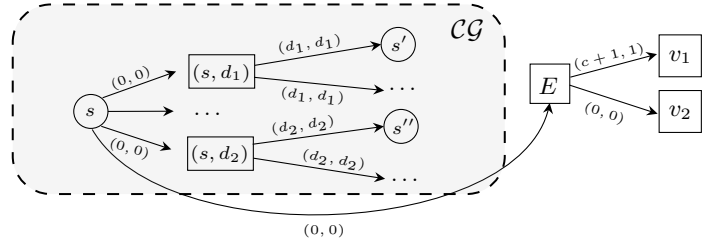
\begin{figure}
    \centering
    \begin{tikzpicture}[x=0.75pt,y=0.75pt,yscale=-1,xscale=1]
    \draw    (227.15,112.38) .. controls (263.54,175.23) and (440.75,167.12) .. (455.12,101.08) ;
    \draw [shift={(455.51,99.06)}, rotate = 99.34] [fill={rgb, 255:red, 0; green, 0; blue, 0 }  ][line width=0.08]  [draw opacity=0] (5.36,-2.57) -- (0,0) -- (5.36,2.57) -- (3.56,0) -- cycle    ;
    \draw   (445.73,79.48) -- (465.52,79.48) -- (465.52,99.27) -- (445.73,99.27) -- cycle ;
    \draw   (510.13,64.8) -- (531.2,64.8) -- (531.2,85.87) -- (510.13,85.87) -- cycle ;
    \draw   (510.13,92.8) -- (531.2,92.8) -- (531.2,113.87) -- (510.13,113.87) -- cycle ;
    \draw    (465.83,86.07) -- (506.17,75.26) ;
    \draw [shift={(509.07,74.48)}, rotate = 165] [fill={rgb, 255:red, 0; green, 0; blue, 0 }  ][line width=0.08]  [draw opacity=0] (5.36,-2.57) -- (0,0) -- (5.36,2.57) -- (3.56,0) -- cycle    ;
    \draw    (465.43,90.07) -- (506.55,103.55) ;
    \draw [shift={(509.4,104.48)}, rotate = 198.15] [fill={rgb, 255:red, 0; green, 0; blue, 0 }  ][line width=0.08]  [draw opacity=0] (5.36,-2.57) -- (0,0) -- (5.36,2.57) -- (3.56,0) -- cycle    ;
    \draw  [fill={rgb, 255:red, 155; green, 155; blue, 155 }  ,fill opacity=0.1 ][dash pattern={on 4.5pt off 4.5pt}][line width=0.75]  (187.47,68.57) .. controls (187.47,58.04) and (196.01,49.5) .. (206.54,49.5) -- (413.99,49.5) .. controls (424.53,49.5) and (433.07,58.04) .. (433.07,68.57) -- (433.07,125.79) .. controls (433.07,136.33) and (424.53,144.87) .. (413.99,144.87) -- (206.54,144.87) .. controls (196.01,144.87) and (187.47,136.33) .. (187.47,125.79) -- cycle ;
    \draw   (217.93,103.31) .. controls (217.93,98.6) and (221.76,94.78) .. (226.47,94.78) .. controls (231.18,94.78) and (235.01,98.6) .. (235.01,103.31) .. controls (235.01,108.03) and (231.18,111.85) .. (226.47,111.85) .. controls (221.76,111.85) and (217.93,108.03) .. (217.93,103.31) -- cycle ;
    \draw    (232.34,96.44) -- (266.21,84.22) ;
    \draw [shift={(269.03,83.2)}, rotate = 160.15] [fill={rgb, 255:red, 0; green, 0; blue, 0 }  ][line width=0.08]  [draw opacity=0] (5.36,-2.57) -- (0,0) -- (5.36,2.57) -- (3.56,0) -- cycle    ;
    \draw    (231.68,110.78) -- (263.79,121.45) ;
    \draw [shift={(266.63,122.4)}, rotate = 198.39] [fill={rgb, 255:red, 0; green, 0; blue, 0 }  ][line width=0.08]  [draw opacity=0] (5.36,-2.57) -- (0,0) -- (5.36,2.57) -- (3.56,0) -- cycle    ;
    \draw    (235.01,103.31) -- (264.83,103.21) ;
    \draw [shift={(267.83,103.2)}, rotate = 179.8] [fill={rgb, 255:red, 0; green, 0; blue, 0 }  ][line width=0.08]  [draw opacity=0] (5.36,-2.57) -- (0,0) -- (5.36,2.57) -- (3.56,0) -- cycle    ;
    \draw   (275.43,74.8) -- (308.23,74.8) -- (308.23,91.6) -- (275.43,91.6) -- cycle ;
    \draw    (308.23,80.14) -- (368.41,70.09) ;
    \draw [shift={(371.37,69.6)}, rotate = 170.52] [fill={rgb, 255:red, 0; green, 0; blue, 0 }  ][line width=0.08]  [draw opacity=0] (5.36,-2.57) -- (0,0) -- (5.36,2.57) -- (3.56,0) -- cycle    ;
    \draw    (307.9,86.14) -- (369.92,92.49) ;
    \draw [shift={(372.9,92.8)}, rotate = 185.85] [fill={rgb, 255:red, 0; green, 0; blue, 0 }  ][line width=0.08]  [draw opacity=0] (5.36,-2.57) -- (0,0) -- (5.36,2.57) -- (3.56,0) -- cycle    ;
    \draw    (307.23,119.81) -- (369.79,111.24) ;
    \draw [shift={(372.77,110.83)}, rotate = 172.2] [fill={rgb, 255:red, 0; green, 0; blue, 0 }  ][line width=0.08]  [draw opacity=0] (5.36,-2.57) -- (0,0) -- (5.36,2.57) -- (3.56,0) -- cycle    ;
    \draw    (307.57,126.81) -- (370.4,136.73) ;
    \draw [shift={(373.37,137.2)}, rotate = 188.97] [fill={rgb, 255:red, 0; green, 0; blue, 0 }  ][line width=0.08]  [draw opacity=0] (5.36,-2.57) -- (0,0) -- (5.36,2.57) -- (3.56,0) -- cycle    ;
    \draw   (371.37,69.6) .. controls (371.37,64.89) and (375.19,61.06) .. (379.9,61.06) .. controls (384.62,61.06) and (388.44,64.89) .. (388.44,69.6) .. controls (388.44,74.31) and (384.62,78.14) .. (379.9,78.14) .. controls (375.19,78.14) and (371.37,74.31) .. (371.37,69.6) -- cycle ;
    \draw   (372.77,110.83) .. controls (372.77,106.12) and (376.59,102.3) .. (381.3,102.3) .. controls (386.02,102.3) and (389.84,106.12) .. (389.84,110.83) .. controls (389.84,115.55) and (386.02,119.37) .. (381.3,119.37) .. controls (376.59,119.37) and (372.77,115.55) .. (372.77,110.83) -- cycle ;
    \draw   (274.77,116.13) -- (307.57,116.13) -- (307.57,132.93) -- (274.77,132.93) -- cycle ;

    \draw (448.67,85.13) node [anchor=north west][inner sep=0.75pt]  [font=\small] [align=left] {$E$};
    \draw (513.73,71.4) node [anchor=north west][inner sep=0.75pt]  [font=\small] [align=left] {$v_{1}$};
    \draw (513.73,99.4) node [anchor=north west][inner sep=0.75pt]  [font=\small] [align=left] {$v_{2}$};
    \draw (465,76.31) node [anchor=north west][inner sep=0.75pt]  [font=\tiny,rotate=-345.28] [align=left] {$(c+1,1)$};
    \draw (473.02,96.19) node [anchor=north west][inner sep=0.75pt]  [font=\tiny,rotate=-18.17] [align=left] {$(0,0)$};
    \draw (331.03,158.6) node [anchor=north west][inner sep=0.75pt]  [font=\tiny] [align=left] {$(0,0)$};
    \draw (405.02,56.13) node [anchor=north west][inner sep=0.75pt]  [font=\normalsize] [align=left] {$\mathcal{CG}$};
    \draw (222.33,100.53) node [anchor=north west][inner sep=0.75pt]  [font=\scriptsize] [align=left] {$s$};
    \draw (276.63,102.73) node [anchor=north west][inner sep=0.75pt]  [font=\footnotesize] [align=left] {$\dotsc $};
    \draw (275.77,78.13) node [anchor=north west][inner sep=0.75pt]  [font=\scriptsize] [align=left] {$(s,d_{1})$};
    \draw (274.97,119.13) node [anchor=north west][inner sep=0.75pt]  [font=\scriptsize] [align=left] {$(s,d_{2})$};
    \draw (373.9,90) node [anchor=north west][inner sep=0.75pt]  [font=\footnotesize] [align=left] {$\dotsc $};
    \draw (374.93,64.13) node [anchor=north west][inner sep=0.75pt]  [font=\scriptsize] [align=left] {$s'$};
    \draw (375.1,136) node [anchor=north west][inner sep=0.75pt]  [font=\footnotesize] [align=left] {$\dotsc $};
    \draw (375.6,105.73) node [anchor=north west][inner sep=0.75pt]  [font=\scriptsize] [align=left] {$s''$};
    \draw (317.79,67.57) node [anchor=north west][inner sep=0.75pt]  [font=\tiny,rotate=-351.2] [align=left] {$(d_{1} ,d_{1})$};
    \draw (314.54,108.33) node [anchor=north west][inner sep=0.75pt]  [font=\tiny,rotate=-351.92] [align=left] {$(d_{2} ,d_{2})$};
    \draw (233.78,86.54) node [anchor=north west][inner sep=0.75pt]  [font=\tiny,rotate=-339.05] [align=left] {$(0,0)$};
    \draw (235.84,114.28) node [anchor=north west][inner sep=0.75pt]  [font=\tiny,rotate=-17.09] [align=left] {$(0,0)$};
    \draw (320.23,89.2) node [anchor=north west][inner sep=0.75pt]  [font=\tiny,rotate=-5.43] [align=left] {$(d_{1} ,d_{1})$};
    \draw (331.29,120.3) node [anchor=north west][inner sep=0.75pt]  [font=\tiny,rotate=-9.12] [align=left] {$(d_{2} ,d_{2})$};
    \end{tikzpicture}
    \caption{Reduction from countdown games to the \prob{NCNS} problem with two players and bounded reachability objectives.}
    \label{fig:threshold-NCNS-2players}
\end{figure}

\begin{proof}
    For the \exptime{} upper bound, we only give a sketch of the proof. Given an objective $\reachBounded{d_i}{\reach{i}}$ for each player~$i \in \{1,\dots,t\}$ and $\reachBounded{c+1}{\reach{0}}$ for player~$0$, we can construct an adapted version of the B\"uchi tree automata from~\cite{ConduracheFGR16}. In short, their tree automaton works with qualitative reachability objectives. It guesses at each vertex $v \in V_i$ along a branch whether $v$ is in the winning region of player~$i$, and the move of a winning strategy from $v$. Then, it keeps track of deviations in case player~$i$ does not follow the guessed winning strategy. The objective of the tree automaton is a B\"uchi objective verifying whether every guess is done consistently with the infinite branch, i.e., when a player always followed a winning strategy, he actually wins; but it also verifies whether either player~$0$ visits $\reach{0}$ or a player~$i$ has a profitable deviation, i.e., player~$i$ saw the winning region, deviated and loses.

    The needed changes to shift to bounded reachability objectives are the extension of the tree automaton with accumulated weights, as it is done for the extended arena defined in \Cref{appendix:exp-algo-2p-ncns}, to keep track of players who are not able to satisfy their bounded reachability objectives because they exceed their bound $d_i$. Guesses of winning regions and winning strategies are done as before. Therefore, the constructed tree automaton has an exponential size, and deciding whether a B\"uchi tree automaton has an empty language can be done in time polynomial in this tree automaton. In other words, our algorithm is in \exptime{}.

    The lower bound is obtained by reduction from the countdown game problem to our problem with one player in the environment. Given a countdown game $\mathcal{CG}$ and a threshold $c$, we build a game $\game$ as illustrated in \Cref{fig:threshold-NCNS-2players} (where the grey part contains the countdown game), with two players~$0$ and~$1$ and where the objective of both players is $\reachBounded{c+1}{\{v_1,v_2\}}$. Player~$0$ owns the circle vertices and can decide to exit the grey part of the figure to go to the vertex $E$. Player~$1$ owns the square vertices, in particular $E$ from which he can go either to $v_1$ or to $v_2$. The weight function of each player is indicated on the edges. The initial vertex of $\game$ is the initial vertex of the countdown game.

    Suppose first that player~$0$ has a winning strategy $\strategyfor{0}$ in $\mathcal{CG}$. We can simulate it in $\game{}$ and make player~$0$ exit $\mathcal{CG}$ when his cost is exactly $c$. From $E$, if player~$1$ chooses to go to $v_1$, then he gets a cost of $c+1$, and thus his objective $\reachBounded{c+1}{\{v_1,v_2\}}$ is not satisfied. Therefore, every \fixedEquilibriaStrategy{} outcome $\pi$ necessarily goes to $v_2$, and thus we get $\cost{0}{\pi} = c$ showing that $\strategyfor{0}$ is solution to the modified \prob{NCNS} problem.

    Suppose now that player~$0$ has no winning strategy in the countdown game. Then, for every strategy $\strategyfor{0}$ in $\mathcal{CG}$, there is a play $\pi$ consistent with $\strategyfor{0}$ that never reaches a cost of exactly $c$ in $\mathcal{CG}$. Let us consider the possible cases in $\game$. If $\pi$ stays forever in $\mathcal{CG}$, we have $\cost{0}{\pi} > c$. If player~$0$ leaves $\pi$ to go to $E$, the current accumulated weight in $E$ is either $k < c$ or $k>c$.
    \begin{itemize}
    \item In $k < c$, going to $v_1$ or $v_2$ leads to two \fixedEquilibriaStrategy{} outcomes, as the objective of player~$1$ is satisfied in both outcomes. And when going to $v_1$, the cost for player~$0$ is $k+c+1 > c$.
    \item If $k > c$, we have again two \fixedEquilibriaStrategy{} outcomes when going to $v_1$ or $v_2$, but with the objective of player~$1$ being not satisfied. And for both outcomes, the cost for player~$0$ is $> c$. 
    \end{itemize}
    Hence there exists no solution to the modified \prob{NCNS} problem.
\end{proof}

\end{document}

%% file: macros.tex
\mathchardef\mhyphen="2D

\newcommand{\EE}{\emph{Eve}}
\renewcommand{\AA}{\emph{Adam}}
\newcommand{\A}{\text{\AA}}
\newcommand{\E}{\text{\EE}}

\newcommand{\playersquare}{\text{\scriptsize$\square$\normalsize}}
\newcommand{\playerdiamond}{\text{\large$\diamond$\normalsize}}

\newcommand{\Players}{\mathcal{P}}

\newcommand{\Visit}[1]{\mathsf{Visit}(#1)}

\newcommand{\N}{\mathbb{N}}
\newcommand{\Z}{\mathbb{Z}}

\newcommand{\Nzero}{\N^{>0}}

\newcommand{\ssetminus}{\! \setminus \!}

\newcommand{\Plays}{\mathsf{Plays}}
\newcommand{\Hist}{\mathsf{Hist}}
\newcommand{\occ}[1]{\mathsf{Occ}({#1})}

\newcommand{\arena}{\mathcal{A}}

\newcommand{\game}{\mathcal{G}}

\newcommand{\upperVal}[2]{\overline{\mathsf{Val}_{#1}}({#2})}
\newcommand{\lowerVal}[2]{\underline{\mathsf{Val}_{#1}}({#2})}

\newcommand{\ValPlayer}[2]{\mathsf{Val}_{#1}({#2})}
\newcommand{\ValStar}[1]{\mathsf{Val}^*({#1})}
\newcommand{\ValStarFunc}{$\mathsf{Val}^*$}

\newcommand{\payoff}[1]{\mathsf{cost}_{\textnormal{env}}({#1})}
\newcommand{\payoffprim}[1]{\mathsf{cost'}_{\textnormal{env}}({#1})}

\newcommand{\costfunc}[1]{\mathsf{cost}_{#1}}
\newcommand{\costfuncprim}[1]{\mathsf{cost'}_{#1}}
\newcommand{\cost}[2]{\costfunc{#1}(#2)}
\newcommand{\costprim}[2]{\costfuncprim{#1}(#2)}

\newcommand{\strategyfor}[1]{\sigma_{#1}}

\newcommand{\machine}[1]{\mathcal{M}_{#1}}
\newcommand{\strategiesmachine}[1]{\llbracket\mathcal{M}_{#1}\rrbracket}

\newcommand{\strategyProfile}{\sigma}
\newcommand{\outcome}[1]{\langle #1 \rangle}
\newcommand{\outcomefrom}[2]{\outcome{#1}_{#2}}

\newcommand{\Histsigma}[1]{\mathsf{Hist}_{#1}}

\newcommand{\reach}[1]{T_{#1}}

\newcommand{\reachBounded}[2]{\mathsf{Reach}_{< #1}(#2)}
 
\newcommand{\safetyBounded}[2]{\mathsf{Safe}_{\geq #1}(#2)}

\newcommand{\np}{$\mathsf{NP}$}
\newcommand{\npHard}{\np{}-hard}
\newcommand{\npComplete}{\np{}-complete}

\newcommand{\pspace}{$\mathsf{PSPACE}$}
\newcommand{\pspaceHard}{\pspace{}-hard}
\newcommand{\pspaceComplete}{\pspace{}-complete}
\newcommand{\nexptime}{$\mathsf{NEXPTIME}$}
\newcommand{\nexptimeHard}{\nexptime{}-hard}
\newcommand{\nexptimeComplete}{\nexptime{}-complete}
\newcommand{\exptime}{$\mathsf{EXPTIME}$}
\newcommand{\exptimeHard}{\exptime{}-hard}
\newcommand{\exptimeComplete}{\exptime{}-complete}

\newcommand{\conp}{$\mathsf{coNP}$}
\newcommand{\conpHard}{\conp{}-hard}
\newcommand{\conpComplete}{\conp{}-complete}

\newcommand{\sigmaClass}{$\Sigma_2^\mathsf{P}$}

\newcommand{\sigmaComplete}{\sigmaClass{}-complete}

\newcommand{\piClass}{$\Pi_2^\mathsf{P}$}
\newcommand{\piHard}{\piClass{}-hard}
\newcommand{\piComplete}{\piClass{}-complete}

\newcommand{\npConp}{$\mathsf{NP}^{\mathsf{coNP}}$}

\newcommand{\coQBF}{coQBF}
\newcommand{\QBF}{QBF}

\newcommand{\sigmaQBF}{$\Sigma_2$QBF}
\newcommand{\sigmaQBFneg}{$\Sigma_2$QBF$^{\mathsf{neg}}$}

\newcommand{\probname}[4]{%
        \ifx#1\empty \else Universal \fi%
        \ifx#2CCooperative \else Non-Cooperative \fi%
        \ifx#3PPareto \else\ifx#3NNash \else Subgame Perfect \fi\fi%
        \ifx#4SSynthesis\else Verification\fi
}
\newcommand{\prob}[1]{$\mathsf{#1}$}
\newcommand{\coprob}[1]{$\mathsf{co#1}$}

\newcommand{\fixed}{$0$-fixed}
\newcommand{\fixedEquilibria}{\fixed{} NE}
\newcommand{\fixedStrategy}{$\sigma_0$-fixed}
\newcommand{\fixedEquilibriaStrategy}{\fixedStrategy{} NE}

\newcommand{\last}[1]{\mathsf{last}({#1})}
\newcommand{\first}[1]{\mathsf{first}({#1})}
\newcommand{\succc}[1]{\mathsf{succ}({#1})}

%% file: tikz.tex
\usetikzlibrary{
  automata,
  arrows,
  positioning,
  shapes.geometric,
  patterns
}
\tikzset{
edge with arrows/.style = {
    ->,
    >=stealth,
    shorten >=1pt,
},
directed/.style = {
    edge with arrows,
    node distance=2.3cm,
    on grid,
    semithick,
    double distance=1.5pt,
},
automaton/.style = {
    directed,
    auto,
    initial text={},
    pin distance = 1ex,
    every pin edge/.style = {
        draw=none
    },
    every state/.style={
      minimum size=1.5mm,
      inner sep=1pt,
    }
},
system/.style = {
    state,
    circle,
    minimum size=0mm,
    inner sep=5pt,
},
system2/.style={
    state,
    ellipse,
    minimum size=0mm,
    inner sep=2pt,
},
environment/.style = {
    state,
    rectangle,
    minimum size=0mm,
    inner sep=8pt,
},
environment2/.style = {
    state,
    diamond,
    minimum size=0mm,
    inner sep=4pt,
},
}